\documentclass[12pt]{article}

\usepackage{amsmath, dsfont, amscd, latexsym, epsfig, color, times, ifthen, epstopdf}
\usepackage{amssymb, amsfonts}
\usepackage{amsthm}

\makeatletter
\renewcommand{\section}{\@startsection{section}{1}{0pt}{-12pt}{5pt}{\large\bf}}
\renewcommand{\subsection}{\@startsection{subsection}{2}{0pt}{-12pt}{5pt}{\normalsize\bf}}
\renewcommand{\subsubsection}{\@startsection{subsubsection}{3}{0pt}{-12pt}{5pt}{\normalsize\bf}}
\makeatother

\newtheorem{theorem}{Theorem}[section]

\theoremstyle{lemma}
\newtheorem{lemma}[theorem]{Lemma}

\theoremstyle{Remark}
\newtheorem{Remark}[theorem]{Remark}

\theoremstyle{fact}

\theoremstyle{definition}
\newtheorem{definition}[theorem]{Definition}

\theoremstyle{proposition}
\newtheorem{proposition}[theorem]{Proposition}

\theoremstyle{claim}
\newtheorem{claim}[theorem]{Claim}

\theoremstyle{conjecture}

\theoremstyle{corollary}

\newcommand{\eps}{\epsilon}    
\newcommand{\hd}{\Delta x}     
\newcommand{\rd}{\mathcal{RD}} 
\newcommand{\littlesum}{\mathop{\textstyle \sum}}
\newcommand{\littleprod}{\mathop{\textstyle \prod}}

\newcommand{\opt}{\mathrm{OPT}}   
\newcommand{\chd}{\mathrm{CHORD}} 

\newcommand{\comb}{\mathrm{Comb}} 

\newcommand{\I}{\mathcal{I}} 
\newcommand{\cp}{\mathrm{CP}} 
\newcommand{\p}{\mathrm{P}} 
\newcommand{\env}{\mathrm{LE}} 
\newcommand{\ch}{\mathcal{CH}} 

\newcommand{\poly}{\mathrm{poly}} 

\newcommand{\D}{\mathcal{D}} 
\newcommand{\E}{\mathbb{E}} 

\newcommand{\R}{\mathbb{R}} 
\newcommand{\Z}{\mathbb{Z}} 

\newcommand{\eqdef}{\stackrel{\textrm{def}}{=}}

\newcommand{\ignore}[1]{}

\def\squareforqed{\hbox{\(\blacksquare\)}}
\def\qed{\ifmmode\squareforqed\else{\unskip\nobreak\hfill\penalty50\hskip1em\null\nobreak\hfil\squareforqed\parfillskip=0pt\finalhyphendemerits=0\endgraf}\fi}

\long\def\symbolfootnote[#1]#2{\begingroup%
\def\thefootnote{\fnsymbol{footnote}}\footnote[#1]{#2}\endgroup}

\newcommand{\commented}{yes}

\ifthenelse{\equal{\commented}{yes}}{
\newcommand{\mnote}[1]{\footnote{{\bf [[Mihalis: {#1}\bf ]]  }}}
\newcommand{\inote}[1]{\footnote{{\bf [[Ilias: {#1}\bf ]]  }}}
\newcommand{\cnote}[1]{\footnote{{\bf [[Costas: {#1}\bf ]]  }}}
}

\ifthenelse{\equal{\commented}{no}}{
\newcommand{\inote}[1]{}
\newcommand{\mnote}[1]{}
\newcommand{\cnote}[1]{}
}

\begin{document}

\title{How good is the Chord algorithm?\footnote{An extended abstract of this work appeared in the \emph{Proceedings of the Twenty-First Annual ACM-SIAM Symposium on Discrete Algorithms (SODA 2010)}, pp. 978-991.}}


\author{    Constantinos Daskalakis\thanks{CSAIL, MIT. Email: {\tt costis@csail.mit.edu}.
            Work done while the author was a postdoctoral researcher at Microsoft Research New England.} \and
            Ilias Diakonikolas\thanks{School of Informatics, University of Edinburgh. 
            Email: {\tt ilias.d@ed.ac.uk}. Research partially supported by a Simons Postdoctoral Fellowship at UC Berkeley. 
            Most of this work was done while the author was at Columbia University, 
            supported by NSF grants CCF-04-30946, CCF-07-28736, and by an Alexander S. Onassis Foundation Fellowship.}\\
            \and
            Mihalis Yannakakis\thanks{Department of Computer Science, Columbia University.
            Email: {\tt mihalis@cs.columbia.edu}.
            Research supported by NSF grants CCF-07-28736, CCF-10-17955.}\\
}



\maketitle



\begin{abstract}
The Chord algorithm is a popular, simple  method for the succinct approximation of curves,
which is widely used, under different names, in a variety of areas, such as,
multiobjective and parametric optimization,  computational geometry, and graphics.
We analyze the performance of the Chord algorithm, as compared to the
optimal approximation that achieves a desired accuracy with the minimum number of points.
We prove sharp upper and lower bounds, both in the worst case and average case setting.
\end{abstract}



\section{Introduction} \label{sec:intro}

Consider a typical design problem with  more than one objectives
(design  criteria). For example, we may want to design a network that
provides the maximum capacity with the minimum cost,
or we may want to design a radiation therapy for a patient
that maximizes the dose to the tumor and minimizes the dose
to the healthy organs. In such {\em multiobjective} (or {\em multicriteria})
problems there is typically no solution that optimizes simultaneously all the
objectives, but rather a set of so-called {\em Pareto optimal}
solutions, i.e., solutions that are not dominated by any other solution in all the objectives.
The trade-off between the different objectives is
captured by the {\em trade-off} or {\em Pareto curve} (surface for three or more objectives),
the set of values of the objective functions for all the Pareto optimal solutions.

Multiobjective problems are prevalent in many fields, e.g., engineering, economics, management, healthcare, etc. There is
extensive research in this area published in different fields; see \cite{Ehr,EG,FGE,Miet} for some books and surveys. In a
multiobjective problem, we would ideally like to compute the Pareto curve and present it to the decision maker to select the
solution that strikes the `right' balance between the objectives according to his/her preferences (and different users may
prefer different points on the curve). The problem is that the Pareto curve has typically an enormous number of points, or even
an infinite number for continuous problems (with no closed form description), and thus we cannot compute all of them. We can
only compute a limited number of solutions (points), and of course we want the computed points to provide a good approximation
to the Pareto curve so that the decision maker can get a good sense of the range of possibilities in the design space.

We measure the quality of the approximation provided by a
computed solution set using a (multiplicative) approximation ratio,
as in the case of approximation algorithms for single-objective problems.
Assume (as usual in approximation algorithms) that the objective functions
take positive values. A set of  solutions $S$ is an {\em $\epsilon$-Pareto set}
if the members of $S$ approximately dominate within $1+\epsilon$ every other solution, i.e.,
for every solution $s$ there is a solution $s' \in S$ such that
$s'$ is within a factor $1+\epsilon$ or better of $s$ in all the objectives \cite{PY1}.
Often, after computing a finite set $S$ of solutions and their corresponding points in objective space (i.e., their vectors of
objective values), we `connect the dots', taking in effect also the convex combinations of the solution points. In problems
where the solution points form a convex set (examples include multiobjective flows, Linear Programming, Convex Programming),
this convexification is justified and provides a much better approximation of the Pareto curve than the original set $S$ of
individual points. A set $S$ of solutions is called an  {\em $\epsilon$-convex Pareto set} if the convex hull of the solution
points corresponding to $S$ approximately dominates within $1+\epsilon$ all the solution points \cite{DY2}. Even for
applications with nonconvex solution sets, sometimes solutions that are dominated by convex combinations of other solutions are
considered inferior, and one is interested only in solutions that are not thus dominated, i.e., in solutions whose objective
values are on the (undominated) boundary of the convex hull of all solution points, the so-called {\em convex Pareto set}. Note
that every instance of a multiobjective problem has a unique Pareto set and a unique convex Pareto set, but in general it has
many different $\epsilon$-Pareto sets and $\epsilon$-convex Pareto sets, and furthermore these can have drastically different
sizes. It is known that for every multiobjective problem with a fixed number $d$ of polynomially computable objective
functions, there exist an $\epsilon$-Pareto set (and also $\epsilon$-convex Pareto set) of polynomial size, in particular of
size $O(({m \over \epsilon})^{d-1})$, where $m$ is the bit complexity of the objective functions (i.e., the functions take
values in the range $[2^{-m},2^m]$) \cite{PY1}. Whether such approximate sets can be constructed in polynomial time is another
matter: necessary and sufficient conditions for polynomial time constructibility of  $\epsilon$-Pareto and $\epsilon$-convex
Pareto sets are given respectively in \cite{PY1,DY2}.

The most common approach to the generation of Pareto points (called weighted-sum method) is to give nonnegative weights $w_i$ to the
different objective functions $f_i$ (assume for simplicity that they are all minimization objectives) and then optimize the
linear combining function $\sum_i w_i f_i$; this approach assumes availability of a subroutine $\mathrm{Comb}$ that optimizes such
linear combinations of the objectives. For any set of nonnegative weights, the optimal solution is clearly in the Pareto set,
actually in the convex Pareto set. In fact, the convex Pareto set is precisely the set of optimal solutions for \emph{all}
possible such weighted linear combinations of the objectives. Of course, we cannot try all possible weights; we must select
carefully a finite set of weights, so that the resulting set of solutions provides a good approximation, i.e., is an
$\epsilon$-convex Pareto set for a desired small $\epsilon$. It is shown in \cite{DY2} that a necessary and sufficient
condition for the polynomial time constructibility of an $\epsilon$-convex Pareto set is the availability of a polynomial time
$\mathrm{Comb}$ routine for the approximate optimization of nonnegative linear combinations.

In a typical multiobjective problem,
the $\mathrm{Comb}$ routine is a nontrivial piece of software, each call takes a substantial
amount of time, thus we want to make the best use of the calls
to achieve as good a representation of the solution space as possible.
Ideally, we would like to achieve the smallest possible approximation error $\epsilon$
with the fewest number of calls to the $\mathrm{Comb}$ routine.
That is, given $\epsilon>0$, compute an $\epsilon$-convex Pareto set
for the instance at hand using \emph{as few Comb calls
as possible}. (In \cite{DY2} we study a different cost metric;
we explain the difference at the end of this section, and we note
that both metrics are relevant for different aspects of decision making
in multiobjective problems.)

We measure the performance of an algorithm by the ratio of its cost, i.e., the number of calls that it makes, to the minimum
possible number required for the instance at hand (as is usual in approximation algorithms). Let $\mathrm{OPT}_{\epsilon}(I)$
be the number of points in the smallest $\epsilon$-convex Pareto set for an instance $I$. Clearly, every algorithm that
computes an $\epsilon$-convex Pareto set needs to make at least $\mathrm{OPT}_{\epsilon}(I)$ calls. The {\em performance
(competitive) ratio} of an algorithm $\mathcal{A}$ that computes an $\epsilon$-convex Pareto set using $\mathcal{A}(I)$ calls
for each instance $I$ is $\sup_I \frac{{\mathcal{A}(I)}}{\mathrm{OPT}_{\epsilon}(I)}$. An important point that should be
stressed here is that, as in the case of online algorithms, the algorithm does not have complete information about the input,
i.e., the (convex) Pareto curve is {\em not} given explicitly, but can only be accessed indirectly through calls to the $\mathrm{Comb}$
routine; in fact, the whole purpose of the algorithm is to obtain an approximate knowledge of the curve.

In this paper we investigate the performance of the {\em Chord} algorithm, a simple, natural greedy algorithm for the
approximate construction of the Pareto set. The algorithm and variants of it have been used often, under various names, for
multiobjective problems \cite{AN,CCS,CHSB,FBR,RF,Ro,YG} as well as several other types of applications involving the
approximation of curves, which we will describe later on. We focus on the bi-objective case; although the algorithm can be
defined (and has been used) for more objectives, most of the literature concerns the bi-objective case, which is already rich
enough, and covers also most of the common uses of the algorithm.

We now briefly describe the algorithm. Let $f_1$, $f_2$ be the two objectives (say minimization objectives for concreteness),
and let $P$ be the (unknown) convex Pareto curve. First optimize $f_1$, and $f_2$ separately (i.e., call $\mathrm{Comb}$ for the weight
tuples $(1,0)$ and $(0,1)$) to compute the leftmost and rightmost points $a,b$ of the curve $P$. The segment $(a,b)$ is a first
approximation to $P$; its quality is determined by a point $q \in P$ that is least well covered by the segment. It is easy to
see that this worst point $q$ is a point of the Pareto curve $P$ that minimizes the linear combination $f_2 + \lambda_{ab} f_1$,
where $\lambda_{ab}$ is the absolute value of the slope of $(a,b)$, i.e., it is a point of $P$ with a supporting line parallel to the
`chord' $(a,b)$. Compute such a worst point $q$; if the error is at most $\epsilon$, then terminate, otherwise add $q$ to the set
$S$ to form an approximate set $\{a,q,b\}$ and recurse on the two intervals $(a,q)$ and $(q,b)$. In Section~\ref{sec:prelim} we
give a more detailed formal description (for example, in some cases one can determine from previous information that the maximum
possible error in an interval is upper bounded by $\epsilon$ and do not need to call $\mathrm{Comb}$).

The algorithm is quite natural, it has been often reinvented and is commonly used
for a number of other purposes.
As pointed out in \cite{Ro}, an early application was by Archimedes who used it to approximate
a parabola for area estimation \cite{Ar}.
In the area of {\em parametric optimization}, the algorithm is known as
the ``Eisner-Severance" method after \cite{ES}.
Note that parametric optimization is closely related to bi-objective optimization.
For example, in the parametric shortest path problem, each edge $e$ has
cost $c_e + \lambda d_e$ that depends on a parameter $\lambda$.
The length of the shortest path is a piecewise linear function of
$\lambda$ whose pieces correspond to the vertices of the convex Pareto
curve for the bi-objective shortest path problem with cost vectors $c,d$ on the
edges. A call to the $\mathrm{Comb}$ routine for the bi-objective problem corresponds
to solving the parametric problem for a particular value of the parameter.

The Chord algorithm is also useful for the approximation of convex functions, and for the approximation and smoothening of
convex and even non-convex curves. In the case of functions, an appropriate measure of distance between the function $f$ and the
approximation is the vertical distance, while for curves a natural measure is the Hausdorff distance; note that for a given
curve and approximating segment $ab$, the same point $q$ of the curve with supporting line parallel to $ab$ maximizes also the
above distances. In the context of smoothening and compressing curves and polygonal lines for graphics and related
applications, the Chord algorithm is known as the Ramer-Douglas-Peucker algorithm, after~\cite{Ra,DP} who independently
proposed it.

Previous work has analyzed the Chord algorithm (and variants) for achieving an $\epsilon$-approximation of a function or curve
with respect to vertical and Hausdorff distance, and proved bounds on the cost of the algorithm as a function of  $\epsilon$:
for all convex curves of length $L$, (under some technical conditions on the derivatives) the algorithm uses at most $O
(\sqrt{L/\epsilon})$ calls to construct an $\epsilon$-approximation ,
and there are curves (for example, a portion of a circle) that require
$\Omega (\sqrt{L / \epsilon})$ calls~\cite{Ro, YG}.

Note however that these results do not tell us what the performance ratio is, because for many instances, the optimal cost
$\mathrm{OPT}_{\epsilon}$ may be much smaller than $\sqrt{L/\epsilon}$, perhaps even a constant. For example, if $P$ is a
convex polygonal line with  few vertices, then the Chord algorithm will perform very well for $\eps=0$; in fact, as shown by
\cite{ES} in the context of parametric optimization, if there are $N$ breakpoints, then the algorithm will compute the exact
curve after $2N-1$ calls. (The problem is of course that in most bi-objective and parametric problems, the number $N$ of
vertices is huge, or even infinite for continuous problems, and thus we have to approximate.)

In this paper we provide sharp upper and lower bounds on the performance (competitive) ratio of the Chord algorithm, both in
the worst case and in the average case setting. Consider a bi-objective problem where the objective functions take values in
$[2^{-m},2^m]$. We prove that the worst-case performance ratio of the Chord algorithm for computing an $\epsilon$-convex Pareto
set is $\Theta (\frac{m + \log(1/\epsilon)}{\log m + \log \log (1/\epsilon)})$. The upper bound implies in particular that for
problems with polynomially computable objective functions and a polynomial-time (exact or approximate) Comb routine, the Chord
algorithm runs in polynomial time in the input size and $1/\epsilon$.
We show furthermore that there is no algorithm with constant performance ratio.
In particular, every algorithm (even randomized) has performance ratio at least
$\Omega (\log m + \log \log (1/\epsilon))$.

Similar results hold for the approximation of convex curves with respect to the Hausdorff distance. That is, the performance
ratio of the Chord algorithm for approximating a convex curve of length $L$ within Hausdorff distance $\epsilon$, is
$\Theta(\frac{\log(L/\epsilon)}{ \log \log (L/\epsilon)})$. Furthermore, every algorithm has worst-case performance ratio at
least $\Omega (\log \log (L/\epsilon))$.

We also analyze the expected performance of the Chord algorithm for some natural probability distributions. Given that the
algorithm is used in practice in various contexts with good performance, and since worst-case instances are often pathological
and extreme, it is interesting to analyze the average case performance of the algorithm. Indeed, we show that the performance
on the average is exponentially better. Note that Chord is a simple natural greedy algorithm, and is not tuned to any
particular distribution. We consider instances generated by a class of product distributions that are ``approximately'' uniform
and prove that the expected performance ratio of the Chord algorithm is $\Theta (\log m + \log \log (1/\eps))$ (upper and lower
bound). Again similar results hold for the Hausdorff distance.

\smallskip

\noindent {\bf Related Work.} There is extensive work on multiobjective optimization, as well as on approximation of curves in
various contexts. We have discussed already the main related references. The problem addressed by the Chord algorithm fits
within the general framework of determining the shape by probing \cite{CY}. Most of the work in this area concerns the exact
reconstruction, and the analytical works on approximation (e.g., \cite{LB,Ro,YG}) compute only the worst-case cost of the
algorithm in terms of $\epsilon$ (showing bounds of the form $O(\sqrt{L/\epsilon})$). There does  not seem to be any prior work
comparing the cost of the algorithm to the optimal cost for the instance at hand, i.e., the approximation ratio, which is the
usual way of measuring the performance of approximation algorithms.

The closest work in multiobjective optimization is our prior work \cite{DY2} on the approximation of convex Pareto curves using
a different cost metric. Both of the metrics are important and reflect different aspects of the use of the approximation in the
decision making process. Consider a problem, say with two objectives, suppose we make several calls, say $N$, to the $\mathrm{Comb}$
routine, compute a number of solution points, connect them and present the resulting curve to the decision maker to visualize
the range of possibilities, i.e., get an idea of the true convex Pareto curve. (The process may not end there, e.g., the
decision maker may narrow the range of interest, followed by computation of a better approximation for the narrower range, and
so forth). In this scenario, we want to achieve as small an error $\epsilon$ as possible, using as small a number $N$ of calls
as we can, ideally, as close as possible to the minimum number $\mathrm{OPT}_{\epsilon}(I)$ that is absolutely needed for the
instance. In this setting, the cost of the algorithm is measured by the number of calls (i.e., the computational effort); this
is the cost metric that we study in this paper, and the performance ratio is as usual the ratio of the cost to the optimum
cost. Consider now a scenario where the decision maker does not just inspect visually the curve, but will look more closely at
a set of solutions to select one; for instance a physician in the radiotherapy example will consider carefully a small number
of possible treatments in detail to decide which one to follow. Since human time is much more limited than computational time
(and more valuable, even small constant factors matter a lot), the primary metric in this scenario is the number $n$ of
selected solutions that is presented to the decision maker for closer investigation (we want $n$ to be as close as possible to
$\mathrm{OPT}_{\epsilon}(I)$), while the computational time, i.e., the number $N$ of calls, is less important and can be much larger (as
long as it is feasible of course). This second cost metric (the size $n$ of the selected set) is studied in \cite{DY2} for the
convex Pareto curve (and in \cite{VY,DY} in the nonconvex case). Among other results, it is shown there that for all
bi-objective problems with an exact $\mathrm{Comb}$ routine and a continuous convex space, an optimal $\epsilon$-convex Pareto set (i.e.,
one with $\mathrm{OPT}_{\epsilon}(I)$ solutions) can be computed in polynomial time using $O(m/\epsilon)$ calls to $\mathrm{Comb}$ in
general, (though more efficient algorithms are obtained for specific important problems such as bi-objective LP). For discrete
problems, a $2$-approximation to the minimum size can be obtained in polynomial time, and the factor $2$ is inherent. As remarked
above, both cost metrics are important for different stages of the decision making. Recall also that, as noted earlier, the
Chord algorithm runs in polynomial time, and furthermore, one can show that its solution set can be post-processed to select a subset that is a
$\epsilon$-convex Pareto set of size at most $2\mathrm{OPT}_{\epsilon}(I)$.

\bigskip

\noindent {\bf Structure of the paper.} The rest of the paper is organized as follows: Section~\ref{sec:prelim} describes the model and states our main results,
Section~\ref{sec:worst} concerns the worst-case analysis, and Section~\ref{sec:average} the average-case analysis. Section~\ref{sec:concl} concludes the paper
and suggests the most relevant future research directions.

\medskip








\section{Model and Statement of Results} \label{sec:prelim}

This section is structured as follows: After giving basic notation, in Section~\ref{ssec:defs} we describe the relevant definitions and framework from multiobjective optimization. In Section~\ref{ssec:chord-description} we provide a formal description of the Chord algorithm in tandem with an intuitive explanation.
Finally, in Section~\ref{ssec:results} we state our results on the performance of the algorithm as well as our general lower bounds.

\smallskip

\noindent {\bf Notation.} We start by introducing some notation used throughout the paper. For $n,i,j \in \Z_+$, we will denote $[n] :=
\{1,2,\ldots,n\}$ and $[i, j] := \{ i, i+1, \ldots, j \}$. For $p,q,r \in \R^2$, we denote by $pq$ the line segment with
endpoints $p$ and $q$, $(pq)$ denotes its length, $\triangle (p q r)$ is the triangle defined by $p, q, r$ and $\angle (pqr)$
is the internal angle of $\triangle (p q r)$ formed by $pq$ and $qr$.

We will use $x$ and $y$ as the two coordinates of the plane. If $p$ is a point on the plane, we use $x(p)$ and $y(p)$ to
denote its coordinates; that is, $p = \big(x(p),y(p)\big)$. We will typically use the symbol $\lambda$ to denote the
(absolute value of the) slope of a line, unless otherwise specified. Sometimes we will add an appropriate subscript, i.e., we will use
$\lambda_{pq}$ to denote the slope of the line defined by $p$ and $q$. For a Lebesgue measurable set $A \subseteq \R^2$ we will
denote its area by $S(A)$.

\subsection{Definitions and Background} \label{ssec:defs}

We describe the general framework of a bi-objective problem $\Pi$ to which our results are applicable. A bi-objective
optimization problem has a set of valid instances, and every instance has an associated set of feasible solutions, usually
called the solution or decision space. There are two objective functions, each of which maps every instance--solution pair to a
real number. The problem specifies for each objective whether it is to be maximized or minimized.

Consider the plane whose coordinates correspond to the two objectives. Every solution is mapped to a point on this plane. We
denote the objective functions by $x$ and $y$ and we use $\I$ to denote the objective space (i.e., the set of $2$-vectors of
objective values of the feasible solutions for the given instance). As usual in approximation, we assume that the objective
functions are polynomial time computable and take values in $[2^{-m}, 2^m]$, i.e., $\I \subseteq [2^{-m}, 2^m]^2$, where $m$ is
polynomially bounded in the size of the input. Note that this framework covers all discrete combinatorial optimization problems
of interest (e.g., shortest paths, spanning tree, matching, etc), but also contains many continuous problems (e.g., linear and
convex programs, etc). Throughout this paper we assume, for the sake of concreteness, that both objective functions $x$ and $y$
are to be minimized. All our results hold also for the case of maximization or mixed objectives.

\smallskip

Let $p,q \in \mathbb{R}^2_{+}$. We say that $p$ {\em dominates} $q$ if $p \leq q$ (coordinate-wise). We say that $p$ {\em $\eps$-covers}
$q$ ($\eps \ge 0$) if $p \le (1+\eps)q$. Let $A \subseteq \R^2_+$. The {\em Pareto set} of $A$, denoted by $\p(A)$, is the subset of
undominated points in $A$ (i.e., $p \in \p(A)$ iff $p \in A$ and no other point in $A$ dominates $p$). The convex Pareto set of
a $A$, denoted by $\cp(A)$, is the minimum subset of $A$ whose convex combinations dominate (every point in) $A$. We also use
the term {\em lower envelope} of $A$ to denote the Pareto set of its convex hull, i.e., $\env(A) = \p(\ch(A))$. In particular, if $A$
is convex its lower envelope is identified with its Pareto set. If $A$ is finite, its lower envelope is a convex polygonal
chain with vertices the points of $\cp(A)$.  Note that, for any set $A$, the lower envelope $\env(A)$ is a convex and monotone
decreasing planar curve. For $p,q \in \env(A)$ we will denote by $\env(pq)$ the subset of $\env(A)$ with endpoints $p$, $q$.

An {\em $\eps$-convex Pareto set} of $A$ (henceforth $\eps$-CP) is a subset $\cp_{\eps}(A)$ of $A$ whose convex combinations
$\eps$-cover (every point in) $A$. Note that such a set need not contain points dominated by convex combinations of other
points, as they are redundant. If a set contains no redundant points, we call it non-redundant. Let $S = \{s_i\}_{i=1}^k
\subseteq \mathbb{R}^2_{+}$, $x(s_i) < x(s_{i+1})$ and $y(s_i) > y(s_{i+1})$, be a non-redundant set. By definition, $S$ is an
$\eps$-CP for $A$ if and only if the polygonal chain $\env(S) = \langle s_1, \ldots, s_k \rangle$ $\eps$-covers $\cp(A)$. 

We define the {\em ratio distance} from a point $p$ to a point $q$ as $\mathcal{RD}(p,q) = \max \{ x(q)/x(p)-1, y(q)/y(p)-1, 0 \}$. 
(Note the asymmetry in the definition.)
Intuitively, it is the minimum value of $\eps \geq 0$ such that $q$ $\eps$-covers $p$. We also define the ratio distance between sets of points.
If $S_1, S_2 \subseteq \mathbb{R}^2_{+}$, then $\mathcal{RD}(S_1, S_2) = \sup_{q_1 \in S_1} \inf_{q_2 \in S_2}
\mathcal{RD}(q_1,q_2)$. As a corollary of this definition, the set $S\subseteq A$ is an $\eps$-CP for $A$ if and only if
$\rd(\env(A), \env(S)) = \rd(\cp(A), \env(S)) \le \eps$.

\smallskip

The above definitions apply to any set $A \subseteq \R^2_+$. Let $\Pi$ be a bi-objective optimization problem in the
aforementioned framework. For an instance of $\Pi$, the set $A$ corresponds to the objective space $\I$ (for the given
instance). We do not assume that the objective space $\I$ is convex; it may well be discrete or a continuous non-convex set. It
should be stressed that the objective space is not given explicitly, but rather implicitly through the instance. In particular,
we access the objective space $\I$ of $\Pi$ via an oracle $\comb$ that (exactly or approximately) minimizes non-negative linear
combinations $y+\lambda x$ of the objectives. That is, the oracle takes as input a parameter $\lambda \in \R_+$ and outputs a
point $q \in \I$ (i.e., a feasible point) that (exactly or approximately) minimizes the combined objective function
$h_{\lambda}(x,y) = y+ \lambda x$. We use the convention that, for $\lambda=+\infty$, the oracle minimizes the $x$ objective.

Formally, for $\lambda \in \mathbb{R}_+$, we denote by $\comb(\lambda)$ the problem of optimizing the combined objective
$h_{\lambda}(x,y)$ over $\I$. Let $\delta \in \R_+$ be an accuracy parameter. Then, for $\lambda \in \mathbb{R}_+$, we will
denote by $\comb_{\delta}(\lambda)$ a routine that returns a point $q \in \I$ that optimizes $h_{\lambda}$ up to a factor of
$(1+\delta)$, i.e., $h_{\lambda}(q) \leq (1+\delta) \cdot \min \{h_{\lambda}(p) \mid p \in \I \}$. In other words, the
$\comb_{\delta}$ routine is an ``approximate optimization oracle'' for the objective space $\I$. We say that the $\comb$
problem has a polynomial time approximation scheme (PTAS), if for any instance of $\Pi$ and any $\delta>0$ there exists a routine $\comb_{\delta}(\lambda)$ 
(as specified above) that runs in time polynomial in the size of the instance. As shown in~\cite{DY2}, for any bi-objective problem in the aforementioned
framework, there is a PTAS for constructing an $\eps$-convex Pareto set if and only if there is a PTAS for the $\comb$ problem.

\smallskip

We now provide a geometric characterization of $\comb_{\delta}(\lambda)$ that will be crucial throughout this paper.
Consider the point $q \in \I$ returned by $\comb_{\delta}(\lambda)$ and the corresponding line $\ell(q, \lambda)$ through $q$
with slope $-\lambda$, i.e., $\ell(q, \lambda) = \{(x,y) \in \mathbb{R}^2 ~\mid~ h_{\lambda}(x,y) = h_{\lambda}(q)\}$.  Then
there exists no solution point (i.e., no point in $\I$) below the line $(1+\delta)^{-1} \cdot \ell(q, \lambda) \eqdef \{ (x,y) \in
\mathbb{R}^2 ~\mid~ h_{\lambda}(x,y) = h_{\lambda}(q)/(1+\delta) \}$.  Geometrically, this means that we sweep a line of
absolute slope $\lambda$, 
until it touches (exactly or approximately) the undominated boundary (lower envelope) of the objective space $\I$. 
For $\delta=0$, the routine returns a point on the lower envelope $\env(\I)$, while for $\delta>0$
it returns a (potentially) dominated point of $\I$ ``close'' to the boundary (where the notion of ``closeness'' is
quantitatively defined by the aforementioned condition). See Figure~\ref{fig:comb} for an illustration.

\begin{figure} 
\begin{center}
\epsfig{file=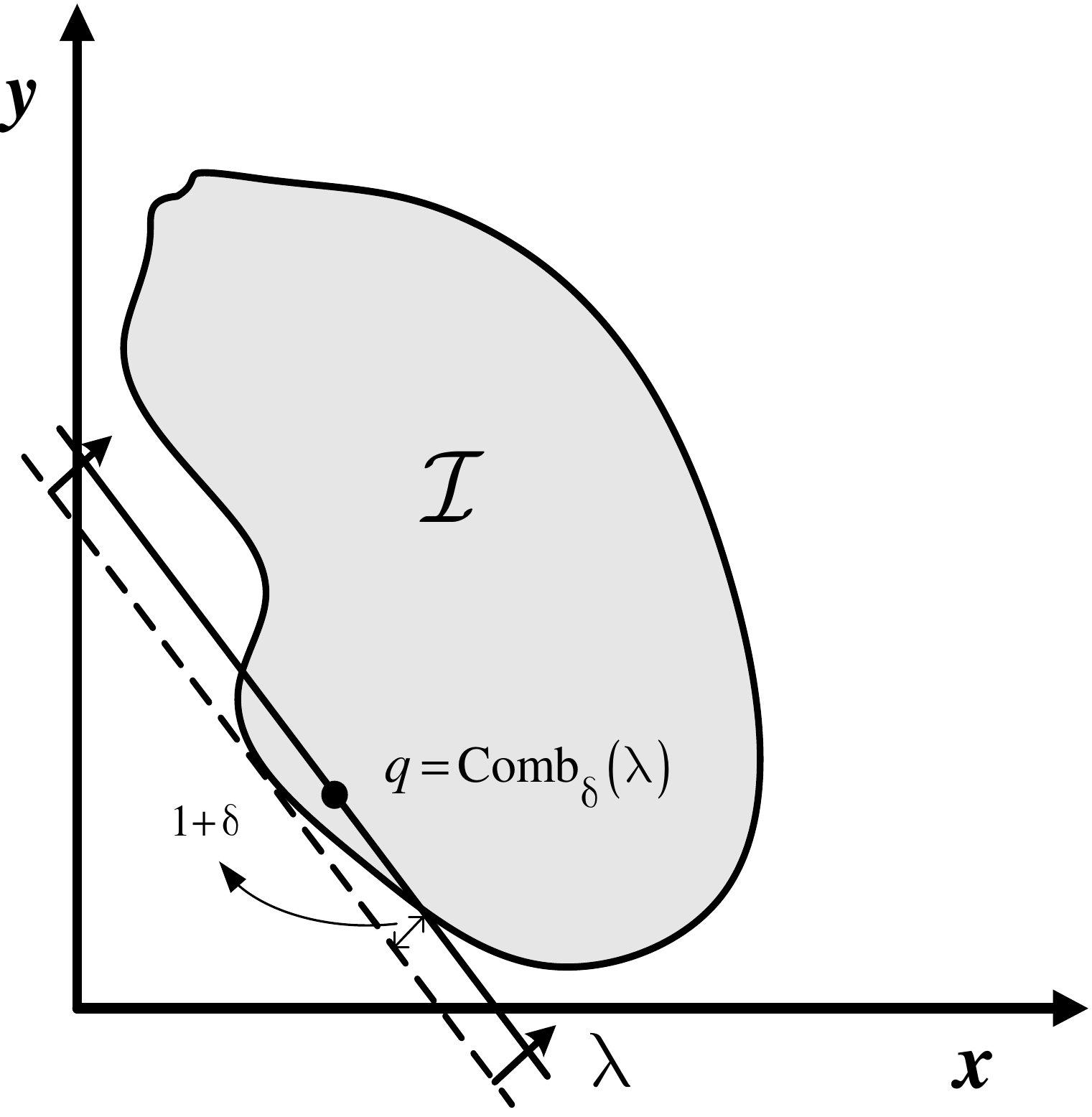,width=8cm}
\end{center}
\caption{A geometric interpretation of the $\comb_{\delta}(\lambda)$ routine. The shaded region represents the objective space
$\I$. The routine guarantees that there exist no solution points strictly below the dotted line.} \label{fig:comb}
\end{figure}

If $q$ is the (feasible) point in $\I$ returned by $\comb_{\delta}(\lambda)$, then we write $q = \comb_{\delta}(\lambda)$. We
assume that either $\delta=0$ (i.e., we have an exact routine), or we have a PTAS. For $\delta=0$, i.e., when the optimization is
exact, we omit the subscript and denote the Comb routine simply by $\comb(\lambda)$.

\smallskip

We will denote by $\opt_{\eps} (\I)$ the size of an optimum $\eps$-convex Pareto set for the given instance, i.e., an
$\eps$-convex Pareto set with the minimum number of points. Note that, obviously every algorithm that constructs an $\eps$-CP,
must certainly make at the very least $\opt_{\eps} (\I)$ calls to $\comb$, just to get $\opt_{\eps} (\I)$ points -- which are
needed at a minimum to form an $\eps$-CP; this holds even if the algorithm somehow manages to always be lucky and call $\comb$
with the right values of $\lambda$ that identify the points of an optimal $\eps$-CP. Having obtained the points of an optimal
$\eps$-CP, another $\opt_{\eps} (\I)$ many calls to Comb with the slopes of the edges of the polygonal line defined by the
points, suffice to verify that the points form an $\eps$-CP. Hence, the ``offline'' optimum number of calls is at most $2 \cdot
\opt_{\eps} (\I)$.

Let $\chd_\eps (\I)$ be the number of $\comb$ calls required by the Chord algorithm on instance $\I$. The {\em worst-case performance ratio} of the algorithm is defined to be $\sup_{\I} \frac{\chd_\eps(\I)}{\opt_\eps(\I)}$. If the inputs are drawn
from some probability distribution $\D$, then we will use the {\em expected performance ratio} $\E_{\I \sim \D} \left[
\frac{\chd_\eps(\I)}{\opt_\eps(\I)} \right]$ as a measure (note that we shall omit the subscript ``$\I \sim \D$'' when the
underlying distribution over instances will be clear from the context).

\medskip

While the main focus of this paper is on approximation of multiobjective optimization problems, 
our analysis also applies (with minor modifications) to related settings (in which the the Chord algorithm has been used). 
Consider for example the following classical {\em curve simplification} problem: 
Given a convex curve $C$ of length (at most $L)$ on the plane, find the minimum number of points on $C$ so that
the corresponding convex polygonal chain $C_{\eps}$ approximates the curve $C$ within distance error $\eps$.
A popular distance measure in this setting is the {\em Hausdorff distance} of $C_{\eps}$ from $C$, i.e.,
the maximum euclidean distance of a point in the actual curve from the approximating curve.
Note that the Hausdorff distance is invariant under translation, while our ratio distance is invariant under scaling.



We remark that our upper and lower bounds for the performance of the Chord algorithm wrt the ratio distance 
apply with minor modifications for the Hausdorff distance. This can be seen as follows:
A curve of length $L$ located anywhere on the plane can be scaled down by $L$ and translated to
the unit square $[1,2]$. By definition, approximating the original curve within Hausdorff distance $\eps$ 
is equivalent to approximating the new curve within Hausdorff distance error $\epsilon/L$. 
A simple calculation shows that for a convex curve in the unit square $[1, 2]$ the two metrics (Hausdorff and ratio distance) 
are within a constant factor of each other. Hence, the upper and lower bounds on the approximation wrt ratio distance give
immediately corresponding bounds wrt Hausdorff.

We also define the horizontal distance. We use this distance as an intermediate tool for our lower bound construction in
Section~\ref{ssec:worst-lower}. The {\em horizontal distance} from a point $p$ to 
a point $q$ is defined $\hd(p, q) = \max \{ x(q) - x(p), 0 \}$. The
horizontal distance from $p$ to a line $\ell$ (that is not horizontal) is $\hd(p, \ell) = \hd (p, p_\ell)$, where $p_\ell$ is the
$y$-projection of $p$ on $\ell$ (i.e., the point that a horizontal line from $p$
intersects $\ell$).


\smallskip

\begin{Remark}
\emph{All the upper bounds of this paper on the performance of the Chord algorithm hold under the assumption that we have a
PTAS for the Comb problem. On the other hand, our lower bounds apply even for the special case that an exact routine is
available. For the clarity of the exposition, we describe the Chord algorithm and prove our upper bounds for the case
of an exact Comb routine. We then describe the simple modifications in the algorithm and analysis for the case of an approximate routine.}
\end{Remark}


\subsection{The Chord Algorithm} \label{ssec:chord-description}

We have set the stage to formally describe the algorithm. Let $\Pi$ be a bi-objective problem with an efficient exact $\comb$
routine. Given $\eps>0$ and an instance $\I$ of $\Pi$ (implicitly via $\comb$), we would like to construct an $\eps$-CP for
$\I$ using as few calls to $\comb$ as possible. As mentioned in the introduction, a popular algorithm for this purpose is the
Chord algorithm that is the main object of study in this paper. In Table~\ref{table:chord} below we describe the algorithm in
detailed pseudo-code. The pseudo-code corresponds exactly to the description of the algorithm in the introduction.

The (recursively defined) routine Chord is called from the main algorithm and returns a set of points $Q \subseteq \I$ that is
an $\eps$-CP for $\I$. The recursive description of the algorithm is quite natural and will be useful in the analysis.

\vspace{-.4in}

\begin{table}[!h]

\begin{tabular}{p{3in}p{3.8in}}

\begin{flushleft}
\textbf{Chord Algorithm} (\emph{Input:} $\I$, $\eps$)\\

\medskip

$a =\comb (+\infty)$;\\
\smallskip
$b = \comb(0)$; \\
\smallskip
$c = (x(a), y(b))$; \\

\medskip

\textbf{Return} $Q = \textrm{Chord } (\{a, b, c \}, \eps)$.
\end{flushleft}

&

\begin{flushleft}
$\textbf{Routine } \textrm{Chord } (\textit{Input: } \{l, r, s \}, \eps)$\\
\smallskip
\textbf{If} $\rd(s, lr) \leq \eps$ \textbf{return} $\{ l, r \}$;\\
\smallskip
$\lambda_{lr} = \textrm{ absolute slope of } lr$;
$q = \comb (\lambda_{lr})$;\\
\textbf{If} $\rd (q, lr) \leq \eps$ \textbf{return} $\{ l, r \}$;\\
\smallskip
$\ell(q):=$ line parallel to $lr$ through $q$;\\
$s_l = ls \cap \ell(q)$; $s_r = rs \cap \ell(q)$;\\

\smallskip

$Q_l = \textrm{Chord}(\{l, q, s_l\}, \eps)$; $Q_r = \textrm{Chord}(\{q, r, s_r\}, \eps)$;\\

\smallskip
\textbf{Return} $Q_l \cup Q_r$.
\end{flushleft}

\end{tabular}


\caption{Pseudo-code for Chord algorithm.} \label{table:chord}
\end{table}


Let us proceed to explain the notation used in the pseudo-code in tandem with some intuitive understanding of the algorithm.
First, the feasible points $a, b \in \I$ minimize the $x$-objective and $y$-objective respectively. (Note that these points may
be dominated, i.e., are not necessarily the extreme points of $\cp(\I)$; however, this does not affect our analysis.) By
monotonicity and convexity, the lower envelope is contained in the right triangle $\triangle(acb)$, i.e., $\env(\I) \subseteq
\triangle(acb)$. (Note that the point $c$ is not a feasible point, but is solely defined for the purpose of ``sandwiching'' the
lower envelope.)

The Chord routine takes as input (i) The desired error tolerance $\eps$, and (ii) The ordered $3$-set of points $\{l,r,s\}$. In
every recursive call of the Chord routine, the points $l$ (left) and $r$ (right) are (feasible) points of the lower envelope,
i.e., $l, r \in \I \cap \env(\I)$. Moreover, the point $s$ is a (not necessarily feasible) point and the following conditions
are satisfied:
\begin{itemize}
\item The point $s$ lies to the right of $l$, to the left of $r$ and below the line segment $lr$.
In particular, this implies that the triangle $\triangle(lsr)$ is either right or obtuse, i.e., $\angle(lsr) \in [\pi/2, \pi)$.

\item The subset of the lower envelope (convex curve) with endpoints $l$ and $r$ is contained in $\triangle(lsr)$, i.e.,
$\env(lr) \subseteq \triangle(lsr)$.
\end{itemize}

See Figure~\ref{fig:chord} for an illustration. The red curve represents the lower envelope between the points $l$ and $r$,
i.e., the unknown curve we want to approximate. (Note that the feasible points $l, r$ are the results of previous recursive
calls.) The point $q = \comb(\lambda_{lr})$ is the feasible point in $\env(\I)$ computed in the current iteration (recursive
call). We remark that this point is at maximum ratio distance from the ``chord'' $lr$ -- among all points of $\env(lr)$. The
Chord routine will recurse on the triangles $\triangle(ls_lq)$ and $\triangle(qs_rr)$. Note that, by construction, the line
$s_ls_r$ is parallel to $lr$.

During the execution of the algorithm, we ``learn'' the objective space in an ``online'' fashion. After a number of iterations,
we have obtained information that imposes an ``upper'' and a ``lower'' approximation to $\env(\I)$. In particular, the computed
solution points define a polygonal chain that is an upper approximation to $\env(\I)$ and the supporting lines at these points
define a lower approximation. As the number of iterations increases, these bounds become more and more refined, hence we obtain
a better approximation to the curve.

Consider for example Figure~\ref{fig:chord}. Before the current iteration of the Chord routine, the only information available
to the algorithm is that the lower envelope (between points $l$ and $r$) lies in $\triangle(lsr)$, i.e., $lr$ is an upper
approximation and the polygonal chain $\langle l,s,r \rangle$ is a lower approximation. Given the information the algorithm has
at this stage, the potential error of this initial approximation is the ratio distance $\rd(s,lr)$. (If this distance is at
most $\eps$, then the segment $lr$ $\eps$-covers the subset of the lower envelope between $l,r$ and the routine terminates.
Otherwise, the point $q$ is computed and the approximation is refined, if necessary, as explained above.) After the current
iteration, the upper approximation is refined to be $\langle l,q,r \rangle$ and the lower approximation is $\langle l, s_l,
s_r, r \rangle$ . The potential error given the available information now is $\max \{\rd(s_l, lq), \rd(s_r, qr) \} <
\rd(s,lr)$. By applying these arguments recursively, we get that the Chord algorithm always terminates and, upon termination, it
outputs an $\eps$-CP for $\I$ (see Lemma~\ref{lem:chord-finds-eps-convex} for a rigorous proof).
%



\begin{figure}[h!]
\begin{center}
\epsfig{file=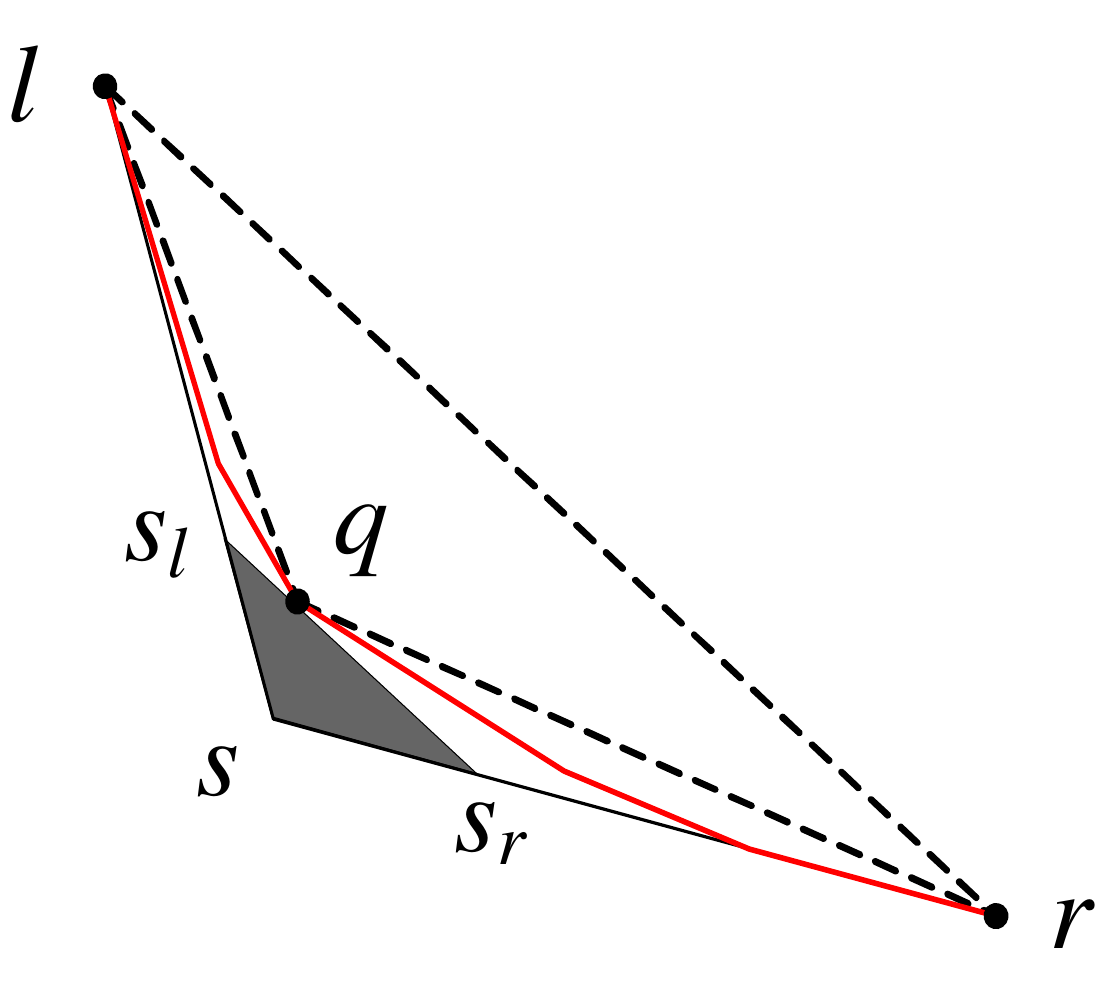,width=9cm}
\end{center}
\caption{Illustration of one iteration of the Chord routine.} \label{fig:chord}
\end{figure}


Consider the recursion tree built by the Chord algorithm. Every node of the tree corresponds to a triangle $\triangle(lsr)$
(input to the Chord routine at the corresponding recursive call). In the analysis, we shall use the following convention: There
is no node in the recursion tree if, at the corresponding step, the routine terminates without calling 
$\comb$ (i.e., if $\rd(s, lr) \leq \eps$ in the aforementioned description).

\medskip

The pseudo-code of Table~\ref{table:chord} is specialized for the ratio distance, but one may use other metrics based on the
application. In the context of convex curve simplification, our upper and lower bounds for the Chord algorithm also apply for
the Hausdorff distance (i.e. the maximum euclidean distance of a point in the actual curve from the approximate curve).


\subsection{Our Results}  \label{ssec:results}

We are now ready to state our main results. Our first main result is an  analysis of the Chord algorithm on worst-case
instances that is tight up to constant factors. In particular, for the ratio distance we prove

\begin{theorem} \label{thm:worst}
The worst-case performance ratio of the Chord algorithm (wrt the ratio distance) is $\Theta \big(
\frac{m+\log(1/\epsilon)}{\log m + \log\log(1/\epsilon)} \big).$ 
\end{theorem}

The lower bound on the performance of the Chord algorithm is proved in Section~\ref{ssec:worst-lower}. 
In Section~\ref{ssec:worst-lower-gen}, we also prove a general lower bound of  $\Omega (\log m + \log \log (1/\eps))$ 
on the performance ratio of {\em any} algorithm in the $\comb$ based model, as well as lower bounds for approximating
with respect to the horizontal (resp. vertical) distance.

In Section~\ref{ssec:worst-upper} we
give a proof of the upper bound. In Section~\ref{sssec:worst-upper-simple} we start by presenting the slightly weaker upper bound of $O(m + \log (1/\eps))$; 
this result has the advantage that its proof is simple and intuitive. The proof of the asymptotically tight upper bound requires a more careful analysis and
is presented in Section~\ref{sssec:worst-upper-final}.

\begin{Remark}
\emph{It turns out that the Hausdorff distance behaves very similarly to the ratio distance. In particular, by essentially
identical proofs, it follows that the performance ratio of the Chord algorithm for approximating a convex curve of length $L$
within Hausdorff distance $\epsilon$, is $\Theta(\frac{\log(L/\epsilon)}{ \log \log (L/\epsilon)})$. Furthermore, every
algorithm has worst-case performance ratio at least $\Omega (\log \log (L/\epsilon))$.}
\end{Remark}

In the process, we also analyze the Chord algorithm with respect to the horizontal distance metric (or by symmetry the vertical distance). We
show that in this setting the performance ratio of the algorithm is unbounded. In fact, we can get a strong lower bound in this
case: \emph{Any} algorithm with oracle access to $\comb$ has an unbounded performance ratio (Theorem~\ref{thm:hd-lb-gen-ws}), 
even on instances that lie in the unit square and the desired accuracy is a constant.

\medskip

Our second main result is an asymptotically tight analysis of the Chord algorithm in an average case setting (wrt the ratio
distance). Our random instances are drawn from 
standard distributions that have been widely used for the average
case analysis of geometric algorithms in a variety of settings. In particular, we consider (i) a Poisson Point Process on the
plane and (ii) $n$ points drawn from an ``un-concentrated'' product distribution. We now formally define these distributions.

\begin{definition} \label{def:PPP}
A (spatial, homogeneous) {\em Poisson Point Process (PPP)} of intensity $\nu$
on a bounded subset $\mathcal{S} \subseteq \mathbb{R}^2$ is a collection of random variables
$\{N(A)~|~A \subseteq \mathcal{S}~\text{is Lebesgue measurable}\}$
(representing the number of points occurring in every subset of $\mathcal{S}$), such that:
(i) for any Lebesgue measurable $A$,
$N(A)$ is a Poisson random variable with parameter $\nu \cdot S(A)$;
(ii) for any collection of \emph{disjoint} subsets
$A_1, \ldots, A_k$ the random variables $\{ N(A_i), i \in [k]  \}$ are mutually independent.
\end{definition}

\begin{definition} \label{def:delta-balanced distn}
Let $A$ be a bounded Lebesgue-measurable subset of  ${\mathbb R}^2$, and let $\cal D$ be a distribution over $A$. The
distribution $\cal D$ is called $\gamma$-{\em balanced}, $\gamma \in [0,1)$, if for all Lebesgue measurable subsets $A'
\subseteq A$, ${\cal D}(A') \in \left[ {(1-\gamma) \cdot {\cal U}(A')}, {{\cal U}(A') \over (1-\gamma)}\right]$, where $\cal U$
is the uniform distribution over $A$.
\end{definition}

We assume that $\gamma$ is an absolute constant and we omit the dependence on $\gamma$ in the performance ratio below. We
prove:

\begin{theorem}\label{thm:average}
For the aforementioned classes of random instances, the expected performance ratio of the Chord algorithm (wrt to the ratio distance) is $\Theta \big( \log
m + \log \log (1/\epsilon) \big)$.
\end{theorem}

The upper bound proof is given in Section~\ref{ssec:average-upper} and the lower bound one in Section~\ref{ssec:average-lower}. We
first present detailed proofs for the case of PPP and then present the (more involved) case of product distributions. (We note
that similar results apply also for approximation under the Hausdorff distance.)



\section{Worst--Case Analysis} \label{sec:worst}

\subsection{Lower Bounds} \label{ssec:worst-lower}
In Section~\ref{ssec:worst-lower-chord} we prove a tight lower bound
for the Chord algorithm for the ratio distance metric. In Section~\ref{ssec:worst-lower-gen} we show our general lower bounds, both
for the ratio distance and the horizontal distance.

\subsubsection{Lower Bound for Chord Algorithm} \label{ssec:worst-lower-chord}
Our main result in this section is a proof of the lower bound statement in Theorem~\ref{thm:worst}.
In fact, we show a stronger statement that also rules out the possibility of constant factor bi-criteria approximations, i.e., it 
applies even if the Chord algorithm is allowed (ratio distance) error $\Omega (\eps)$ and we compare it against the optimal $\eps$-CP set.

\begin{theorem} \label{thm:rd-lb-ws}
Let $\mu \geq 1$ be an absolute constant. Let $\eps>0$ be smaller than a sufficiently small constant and $m>0$ be large enough.
There exists an instance $\I_{LB} = \I_{LB} (\eps, m, \mu)$
such that $\opt_\eps (\I_{LB})  \leq 3$ and
$\chd_{\mu \cdot \eps} (\I_{LB}) = \Omega \Big( (1/\mu) \cdot \frac{m+\log(1/\epsilon)}{\log m + \log\log(1/\epsilon)} \Big)$.
\end{theorem}

\begin{proof}
The lower bound applies even if an exact $\comb$ routine is available, hence we restrict ourselves to this case.
Before we proceed with the formal proof, we give an explanation of our construction for the case $\mu=1$ and $m=1$. The rough
intuition is that the algorithm can perform poorly when the input instance is ``skewed'', i.e., we have a triangle $\triangle
(abc)$ where $(ac)>> (bc)$. For such instances one can force the algorithm to select many ``redundant'' points (hence, perform
many calls to $\comb$) to obtain a certificate it has found an $\eps$-CP set, even when few points (calls) suffice.

For the special case under consideration, the ``hard'' instance has endpoints $a = (1,2)$ and $b=(1+ 2 \eps, 1)$, where $\eps$ is
sufficiently small (to be specified later). Initially, the only available information to the algorithm is that the convex Pareto set for the given 
instance lies in the right triangle  $\triangle (acb)$, where $c = (1,1)$. Observe that, for an instance with these endpoints, the initial error $\rd(c, ab)$ is roughly
equal to $2\eps$ and one intermediate point $q^{\ast}$ together with the rightmost point of the curve always suffice 
to form an $\eps$-CP set, i.e., the optimal size is no more than $2.$ Our construction will define a sequence of points $\{q_1, \ldots, q_j \}$ (ordered in increasing
$x$-coordinate) which will form the Convex Pareto set for the corresponding instance and that force the Chord algorithm to select 
\emph{all} the $q_i$'s (in order of increasing $i$), until it finds $q^{\ast} = q_j$. That is, the Chord algorithm will 
monotonically converge to the optimal point by visiting all the vertices of the instance in order (in increasing $x$-coordinate).

Let $\lambda_{ab} = 1/(2\eps)$ be the slope of $ab$. The algorithm starts by calling $\comb(\lambda_{ab})$ to find a
solution point at maximum ratio distance from the chord $ab$. Our construction guarantees that $q_1 = \comb(\lambda_{ab})$.
That is, if $\ell(q_1)$ is the line parallel to $ab$ through $q_1$, $\ell(q_1)$ supports the objective space. 
The point $q_1$ is selected on the line segment $ac$ so that $(q_1c) = (ac)/k$, i.e., it is obtained by subdividing (the length of)
$ac$ geometrically with ratio $k$ -- for an appropriate $k$ (to be specified next).

Consider the point $q_1^{\ast} = \ell(q_1) \cap bc$. 
The error of the approximation $\{a, q_1, b \}$ equals $\rd(q_1^{\ast}, q_1b)$ (note that the error to the left of $q_1$ is $0$). 
If $k \leq 2$, we are already done, since $x(q_1^{\ast}) \geq
1+\eps$, which implies $\rd(q_1^{\ast}, q_1b) \leq \eps$. On the other hand, if $k \geq 1/\eps$, we are also done since $y(q_1)
\leq 1+\eps$, hence $\rd(q_1^{\ast}, q_1b) \leq \eps$. If $\omega(1) \leq k \leq o(1/\eps)$, it is not hard to show that
$$\rd(q_1^{\ast}, q_1b) \approx  \hd (q_1^{\ast}, q_1b) = (q_1^{\ast} b) = 2 \epsilon \cdot (1-1/k)$$
Hence, after the first call to $\comb$, the error has decreased by an additive of $2\eps/k << \eps$
and the algorithm will recurse on the triangle $\triangle(q_1q_1^{\ast}b)$. 

Note that $\lambda_{q_1b} = \lambda_{ab}/k = (2\eps)^{-1}/k$. 
The algorithm proceeds by calling $\comb(\lambda_{q_1b})$ and this call will return the point $q_2$. 
Let $\ell(q_2)$ be the (supporting) line parallel to $q_1b$ through $q_2$. Similarly, $\ell(q_2)$ supports the objective space.
The point  $q_2$ is selected by repeating our ``geometric subdivision trick.'' Recall that there are no feasible points below the line $q_1 q_1^{\ast}$. 
Let $q'_2$ be the projection of $q_2$ on $ac$. We select $q_2$ on the segment $q_1 q_1^{\ast}$ so that $(q'_2c) = (q_1c)/k$. 
The error of the approximation $\{a, q_1, q_2, b \}$ equals
$\rd(q_2^{\ast}, q_2b)$, where $q_2^{\ast} = \ell(q_2) \cap bc$. If $(q'_2c) = (ac)/k^2 >> \eps$, we have that
$$\rd(q_2^{\ast}, q_2b) \approx \hd (q_2^{\ast}, q_2b)  =  (q_2^{\ast} b) \approx 2\epsilon \cdot (1-2/k),$$ i.e., after the second step of the algorithm, 
the error has decreased by another additive $2\eps/k$ and the algorithm will recurse on the triangle $\triangle(q_2q_2^{\ast}b)$. 

We can repeat this process iteratively, where (roughly) in step $i$ we select $q_i$ on the line
$q_{i-1} q_{i-1}^{\ast}$, so that the length of the projection satisfies $(q_i'c) = (q_{i-1}'c)/k$. This iterative process can
continue as long as $(q_i'c)>>\eps$. Also note that the number $j$ of iterations cannot be more than $\approx k/2$ because
$x(q_i) \approx 1 + i \cdot (2\eps/k)$ and $x(q^{\ast}) \leq 1+\eps$. Since,  $(q_i'c)  = 1/k^{i}$ it turns out that the
optimal choice of parameters is $2 j \approx k \approx \frac{\log(1/\eps)}{\log \log (1/\eps)}$.

\smallskip
We stress that the actual construction is more elaborate than
the one presented in the intuitive explanation above.
Also, to show a bi-criterion lower bound,
we need to add one more point $q_{j+1}$ so that the Chord algorithm
selects $\{q_1, \ldots, q_j \}$ until it forms a $(\mu \cdot \epsilon)$-CP, 
while the point $q_{j+1}$ (along with the rightmost point $b$) suffice to form an $\eps$-CP.

\medskip

The formal proof comes in two steps. We first analyze the Chord algorithm with respect to the horizontal distance metric. We show that
the performance ratio of the algorithm is unbounded in this setting (this statement also holds for the vertical distance by
symmetry). In particular, for every $k \in \mathbb{N}$, there exists an instance $\I_G$ (lying entirely in the unit square) so that
the Chord algorithm (applied for additive error $1/2$) has performance ratio $k$ . We then show that, for an appropriate setting of the parameters in
$\I_{G}$, we obtain the instance $\I_{LB}$ that yields the desired lower bound with respect to the ratio distance.

\smallskip

\noindent \emph{\textbf{Step 1:}} The instance $\I_G (H, L, k, j)$ lies in the triangle $\triangle (abc)$, where $a = (1,1+H)$,
$b = (1+L, 1)$ and $c = (1,1)$. The points $a$ and $b$ are (the extreme) vertices of the convex Pareto set. We introduce two
additional parameters. The first one, $k \in \mathbb{N}$, is the ratio used in the construction to geometrically subdivide the
length of the line $ac$  in every iteration. The second one, $j \in \mathbb{N}$ with $j \in [1, k-1]$, is the number of
iterations and equals the number of vertices in the instance.

We define a set of points $Q = \{ q_i \}_{i=0}^{j+2}$ ordered in increasing $x$--coordinate and decreasing $y$--coordinate. Our instance will be the convex polygonal line with vertices the points of $Q$. We
set $q_0 =a$ and $q_{j+2} =b$. The set of points $\{ q_1, \ldots, q_{j+1}  \}$ is defined recursively as follows:
\begin{enumerate}
\item The point $q_1$ has $x(q_1) = x(a)$ and $y(q_1) = y(c) + \big( y(a)- y(c) \big) / k$.

\item For $i \in [2,j+1]$ the point $q_i$ is defined as follows:
      Let $\ell(q_{i-1})$ denote the line parallel to $q_{i-2}b$ through $q_{i-1}$.
      The point $q_i$ is the point of this line with $y(q_i) = y(c) +  \big( y(q_{i-1}) - y(c) \big)/(k+i-1)$.
\end{enumerate}

The algorithm is applied on this instance with desired horizontal distance error 
$$\epsilon_L(L, k, j) \eqdef L \cdot \frac{k-1}{k+j-1}.$$ 
Also denote $$\eps'_L(L,k,j) \eqdef L \cdot \frac{k-1}{k} \cdot
\frac{j}{k+j-1} = (j/k) \cdot \eps_L.$$
Note that $\eps'_L < \eps_L$.

See Figures~\ref{chord-fig-hd1} and~\ref{chord-fig-hd2} for a graphic illustration of the worst-case instances for the Chord
algorithm. We would like to stress that the figures are not drawn to scale. In particular, in the figures below we have $H=L$,
while the actual lower bound for the ratio distance applies for $H>>L$; in particular, for $H = 2^m$ and $L=O(\epsilon)$.

\begin{figure}[h!]
\begin{center}
\epsfig{file=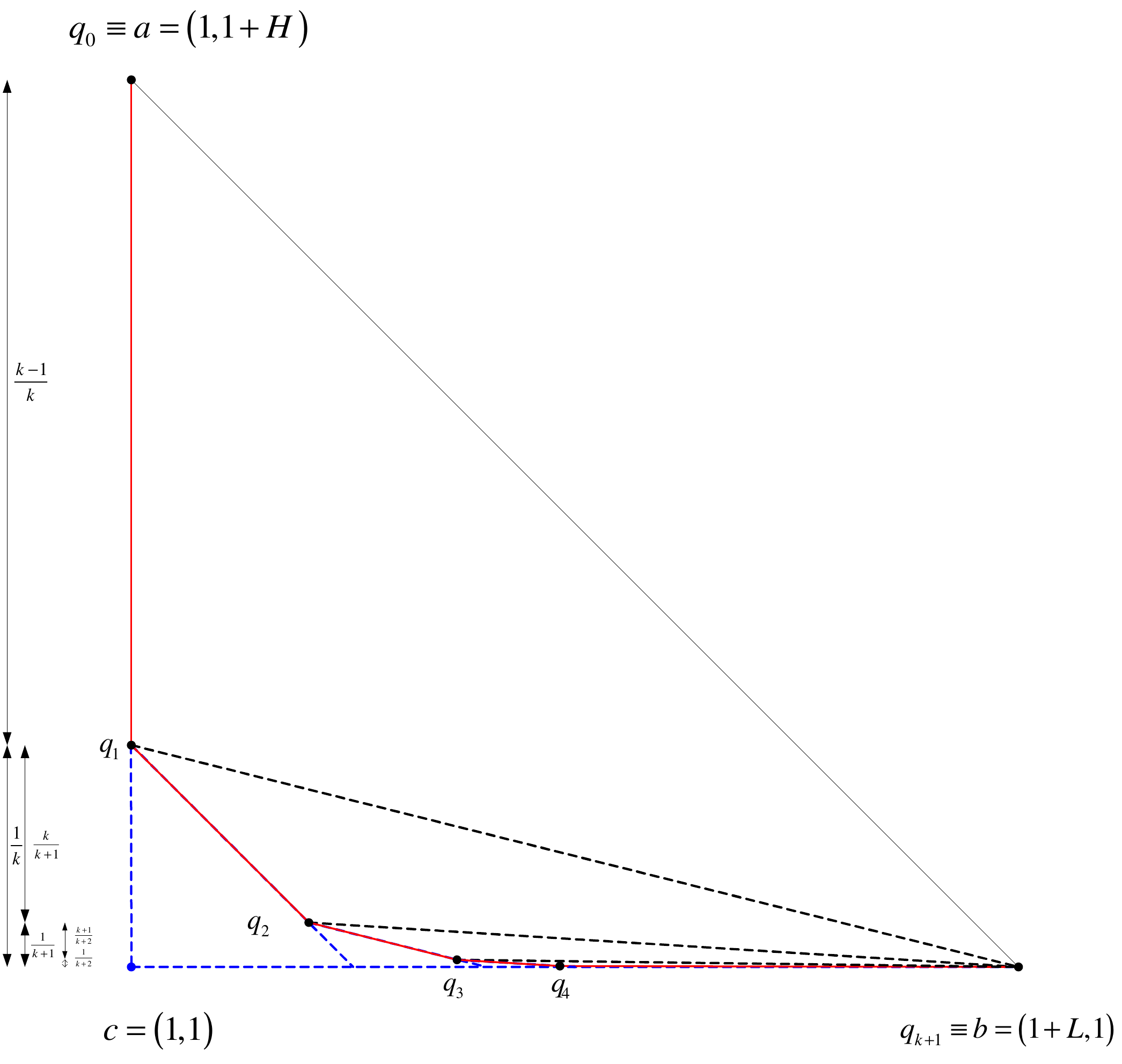,width=10cm}
\end{center}
\caption{Illustration of the lower bound. The figure depicts the case $j=k=4$. 
}
\label{chord-fig-hd1}
\end{figure}

\begin{figure}[h!]
\begin{center}
\epsfig{file=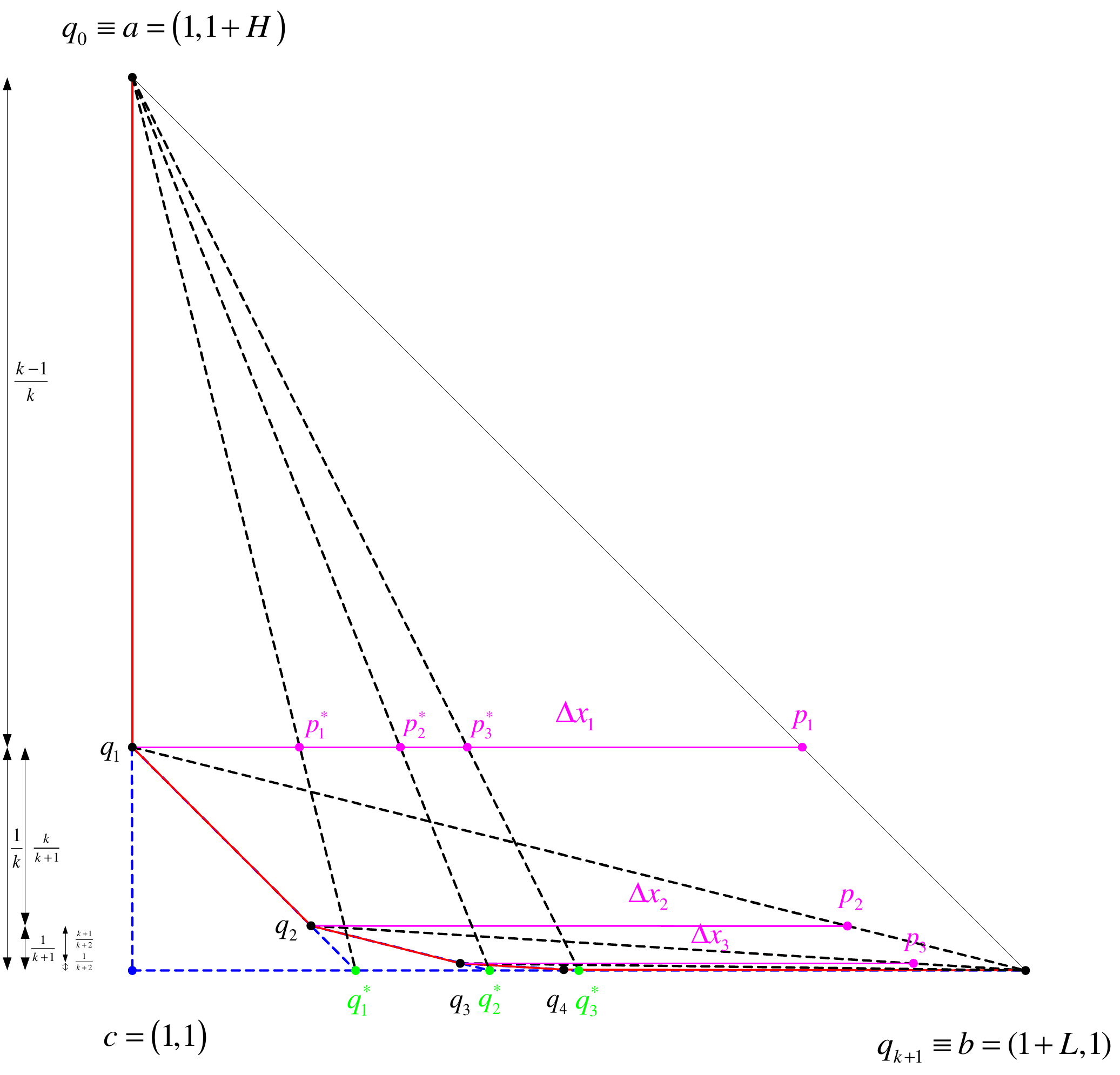,width=10cm}
\end{center}
\caption{Illustration of definitions for counterexample in Figure~\ref{chord-fig-hd1}.} \label{chord-fig-hd2}
\end{figure}


\noindent We show the following:
\begin{lemma} \label{claim:chd-hd}
The Chord algorithm applied to the instance $\I_G$ and error bound $\eps_L$ wrt horizontal distance
selects the sequence of points $\langle q_1, q_2, \ldots, q_j \rangle$, while the set $\{a, q_{j+1}, b\}$
attains error $\eps'_L < \eps_L$.
\end{lemma}
\begin{proof}
For $i \in [j-1]$, let $q_{i}^{\ast}$ be the intersection of the line $\ell
(q_i)$ -- the line parallel to $q_{i-1}b$ through $q_i$--with $bc$. The error of $\{a, q_{j+1}, b\}$ is exactly $\hd(q_1,
aq_{j+1})$. Observe that $\hd(q_1, aq_{j+1}) < \hd(q_1, aq_{j}^{\ast})$. By a simple geometric argument we obtain $\hd(q_1,
aq_{j}^{\ast}) = \eps'_L$, which yields the second statement. For the first statement, we show inductively that the recursion
tree built by the algorithm for $\I_G$ is a path of length $j-1$ and at depth $i-1$, for $i \in [j]$, the Chord subroutine
selects point $q_i$. The proof amounts to noting that the error of the approximation $\{q_1, \ldots, q_i\}$ is
$\hd(q_{i}^{\ast}, q_{i}b)$, which is $>\eps_L$ for $i<j$ and $=\eps_L$ for $i=j$.

To provide the details we need some notation. For $i \in [j]$, we denote $\Delta x_i \eqdef \hd(q_i, q_{i-1}b)$. Let $p_i$
be the $y$-projection of $q_i$ on $q_{i-1}b$, so that $\Delta x_{i} = (q_i p_i)$. Recall that $q_{i}^{\ast}$
denotes the intersection of the line $\ell (q_i)$  -- the line parallel to $q_{i-1}b$ through $q_i$--with $bc$. If $p_i^{\ast}$ is
the $y$-projection of $q_1$ on $aq_{i}^{\ast}$, we have $ \hd(q_1, aq_{i}^{\ast})= (q_1 p_i^{\ast})$. (See
Figure~\ref{chord-fig-hd2} for an illustration of these definitions.) We start with the following claim:


\begin{claim} \label{lem:hd-induction}
For all $i \in [j]$, we have $\Delta x_{i} = L \cdot (k-1)/(k+i-1)$.
\end{claim}

\begin{proof}
By induction on $i$. For the induction basis ($i=1$), we observe that the triangles $\triangle (a q_1 p_1)$ and $\triangle(a c
b)$ are similar, hence
\[ \frac{\Delta x_1}{(bc)} = \frac{(aq_1)}{(ac)} \]
which yields $$\Delta x_1 = (bc) \cdot (1-1/k) = L \cdot (k-1)/k$$ as desired.

Suppose that the claim is true for $i \in [j-1]$. We will prove it for $i+1$. We similarly exploit the similarity of the
triangles $\triangle (q_i q_i^{\ast} b)$ and $\triangle (q_i q_{i+1} p_{i+1})$, from which we get
\begin{equation}
\frac{\Delta x_{i+1}}{(q_i^{\ast} b)} = \frac{(q_iq_{i+1})}{(q_iq_i^{\ast})} = \frac{y(q_i) - y(q_{i+1})}{y(q_i)}
\label{eq:one}
\end{equation}
\noindent where the second equality follows from the collinearity of $q_i, q_{i+1}, q_i^{\ast}$. Observe that the third term in~(\ref{eq:one}) 
is equal to $(k+i-1)/(k+i)$ by construction. Now note that, because of the parallelogram
$(q_i p_i b q_i^{\ast})$, we have $(q_i^{\ast} b) = \Delta x_i$. Hence, (\ref{eq:one}) and the induction hypothesis imply
\[ \Delta x_{i+1} = \Delta x_{i} \cdot \frac{k+i-1}{k+i} = L \cdot \frac{k-1}{k+i-1} \cdot \frac{k+i-1}{k+i} =  L \cdot \frac{k-1}{k+i} \]
which completes the proof of the claim.
\end{proof}

Since $(q_i^{\ast}b) = \Delta x_{i}$ (as noted in the proof of Claim~\ref{lem:hd-induction}), it follows that for all $i \in [j-1]$ we have
\begin{equation}
x(q_i^{\ast}) = 1 + L \cdot \frac{i}{k+i-1} \label{eq:xistar}.
\end{equation}
We will start by showing that the set $\{a, q_{j+1}, b \}$ is an $\eps'_L$-convex Pareto under the horizontal distance. 
First, note that the error to the right of $q_{j+1}$, i.e., the distance of the lower envelope from $q_{j+1}b$ is in fact zero 
(since $q_{j+1} b$ is the rightmost edge of the lower envelope). It suffices to bound from above the error to its left. 
Since $aq_{j+1}$ has absolute slope larger than $ab$, the unique point (of the lower envelope) at maximum distance from $aq_{j+1}$ is $q_1$. 
Thus, we have that $\hd(q_1, aq_{j+1}) < \hd(q_1, aq_{j}^{\ast})$. From the similarity of the triangles $\triangle(caq_{j}^{\ast})$ and $\triangle (q_1 a p_{j}^{\ast})$
we get $$\hd(q_1, aq_{j}^{\ast}) = (1-1/k) \cdot (cq_{j}^{\ast}) = (1-1/k) \cdot  (x(q_{j}^{\ast})-x(c)) = L \cdot (1-1/k) \cdot
\frac{j}{k+j-1} = \eps'_L.$$ Hence, $\hd(q_1, aq_{j+1}) < \eps'_L$ as desired.

\smallskip

We now proceed to analyze the behavior of the Chord algorithm. We will show that the algorithm selects the points $q_1, q_2,
\ldots, q_j$ (in this order) till it terminates. Formally, we consider the recursion tree built by the algorithm for
the instance $\I_G$ and prove that it is a path of length $j-1$. In particular, for all $i \in [j]$, at depth $i-1$, the Chord
subroutine selects point $q_i$.

We prove the aforementioned statement by induction on the depth $d$ of the tree. Recall that the Chord algorithm initially finds
the extreme points $a$ and $b$. For $d=0$ (first recursive call), the algorithm selects a point of the lower envelope with
maximum horizontal distance from $ab$. By construction, all the points (of the lower envelope) in the line segment $q_1 q_2$
have the same (maximum) distance from $ab$ (since $q_1 q_2$ is parallel to $ab$). Hence, any of those points may be potentially
selected. Since $\comb$ is a black-box oracle, we may assume that indeed $q_1$ is
selected\symbolfootnote[2]{This simplifying assumption is only used for the sake of the exposition. We can slightly perturb the instance so that the
absolute slope of $q_1q_2$ is ``slightly'' smaller than $\lambda_{ab}$, so that the effect on the actual distances is
negligible. By doing so, the point $q_1$ will be the ``unique minimizer'' for $\comb(\lambda_{ab})$.}. The maximum
error after $q_1$ is selected equals $\hd(q_1^{\ast}, q_1b) = x(b) - x(q_1^{\ast}) = L \cdot (1-1/k) > \eps_L$. Hence, the
algorithm will not terminate after it has selected $q_1$.

For the inductive step, we assume that the recursion tree is a path up to depth $d \in [j-2]$ and the algorithm selected the
points $\{q_1, q_2, \ldots, q_{d+1} \}$ up to this depth. We analyze the algorithm at depth $d+1$. At depth $d+1$ the information available to the 
algorithm is that (i) the error to the left of $q_{d+1}$ is $0$ and (ii) the error to its right is $\hd(q_{d+1}^{\ast}, q_{d+1}b) = L \cdot
(k-1)/(k+d) > \eps_L$ (since $d \leq j-2$). Hence, the algorithm does not terminate and it calls $\comb$ to find a point between
$q_{d+1}$ and $b$ at maximum distance from $q_{d+1}b$. By construction, the points of maximum distance are those belonging to
the line $q_{d+2} q_{d+3}$ (which is parallel to $q_{d+1}b$); similarly, we can assume $q_{d+2}$ is selected. At this point of the execution 
the algorithm has the information that the error to the left of $q_{d+2}$ is $0$ and the error to its right is $\hd(q_{d+2}^{\ast}, q_{d+2}b) =
 L \cdot (k-1)/(k+d+1) > \epsilon_L$, unless $d = j-2$.
This completes the induction and the proof of Lemma~\ref{claim:chd-hd}. 
\end{proof}

\medskip

\noindent \textbf{\emph{Step 2:}} The instance $\I_{LB}$ is obtained from $\I_G(H, L, k, j)$ by appropriately setting the four relevant parameters.
In particular, 
\begin{enumerate}
\item Fix $H^{\ast} := 2^m-1$, $L^{\ast} := (\mu+1) \cdot \epsilon$,
$j^{\ast} := \Theta \big( (1/\mu) \cdot \frac{\log(H^{\ast}/\eps)}{\log \log (H^{\ast}/\eps)} \big)$
and $k^{\ast} := \mu \cdot j^{\ast} + 1.$
\item Set $\I_{LB} (\eps, m, \mu) := \I_G (H^{\ast}, L^{\ast}, k^{\ast}, j^{\ast}).$
\item Also, define ${\eps_L}^{\ast} := \eps_L(L^{\ast}, k^{\ast}, j^{\ast})$
and ${\eps'}_{L}^{\ast} := \eps'_L (L^{\ast}, k^{\ast}, j^{\ast})$.
\end{enumerate}
Observe that, under this choice of parameters,
we have ${\eps_L}^{\ast} \ge \mu \cdot \eps$ and ${\eps'_L}^{\ast} < \eps$.
Our main lemma for this step is the following:
\begin{lemma} \label{claim:chd-rd}
The Chord algorithm applied to $\I_{LB}$ and error bound ${\eps_L}^{\ast}$ 
wrt ratio distance selects (a superset of) the points
$\{ q_1, \ldots, q_{j^{\ast}/8} \}$,
while the set $\{a, q_{j^{\ast}+1}, b\}$ forms an ${\eps'_L}^{\ast}$-convex Pareto set.
\end{lemma}
\begin{proof}
The main idea of the proof is that for the particular instance under consideration,
the horizontal distance metric is a very good approximation to the ratio distance.
As a consequence, one can show that the behavior of the Chord algorithm in both metrics is similar.
The reason we ``lose'' a constant factor in the number of points
(i.e., the Chord algorithm selects $j^{\ast}/8$ points under the ratio distance as opposed to $j^{\ast}$ under the horizontal distance)
is due to the error term in the approximation between the metrics.
(The factor of $8$ is not important; any constant factor bigger than $1$ would suffice.)

We now proceed with the details. Consider the set $Q = \{ q_i\}_{i=0}^{j^{\ast}+2}$ defining the instance $\I_{LB}$. We will need 
the following lemma that quantifies the closeness of the two metrics in our setting. 
\begin{lemma} \label{lem:hd-vs-rd}
Let $a = (1, 1+H)$, $b = (1+L, 1)$ and $c = (1,1)$. 
Consider a point $s_1$ in $\triangle (abc)$ and let $\lambda$ be the absolute slope of $s_1b$. If $c'$ is the $x$-projection
of $s_1$ on $bc$, then for any point $s_2$ in $\triangle(s_1c'b)$ we have
\begin{equation} \label{eqn:hd-vs-rd}
\rd(s_2, s_1b) < \hd(s_2, s_1b) \leq \rd(s_2, s_1b)  + L^2 + L/\lambda.
\end{equation}
\end{lemma}
\begin{proof}
The proof, though elementary, requires some careful calculations.
Let $c' = (1+\delta, 1)$ be the $x$-projection of $s_1$ on the segment $bc$, that is $x(s_1) = 1+\delta$. 
Clearly, $0 \leq \delta \leq L.$
If $\lambda$ is the absolute slope of $s_1b$, we
have that $$y(s_1) = 1+\lambda \cdot (L-\delta).$$ Fix a point $s_2 = (1+\delta_x, 1+\delta_y) \in \triangle
(s_1c'b)$. We want to show that $\hd(s_2, s_1b)$, the horizontal distance of $s_2$ from $s_1b$, 
is a good approximation to the corresponding ratio distance $\rd(s_2, s_1b)$ when $\lambda$ is large. 
(See Figure~\ref{chord-fig-hrd} for an illustration.)
\begin{figure}[h!]
\begin{center}
\epsfig{file=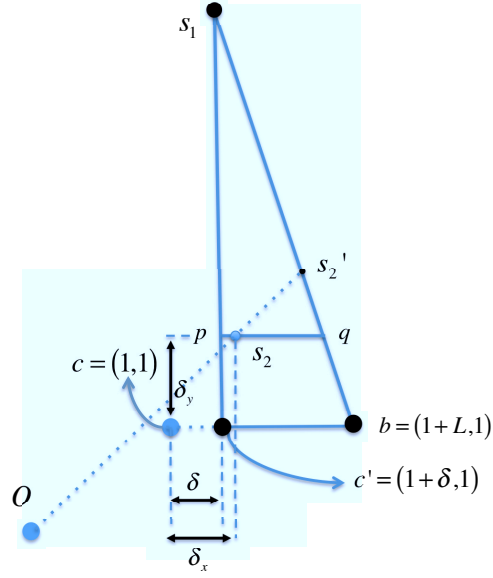,angle=90,width=15cm}
\end{center}
\vspace{-2cm}
\caption{Illustration of the relation between the horizontal and the ratio distance.} \label{chord-fig-hrd}
\end{figure}

We first calculate the horizontal distance $\hd(s_2, s_1b)$. Observe that
\[ \hd(s_2, s_1b) = (s_2q) = (pq) - (ps_2)\]
where $p$ is the $y$-projection of $s_2$ on $s_1c'$, i.e., $x(p) = x(c')$ and $y(p) = y(s_2).$
It is clear that $$(ps_2) = x(s_2) - x(p) = x(s_2) - x(c') =  \delta_x - \delta.$$ From the similarity of the triangles $\triangle (s_1 p q)$ and
$\triangle (s_1 c' b)$ it follows
\[  \frac{(pq)}{(c'b)} = \frac{(s_1p)}{(s_1c')}.\]
Since $(c'b) = L-\delta$, $(s_1 c') =\lambda \cdot (L-\delta) $ and $(s_1 p) = (s_1c') - (pc') = (s_1c') - \delta_y$ we get
\[ (pq) = (L-\delta) \cdot \Big( 1-\frac{\delta_y}{\lambda \cdot (L-\delta)} \Big).\]
Therefore,
\[ \hd(s_2, s_1b) = (L-\delta) \cdot \Big( 1-\frac{\delta_y}{\lambda \cdot (L-\delta)} \Big) - (\delta_x - \delta)
                      = L-\delta_x - \frac{\delta_y}{\lambda} . \]
The ratio distance $\mathbf{r} \eqdef \rd(s_2, s_1b)$ by definition is such that $s'_2 = (1+\mathbf{r}) \cdot s_2 \in s_1b$. We
thus get
\[ (1+\mathbf{r}) \cdot y(s_2) + \lambda (1+\mathbf{r}) \cdot x(s_2)  = y(b) + \lambda x(b) \]
or equivalently
\[  \mathbf{r}  = \frac{\lambda \big( x(b) - x(s_2)  \big)  - \big( y(s_2) - y(b) \big) }{y(s_2)+\lambda x(s_2)}. \]
\noindent Substitution yields that
\[ \mathbf{r} = \frac{\lambda (L-\delta_x) - \delta_y}{1+\delta_y + \lambda (1+\delta_x)}
     = \frac{L-\delta_x - \frac{\delta_y}{\lambda}}{1+\delta_x + \frac{1+\delta_y}{\lambda}} .\]
Observe that the numerator of the above fraction equals $\hd(s_2, s_1b)$ and the denominator is bigger than $1$. Hence, $\rd(s_2,
s_1b) < \hd(s_2, s_1b)$. For the other inequality we can write:
\begin{eqnarray}
\hd(s_2, s_1b) - \rd(s_2, s_1b) &=& \left( L -\delta_x -\frac{\delta_y}{\lambda} \right) \cdot
                                    \left( 1 - \frac{1}{1+\delta_x + \frac{1+\delta_y}{\lambda}} \right) \nonumber\\
     &=& \left( L -\delta_x -\frac{\delta_y}{\lambda} \right) \cdot
     \left( \frac{\delta_x + \frac{1}{\lambda} + \frac{\delta_y}{\lambda}}{1+\delta_x + \frac{1}{\lambda}+ \frac{\delta_y}{\lambda}} \right) \label{eq:prod}\\
     &\le& L \cdot (L + 1/\lambda) \label{eq:ineq}  = L^2 + L/\lambda
\end{eqnarray}
as desired. To obtain~(\ref{eq:ineq}), we bound each term in~(\ref{eq:prod}) separately. The first term is
clearly at most $L$. As for the second term, first note that the denominator is at least $1$. 
To bound the numerator from above, we claim that $\delta_x + \frac{\delta_y}{\lambda} \leq L$. To see this we use our assumption that 
$s_2$ lies in $\triangle(s_1c'b)$. In particular, this implies that $s_2$ lies below (or on) the line segment $s_1b$, i.e.,
$$y(s_2)+\lambda x(s_2) \leq y(b)+\lambda x(b)$$
which gives $$\delta_x + \frac{\delta_y}{\lambda} \leq \delta \leq L$$ as desired. 
This completes the proof of Lemma~\ref{lem:hd-vs-rd}.
\end{proof}

We may now proceed with the proof of Lemma~\ref{claim:chd-rd}.
By Lemma~\ref{claim:chd-hd}, the set $\{a, q_{j^{\ast}+1}, b\}$ attains horizontal distance error at most ${\eps'_L}^{\ast}$. By 
the first inequality of (\ref{eqn:hd-vs-rd}), the ratio distance error of $\{a, q_{j^{\ast}+1}, b\}$ is at most
as big, hence the second statement of Lemma~\ref{claim:chd-rd} follows.

\smallskip

By Lemma~\ref{claim:chd-hd}, the Chord algorithm under the horizontal distance metric selects all the points $\{q_1,
\ldots, q_{j^{\ast}} \}$ in order until it guarantees an $\eps_L^{\ast}$-approximation. In paricular, after the algorithm has selected
the subset $\{q_1, \ldots, q_i\}$, for $i \in [j^{\ast}]$, the horizontal approximation error is 
\begin{equation} \label{eqn:hd-error2}
\hd(q_i^{\ast}, q_ib) = L^{\ast} \cdot (k^{\ast}-1)/(k^{\ast}+i-1).
\end{equation}
We remark that the error term $E(L, \lambda) = L^2+L/\lambda$ in the RHS of~(\ref{eqn:hd-vs-rd}) leads to the ``constant factor loss'', i.e., the fact that the
Chord algorithm under the ratio distance picks $j^{\ast}/8$ (as opposed to $j^{\ast}$) points. (Also note that, since the ratio
distance is a lower bound for the horizontal distance, the Chord algorithm under the former metric will select at most $j^{\ast}$ points.)

Suppose we invoke the Chord algorithm with desired error of $0$, i.e., we want to reconstruct the lower envelope exactly. 
Then the algorithm will select the points $q_i$ in order of increasing $i$. It is also clear that the error of the approximation decreases
monotonically with the number of calls to $\comb$. Hence, to complete the proof, It suffices to show that after the algorithm has selected $\{q_1, \ldots,
q_{j^{\ast}/8} \}$, the ratio distance error will be bigger than $\eps^{\ast}_L$. To do this, we appeal to
Lemma~\ref{lem:hd-vs-rd}.

The ratio distance error of the set $\{q_1, \ldots, q_{j^{\ast}/8} \}$ is $\rd(q_{j^{\ast}/8}^{\ast},
q_{j^{\ast}/8}b)$. An application of (the second inequality of) (\ref{eqn:hd-vs-rd}) for $s_1 = q_{j^{\ast}/8}$, $s_2 =
q_{j^{\ast}/8}^{\ast}$ gives
\[ \rd(q_{j^{\ast}/8}^{\ast}, q_{j^{\ast}/8}b)
 \geq \hd(q_{j^{\ast}/8}^{\ast}, q_{j^{\ast}/8}b)  - E(L^{\ast}, \lambda_{q_{j^{\ast}/8}b}).\]
For the first term of the RHS, it follows from (\ref{eqn:hd-error2}) by substitution that
\[\hd(q_{j^{\ast}/8}^{\ast}, q_{j^{\ast}/8}b) = L^{\ast} \cdot \frac{\mu}{\mu+1/8}\]
and we similarly get $\eps_L^{\ast} =  L^{\ast} \cdot \frac{\mu}{\mu+1}.$
Hence, $$\hd(q_{j^{\ast}/8}^{\ast}, q_{j^{\ast}/8}b)  = \eps_L^{\ast}  + \Gamma,$$ 
where $\Gamma = L^{\ast} \cdot \Omega(1/\mu).$

Consider the error term $E(L^{\ast}, \lambda_{q_{j^{\ast}/8}b}) = (L^{\ast})^2 + L^{\ast} / \lambda_{q_{j^{\ast}/8} b}$. 
We will show that $E(L^{\ast}, \lambda_{q_{j^{\ast}/8}b}) < \Gamma$ which concludes the proof.
The first summand $(L^{\ast})^2 = (\mu+1)^2 \cdot \eps^2 $ is negligible compared to $ L^{\ast}/\mu$, as long as $\eps = o(1/\mu^2).$ 
To bound the second summand from above we need a lower bound on the slope $\lambda_{q_{j^{\ast}/8} b}$.

Recall (Equation~(\ref{eq:xistar}) in Lemma~\ref{lem:hd-induction}) that $x(q_i^{\ast}) = 1 + L^{\ast} \cdot i/(k^{\ast}+i-1)$. It is
also not hard to verify that $$\left| x(q_i)  - x(q_{i-1}^{\ast}) \right |  = O \left( L^{\ast}/(k^{\ast})^i \right).$$ Also recall that 
$$y(q_i) = 1+ \frac{H^{\ast}}{\littleprod_{j=1}^i (k^{\ast}+j-1)} > 1+\frac{H^{\ast}}{(2k^{\ast})^{i}}.$$ Hence, 
$$\lambda_{q_i b} = \frac{y(q_i) - y(b)}{x(b)-x(q_i)} >\left( \frac{H^{\ast}}{L^{\ast}} \right) \cdot (2k^{\ast})^{-i}.$$ 
It is straightforward to check that for the chosen values of
the parameters, $\lambda_{q_{j^{\ast}/8}b} >> \mu$, hence the second summand will also be significantly smaller than $L^{\ast}/\mu$.
In conclusion, $\rd(q_{j^{\ast}/8-1}^{\ast}, q_{j^{\ast}/8-1}b) > \eps_L^{\ast}$ and Lemma~\ref{claim:chd-rd} follows. 
\end{proof}
\noindent This also completes the proof of Theorem~\ref{thm:rd-lb-ws}. 
\end{proof}






\subsubsection{General Lower Bounds} \label{ssec:worst-lower-gen}
We can show an information--theoretic lower bound against any algorithm that uses $\comb$ as a black box
with respect to the ratio distance metric. Even though the bound 
we obtain is exponentially weaker than that attained by the Chord algorithm, 
it rules out the possibility of a constant factor approximation in this model. In particular, we show:
\begin{theorem} \label{thm:rd-lb-ws-gen}
Any algorithm (even randomized) with oracle access to a $\comb$ routine, 
has performance ratio $\Omega \left( \log m + \log \log (1/\eps) \right)$ with respect
to the ratio distance.
\end{theorem}
\begin{proof}
Let $\mathcal{A}$ be a general algorithm with oracle access to $\comb$.  The algorithm 
is given the desired error $\eps$, and it wants to compute an $\eps$-CP set. 
To do this, it queries the $\comb$ routine on a sequence of (absolute) slopes $\{\lambda_i\}_{i=1}^k$ and
terminates when it has obtained a  ``certificate'' that the set of points $\cup_{i=1}^k \comb(\lambda_i)$ forms an $\eps$-CP set.

The queries to $\comb$ can be adaptive, i.e., the $i$-th query $\lambda_i$ of the algorithm $\mathcal{A}$ 
can depend on the information (about the input instance) the algorithm has obtained from the previous queries $\lambda_1, \ldots, \lambda_{i-1}$.
On the other hand, the adversary also specifies the input instance adaptively, based on the queries made by the algorithm.

We will define a family of instances $\mathcal{Q}$ and an error $\eps>0$ with the following properties:
\begin{enumerate}
\item[(1)] Each instance $I \in \mathcal{Q}$ has $\opt_{\eps}(I) \leq 3.$
\item[(2)] In order for an algorithm $\mathcal{A}$ to have a certificate it found an $\eps$-CP set for an (adversarially chosen) instance $I \in \mathcal{Q}$, 
it needs to make at least $\Omega\left( \log m + \log \log (1/\eps) \right)$ calls to $\comb$.          
\end{enumerate}

Our construction uses the lower bound example for the Chord algorithm from Section~\ref{ssec:worst-lower-chord} essentially as a black box. 
Consider the instance $\I_{LB}(\eps, m, \mu:=1)$ (note that for $\mu=1$ we have that $\eps_L^{\ast} = \eps$ and ${\eps'}_L^{\ast} < \eps$)
and let $Q = \{q_i\}_{i=0}^{j+1}$ be the corresponding set of points, where $q_0 =a$ and $q_{j+1} = b$ and 
$j = j^{\ast}/8 = \Omega \left( \frac{m+\log(1/\epsilon)}{\log m + \log\log(1/\epsilon)} \right)$.
Our family of input instances $\mathcal{Q}$ consists of all ``prefixes'' of $Q$, 
i.e., $\mathcal{Q} = \{ I_i \}_{i=1}^{j}$, where $I_i = \{  a\equiv q_0, q_1, \ldots, q_i,  q_{j+1} \equiv b\}$, $i \in [j]$.

By the properties of $Q$, each $I_i$ defines a convex monotone decreasing polygonal curve, 
i.e., all its points are vertices of the corresponding convex Pareto. Moreover, the set $\{a, q_i, b\}$
is an $\eps_i$-CP set for $I_i$, where $\eps_i = \rd(q_1, aq_i)$. Note that $\eps_i < \eps_{i+1}$, $i \in [j]$, and that 
$$\eps_{j+1} <  \hd (q_1, aq_{j+1}) <  \hd (q_1, aq^{\ast}_j) \leq {\eps'}_L^{\ast} < \eps$$
where the first inequality holds from (\ref{eqn:hd-vs-rd}), and the rest follow by construction.
This proves  property $(1)$ above.

We now proceed to prove $(2)$. We claim that, given an arbitrary (unknown) instance $I \in \mathcal{Q}$, 
in order for an algorithm $\mathcal{A}$ to have a certificate it discovered an $\eps$-CP set,  
$\mathcal{A}$ must uniquely identify the instance $I$. In turn, this task
requires $\Omega (\log |\mathcal{Q}|) = \Omega (\log j)$ $\comb$ calls.

Indeed, consider an unknown instance $I \in \mathcal{Q}$ and let $q_i$ be the rightmost point of $I$ (excluding $b$) the algorithm $\mathcal{A}$ has discovered 
up to the current points of its execution. At this point, the information available to the algorithm is that $I \in \{I_{\ell}\}_{\ell \geq i}$.
Hence, the error the algorithm can guarantee for its current approximation is 
$$\rd(q_{i+1}, q_ib) \geq \hd (q_{i+1}, q_ib) -  E(L^{\ast}, \lambda_{q_{i}b}) >\eps_L^{\ast} \geq \eps.$$
The first inequality above follows from Lemma~\ref{lem:hd-vs-rd}. Also note that the term in the RHS 
is minimized for $i=j$ and (by the same analysis as in Lemma~\ref{claim:chd-rd}) 
the minimum value is bigger than $\eps_L^{\ast}$.

It remains to show that an adversary can force any algorithm to make at least $\Omega (\log |\mathcal{Q}|)$ queries to $\comb$
until it has identified an unknown instance of $\mathcal{Q}$. Clearly, identifying an unknown instance $I \in \mathcal{Q}$, i.e.,
finding the index $\ell$ such that $I = I_{\ell}$, is equivalent to identifying $q_{\ell}$ -- the rightmost point of $I$ to the left of $b$.
First, we can assume the algorithm is given the extreme points $a, b$ beforehand.  
A general algorithm $\mathcal{A}$ is allowed to query the $\comb$ routine for any slope $\lambda \in [0, +\infty)$.
Suppose that $I = I_{\ell}$. Then, for $\lambda \in \Lambda_i := [\lambda_{q_{i-1}q_i}, \lambda_{q_i q_{i+1}})$, the $\comb$ routine returns: (i) $q_i$ if $\ell \geq i$, 
(ii) $q_{\ell}$ if $\ell = i-1$, and (iii) $b$ if $\ell < i-1$.\symbolfootnote[2]{Strictly speaking, if $\ell \geq i$ and $\lambda= \lambda_{q_{i-1}q_i}$, the $\comb$ routine can return any point of the edge $q_{i-1} q_i$. However, 
if $\comb( \lambda_{q_{i-1}q_i})$ returns $q_i$ it can only help the algorithm (for our class of instances). Hence, we can make this assumption 
for the purposes of our lower bound.}. That is, the information obtained from a query $\lambda \in \Lambda_i$ is whether $\ell \geq i$, $\ell = i-1$ or $\ell < i-1$.
So, for our class of instances, a general deterministic algorithm $\mathcal{A}$ is equivalent to a ternary decision tree with the corresponding structure. The tree has $|\mathcal{Q}|$ many leaf nodes and there are at most $2^{d-1}$ leaves at depth $d$. Every internal node of the tree corresponds to a query to the $\comb$ routine, hence the depth measures the worst-case number of queries made by the algorithm. It is straightforward that any such tree has depth $\Omega (\log |\mathcal{Q}|)$ and the theorem follows. (If we allow randomization, the algorithm is a randomized decision tree. The expected depth of any such tree remains $\Omega (\log |\mathcal{Q}|)$ 
which yields the theorem for randomized algorithms as well.)
\end{proof}


We next show that, for the horizontal distance metric (or vertical distance by symmetry), any algorithm in our model
has an unbounded performance ratio, even for instances that lie in the unit square and for desired error $1/2$. 
\begin{theorem} \label{thm:hd-lb-gen-ws}
Any algorithm with oracle access to a $\comb$ routine has an unbounded performance ratio with respect to the horizontal distance, 
even on instances that lie in the unit square and for approximation error $1/2$.
\end{theorem}
\begin{proof}
Let $\mathcal{A}$ be an algorithm that, given the desired error $\eps$, computes an $\eps$-approximation with respect
to the horizontal distance. Fix $k \in \mathbb{Z}_+$. We show that an adaptive adversary can force the algorithm to make $\Omega(k)$ queries to $\comb$
even when $O(1)$ queries suffice. That is, the performance ratio of the algorithm is $\Omega(k)$; since $k$ can be arbitrary the theorem follows.

Our family of (adversarially constructed) instances all lie in the unit square and are obtained by a modification of the lower bound
instance $\I_{\mathcal{G}} (H:=1, L:=1, 2k, j:=k)$ for Chord under the horizontal distance (see Step 1 in Section~\ref{ssec:worst-lower-chord}). 
Since the horizontal distance is invariant under translation, we can shift our instances so that 
the points $a=(0,1)$ and $b=(1,0)$ are the leftmost and rightmost points of the convex Pareto set. After this shift, the point $c$ is identified with the origin.

Consider the initial triangle $\triangle(acb)$, where $a = (0,1)$, $b = (1,0)$ and $c =(0,0)$, and let $\lambda_1 >0$ be the first query made by the algorithm.
Given the value of $\lambda_1$, the adversary adds a vertex $q_1$ to the instance so that $q_1 = \comb(\lambda_1)$. The strategy
of the adversary is the following: The point $q_1$ belongs to $ac$, i.e., $x(q_1) =0$. 
If $\lambda_1 \geq \lambda_{ab} = 1$, then the point $q_1$ has $y(q_1) = 1/(2k).$ Otherwise, $y(q_1) = \lambda_1/(2k)$.
Let $\ell(q_1)$ be the line with slope $\lambda_1$ through $q_1$ and $q_1^{\ast}$ be its intersection with $cb$. The current information available
to the algorithm is that the point $q_1$ is a feasible point and that there are no points below the line $\ell(q_1)$.
Hence, the horizontal distance error it can certify is $\hd(q_1^{\ast}, q_1b) = (q_1^{\ast}b) = 1- (cq_1^{\ast}).$
On the other hand, if the true convex Pareto set was
$\cp_1 = \{a, q_1, q_1^{\ast}, b\}$, the set $ \{a, q_1^{\ast}, b\}$ would attain error $\hd(q_1, aq_1^{\ast}) < (cq_1^{\ast}).$
Note that $(cq_1^{\ast}) = y(q_1) / \lambda_1 \leq 1/(2k).$

After its first query, the algorithm knows that the error to the left of $q_1$ is $0$, hence it needs to focus on the triangle
$\triangle(q_1q_1^{\ast}b)$. Therefore, it is no loss of generality to assume that $\lambda_2 < \lambda_{q_1q_1^{\ast}} = \lambda_1$. 
Given $\lambda_2$, the adversary adds a vertex $q_2$ to the instance so that $q_2 = \comb(\lambda_2)$. 
Its strategy is to place $q_2$ on $q_1q_1^{\ast}$ and to set $y(q_2) = \lambda_{q_1b} / (2k)$ if $\lambda_2 \geq \lambda_{q_1b}$
and $y(q_2) = \lambda_2/(2k)$ otherwise. Let $\ell(q_2)$ be the line with slope $\lambda_2$ through $q_2$ and $q_2^{\ast}$ be its intersection with $cb$.
The information available to the algorithm after its second query is that the point $q_2$ is a feasible point and that there are no points below the line $\ell(q_2)$.
Hence, the horizontal distance error it can certify is $\hd(q_2^{\ast}, q_2b) = (q_2^{\ast}b) = 1- (cq_2^{\ast}).$
On the other hand, if the true convex Pareto set was
$\cp_2 = \{a, q_1, q_2, q_2^{\ast}, b\}$, the set $ \{a, q_2^{\ast}, b\}$ would attain error $\hd(q_1, aq_2^{\ast}) < (cq_2^{\ast}).$
Similarly, $(cq_2^{\ast}) = (cq_1^{\ast}) + (q_1^{\ast}q_2^{\ast}) \leq 1/(2k) + y(q_2) / \lambda_2 \leq 2/(2k).$
\begin{figure} 
\begin{center}
\epsfig{file=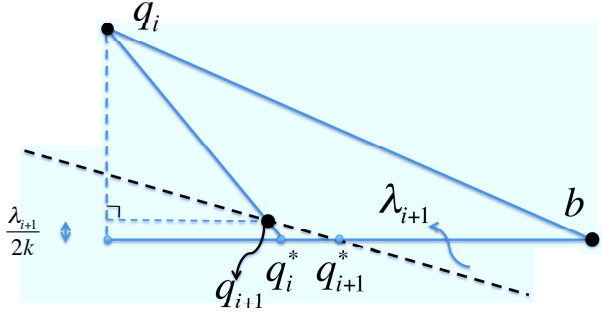, angle=90, width=12cm}
\end{center}
\vspace{-2cm}
\caption{Illustration of general lower bound for horizontal distance.} \label{fig:gen-lb-hd}
\end{figure}
The adversary argument continues by induction on $i$.
The induction hypothesis is the following: After its $i$-th query, the algorithm has computed a set of points $\{q_1, \ldots, q_i\}$, $q_j = \comb(\lambda_j)$, $j \in [i]$,
where $\lambda_j$  is the $j$-th query and $\lambda_j > \lambda_{j+1}$. 
(The $q_j$'s are vertices of the convex Pareto set of the corresponding instance ordered left to right). Let $\ell(q_i)$ be the line 
with slope $\lambda_i$ through $q_i$ and $q_i^{\ast}$ be its intersection with $cb$. Then $(cq_i^{\ast}) = x(q_i^{\ast}) \leq i/(2k).$

We now prove the induction step. We start by noting that the error to the left of $q_i$ is zero,
so the algorithm needs to focus on the triangle $\triangle(q_iq_i^{\ast}b)$; this implies that (without loss of generality) 
$\lambda_{i+1} < \lambda_{q_iq_i^{\ast}} = \lambda_i$.  Given $\lambda_{i+1}$, the adversary adds a vertex $q_{i+1}$ such that $q_{i+1} = \comb(\lambda_{i+1})$. 
Similarly, the strategy of the adversary is to place $q_{i+1}$ on $q_iq_i^{\ast}$ and to set $y(q_{i+1}) = \lambda_{q_ib} / (2k)$ if $\lambda_{i+1} \geq \lambda_{q_ib}$
and $y(q_{i+1}) = \lambda_{i+1}/(2k)$ otherwise. Let $\ell(q_{i+1})$ be the line with slope $\lambda_{i+1}$ through $q_{i+1}$ and $q_i^{\ast}$ be its intersection with $cb$.
The information available to the algorithm after its $(i+1)$-th query is that the point $q_{i+1}$ is a feasible point and that there are no points below the line $\ell(q_{i+1})$.
Hence, the horizontal distance error it can certify is $\hd(q_{i+1}^{\ast}, q_{i+1}b) = (q_{i+1}^{\ast}b) = 1- (cq_{i+1}^{\ast}).$
On the other hand, if the true convex Pareto set was
$\cp_{i+1} = \{a, q_1, q_2, \ldots, q_{i+1}, q_{i+1}^{\ast}, b\}$, the set $ \{a, q_{i+1}^{\ast}, b\}$ would attain error $\hd(q_1, aq_{i+1}^{\ast}) < (cq_{i+1}^{\ast})$.
Now note that $$(cq_{i+1}^{\ast}) = (cq_i^{\ast}) + (q_i^{\ast}q_{i+1}^{\ast}) \leq i/(2k) + y(q_{i+1}) / \lambda_{i+1} \leq (i+1)/(2k)$$
where the second inequality uses the inductive hypothesis and the definition of $q_{i+1}$.
(See Figure~\ref{fig:gen-lb-hd} for an illustration; the figure depicts the case that $\lambda_{i+1} < \lambda_{q_ib}$.)
This completes the induction.

The overall adversary argument is obtained from the above construction for $i=k-1$. That is, the adversarially constructed instance is the set 
$\cp_{k-1} = \{a, q_1, q_2, \ldots, q_{k-1}, q_{k-1}^{\ast}, b\}$, where $q_i = \comb(\lambda_i)$ (recall that $\lambda_i$ is the $i$-th query). 
The set $\{a, q_{k-1}^{\ast}, b\}$ has error $$\hd(q_1, aq_{k-1}^{\ast}) < (cq_{k-1}^{\ast}) \leq 1/2 - 1/(2k) < 1/2,$$
while the algorithm can certify error $$\hd(q_{k-1}^{\ast}, q_{k-1}b) = (q_{k-1}^{\ast}b) = 1- (cq_{k-1}^{\ast}) = 1/2+1/(2k) > 1/2.$$
Thus, after $k-1$ steps, the error of the algorithm remains more than $1/2$, while  $3$ queries suffice to attain error $<1/2$. 
This completes the proof of Theorem~\ref{thm:hd-lb-gen-ws}.
\end{proof}



\subsection{Upper Bound} \label{ssec:worst-upper}
In this section we establish the upper bound statement of Theorem~\ref{thm:worst}, i.e., we show 
that the performance ratio of the Chord algorithm for the ratio distance is $O \big(
\frac{m+\log(1/\epsilon)}{\log m + \log\log(1/\epsilon)} \big)$. For the sake of the exposition, we start by showing the
slightly weaker upper bound of $O(m + \log (1/\eps))$. The proof of the asymptotically tight upper bound is more involved and
builds on the understanding obtained from the simpler argument presented first.

\medskip

\subsubsection{An $O(m + \log (1/\eps))$ Upper Bound} \label{sssec:worst-upper-simple}

\smallskip

Before we proceed with the argument, some comments are in order. Perhaps the most natural approach to prove an upper bound
would be to argue that the error of the approximation constructed by the Chord algorithm decreases substantially (say by a
constant factor) in every iteration (subdivision) or after an appropriately defined ``epoch'' of a few iterations. This, if
true, would yield the desired result -- since the initial error cannot be more than $2^{O(m)}$. Unfortunately, such an approach
badly fails, as implied by the construction of Theorem~\ref{thm:rd-lb-ws}. Recall that, in the simplest setting of that
construction, the initial error is $2\eps$ and decreases by an additive $2\eps/k$ in every iteration, where $k  = \Omega \left( \log
(1/\eps)/ \log \log (1/\eps) \right)$. Hence, the error decreases by a sub-constant factor in every iteration. In fact, we note that
such an argument cannot hold for \emph{any} algorithm in our setting (i.e., given oracle access to $\comb$), as follows from our general lower bound
(Theorem~\ref{thm:rd-lb-ws-gen})\symbolfootnote[2]{In~\cite{RF} the authors -- in essentially the same model as ours -- propose
a variant of the Chord algorithm (that appropriately subdivides the current triangle into three sub-problems). They claim
(Lemma~3 in~\cite{RF}) that the error reduces by a factor of $2$ in every such subdivision. However, their proof is incorrect.
In fact, our counterexample from Section~\ref{ssec:worst-lower} implies the same lower bound for their proposed heuristic as
for the Chord algorithm.}.

Our approach is somewhat indirect. We prove that the \emph{area} between the (implicitly constructed) ``upper'' and ``lower'' approximation
decreases by a constant factor in every iteration of the algorithm (Lemma~\ref{lem:area-half}). This statement can be viewed as a ``potential
function type argument.'' To obtain an upper bound on the performance ratio, one additionally needs to relate the area to the
ratio distance. Indeed, we show that, when the area (between the upper approximation and the lower approximation) has
become ``small enough'' (roughly at most $\eps^2/2^{2m}$), the error of the approximation (ratio distance of the lower
approximation from the upper approximation) is at most $\eps$ (Lemma~\ref{lemma:small-area-means-done}). We combine the above with 
a simple charging argument (Lemma~\ref{lem:opt-lb}) to get the desired performance guarantee.
Formally, we prove:

\begin{theorem}\label{thm:ws-ub-area}
Let $T_1$ be the triangle at the root of the Chord algorithm's recursion tree,
let $S(T_1)$ be its area, and denote $\alpha \eqdef \min \{ x(q), y(q) \mid
q \in T_1 \}$. The algorithm finds an $\eps$-CP set after $O\left(\log \left(S(T_1)/S_0\right)\right) \cdot
\opt_{\eps}$ calls to $\comb$, where $S_0 \eqdef \eps^2 \cdot \alpha^2$.
\end{theorem}

\noindent The claimed upper bound on the performance ratio follows from the previous theorem, since $T_1 \subseteq [2^{-m},
2^m]^2$, which implies $S(T_1) \leq 2^{2m}$ and $\alpha \geq 2^{-m}$.

\smallskip

To prove Theorem~\ref{thm:ws-ub-area} we will need a few lemmas. 
We start by proving correctness, i.e., we show that, upon termination, the algorithm computes an $\eps$-CP set for $\I$. This
statement may be quite intuitive, but it requires a proof. The following claim formally states some basic
properties of the algorithm:

\begin{claim} \label{clm:sandwich}
Let $T_i = \triangle(a_ib_ic_i)$ be the triangle processed by the Chord algorithm at some recursive step. Then the following
conditions are satisfied: (i) $a_i, b_i \in \I \cap \env(\I)$ ; in particular, $a_i = \comb(\lambda_{a_ic_i})$ and $b_i =
\comb(\lambda_{b_ic_i})$, (ii) $x(a_i) \leq x(c_i) \leq x(b_i)$, $y(a_i) \geq y(c_i) \geq y(b_i)$ and $c_i$ lies below the line
$a_ib_i$, and (iii) $\env(a_ib_i) \subseteq T_i$.
\end{claim}
That is, whenever the Chord routine is called with parameters $\{l,r,s\}$ (see Table~\ref{table:chord}), the points $l$ and $r$
are feasible solutions of the lower envelope and the segment of the lower envelope between them is entirely contained in the
triangle $\triangle(lsr)$.
\begin{proof}
Consider the node (corresponding to) $T_i$ in the recursion tree built by the algorithm. We will prove the claim by induction
on the depth of the node. The base case corresponds to either an empty tree or a single node (root). The Chord routine is
initially called for the triangle $T_1 = \triangle (abc)$. Conditions (i) and (ii) are thus clearly satisfied. It follows from
the definition of the $\comb$ routine that there exist no solution points strictly to the left of $a$ and strictly below $b$.
Hence, by monotonicity and convexity, we have $\env(\I) = \env(ab) \subseteq \triangle(acb)$, i.e., condition (iii) is also
satisfied. This establishes the base case.

For the induction step, suppose the claim holds true for every node up to depth $d \geq 0$. We will prove it for any node of
depth $d+1$. Indeed, let $T$ be a depth $d+1$ node and let $T_i = \triangle(a_ib_ic_i)$ be $T$'s parent node in the recursion
tree. By the induction hypothesis, we have that (i) $a_i, b_i \in \I \cap \env(\I)$; $a_i = \comb(\lambda_{a_ic_i})$ and $b_i =
\comb(\lambda_{b_ic_i})$ (ii) $x(a_i) \leq x(c_i) \leq x(b_i)$, $y(a_i) \geq y(c_i) \geq y(b_i)$ and $c_i$ lies below the line
$a_ib_i$  and (iii) $\env(a_ib_i) \subseteq T_i$. We want to show the analogous properties for $T$.

We can assume without loss of generality that $T$ is $T_i$'s left child (the other case being symmetric). Then, it follows by
construction (see Table~\ref{table:chord}) that $T = \triangle (a_i q_i a'_i)$ where $q_i = \comb(\lambda_{a_ib_i})$. We
claim that $q_i \in T_i$. Indeed, note that $\lambda_{a_ic_i} > \lambda_{a_ib_i} > \lambda_{b_ic_i}$ (as follows from property
(ii) of the induction hypothesis). By monotonicity and convexity of the lower envelope, combined with property (iii) of the
induction hypothesis, the claim follows. Now note that $a'_i \in a_ic_i$ and $a'_iq_i \parallel a_ib_i$. Hence, property (i) of
the inductive step is satisfied. Since $a_ic_i$ has non-positive slope (as follows from (ii) of the inductive hypothesis),
$a_i$ lies to the left and above $a'_i$ ; similarly, since $a_ib_i$ has negative slope, $a'_i$ lies to the left and above
$q_i$. We also have that $\lambda_{a_ic_i} = \lambda_{a_ia'_i} \geq \lambda_{a_iq_i} \geq \lambda_{a'_iq_i} =
\lambda_{a_ib_i}$, where the first inequality follows from the fact that $q_i \in T_i$. Property (ii) follows from the
aforementioned. By definition of the Chord routine, we have $a'_iq_i \parallel a_ib_i$ and there are no solution points below $a'_iq_i$.
By convexity, we get that $\env(a_iq_i)$ lies below $a_iq_i$. Hence, property (iii) also follows. This proves the induction and the claim.
\end{proof}
\noindent By exploiting the above claim, we can prove correctness:
\begin{lemma} \label{lem:chord-finds-eps-convex}
The set of points $Q$ computed by the Chord algorithm is an $\eps$-CP set.
\end{lemma}
\begin{proof}
Let $Q = \{ a=:q_0, q_1, q_2, \ldots, q_r, q_{r+1}:=b \}$ be the set of feasible points in $\env(\I)$ output by the algorithm,
where the points of $Q$ are ordered in increasing order of their $x$-coordinate (decreasing order of their $y$-coordinate).
Note that all the $q_i$'s are in convex position. We have that $\rd(\env(\I), \langle q_0, \ldots,
q_{r+1} \rangle) = \max_{i=0}^r \rd(\env(q_iq_{i+1}), q_iq_{i+1})$. So, it suffices to show that, for all $i$,
$\rd(\env(q_iq_{i+1}), q_iq_{i+1}) \leq \eps$.

Since the algorithm terminates with the set $Q$, it follows that, for all $i$, the Chord routine was called by the algorithm
for the adjacent feasible points $\{q_i, q_{i+1}\}$ and returned without adding a new feasible point between them. Let $c_i$ be
the corresponding third point (argument to the Chord routine). Then, by Claim~\ref{clm:sandwich}, we have that $c_i$ is between
$q_i$ and $q_{i+1}$ in both coordinates and below the segment $q_iq_{i+1}$ and moreover $\env(q_iq_{i+1}) \subseteq \triangle
(q_iq_{i+1}c_i)$. Since the Chord routine returns without adding a new point on input $(\{q_i, q_{i+1}, c_i\}, \eps)$, it
follows that either $\rd(c_i, q_iq_{i+1}) \leq \eps$ or the point $q' = \comb(\lambda_{q_iq_{i+1}})$ satisfies $\rd(q',
q_iq_{i+1}) \leq \eps$. In the former case, since $\env(q_iq_{i+1}) \subseteq \triangle (q_iq_{i+1}c_i)$, we obtain
$\rd(\env(q_iq_{i+1}), q_iq_{i+1}) \leq \rd(c_i, q_iq_{i+1}) \leq \eps$ as desired. In the latter case, we have
$\rd(\env(q_iq_{i+1}), q_iq_{i+1}) = \rd(q', q_iq_{i+1}) \leq \eps$. That is, we claim that $q'$ is a point of
$\env(q_iq_{i+1})$ (by Claim~\ref{clm:sandwich}) with maximum ratio distance from $q_{i}q_{i+1}$.
Since the feasible points in $\env(q_{i}q_{i+1})$ lie between $q_iq_{i+1}$ and its parallel line through $q'$, the
claim follows. This completes the proof of Lemma~\ref{lem:chord-finds-eps-convex}.
\end{proof}

\begin{figure}[h!]
\begin{center}
\epsfig{file=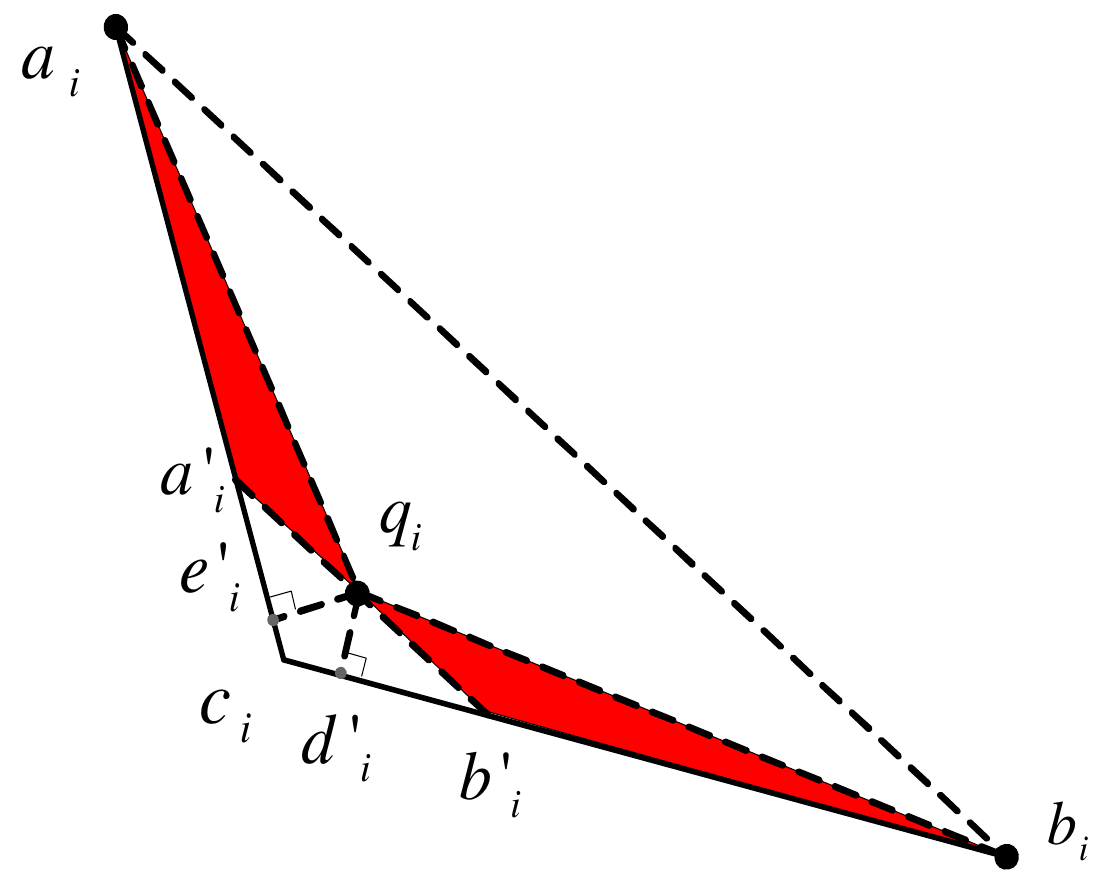,width=9cm}
\end{center}
\caption{Illustration of the area shrinkage property of the Chord algorithm.} \label{fig:chord-area}
\end{figure}

To bound from above the performance ratio we need a few more lemmas. Our first lemma quantifies the area shrinkage property. It
is notable that this is a statement independent of $\epsilon$.
\begin{lemma} \label{lem:area-half}
Let $T_{i} = \triangle (a_ib_ic_i)$ be the triangle processed by the Chord algorithm at some recursive step.
Denote $q_i = \comb (\lambda_{a_i b_i})$. Let $T_{i, l} = \triangle (a_i a'_i q_i)$, $T_{i,r} = \triangle (b_i b'_i q_i)$
be the triangles corresponding to the two new subproblems. Then, we have
\[ S(T_{i,l}) + S(T_{i,r}) \leq S(T_i)/4. \]
\end{lemma}
\begin{proof}
Let $d'_i, e'_i$ be the projections of $q_i$ on $b_ic_i$ and $a_ic_i$ respectively; see Figure~\ref{fig:chord-area} for an
illustration. Let
\begin{equation} \label{eq:y}
y \eqdef (a'_ic_i)/(a_ic_i) = (b'_ic_i)/(b_ic_i) \in [0,1]
\end{equation}
where the equality holds because the triangles $T_i$ and $\triangle (a'_ib'_ic_i)$ are similar (recalling that $a'_ib'_i
\parallel a_ib_i$). Hence, we get
\begin{equation} \label{eq:y2}
S(\triangle(a'_ib'_ic_i)) = y^2 S(T_i).
\end{equation}
We have
\begin{eqnarray}
S(T_{i,l})&+&S(T_{i,r}) = \big( (a_ia'_i) \cdot (q_i e'_i) + (b_ib'_i) \cdot (q_id'_i)  \big)/2 \nonumber\\
                        &=& (1-y) \cdot \big( (a_ic_i) \cdot (q_i e'_i) + (b_ic_i) \cdot (q_id'_i) \big)/2 \label{uses:y}\\
                        &=& (1-y)\cdot \big( S(\triangle(a_ic_iq_i))+ S(\triangle(b_ic_iq_i)) \big) \nonumber \\
                        &=& (1-y)\cdot \big( S(T_{i,l}) + S(T_{i,r}) + S(\triangle(a'_ib'_ic_i)) \big) \nonumber
\end{eqnarray}
where (\ref{uses:y}) follows from (\ref{eq:y}). By using (\ref{eq:y2}) and expanding we obtain
$$S(T_{i,l}) + S(T_{i,r}) = y \cdot (1-y) \cdot S(T_i) \leq S(T_i)/4$$ as desired. 
\end{proof}


Our second lemma gives a convenient lower bound on the value of the optimum.
\begin{lemma} \label{lem:opt-lb}
Consider the recursion tree $\mathcal{T}$ built by the algorithm and let $\mathcal{L}'$ be the set of lowest internal nodes, i.e., the internal nodes whose children are leaves. Then $\opt_{\eps} \geq |\mathcal{L}'|$.
\end{lemma}
\begin{proof}
Recall that, by convention, there is no node in the tree if, for a triangle, the Chord routine terminates without calling
$\comb$. Each lowest internal node of the tree corresponds to a triangle $T_i = \triangle (a_i b_i c_i)$ with the property that
the ratio distance of the convex Pareto set from the line segment $a_ib_i$ is strictly greater than $\eps$ (as otherwise, the
node would be a leaf). Each such triangle must contain a point $q \not \in \{a_i, b_i \}$ of an optimal $\eps$-CP set.
Any two  nodes of $\mathcal{L}'$ are not ancestor of each other, and therefore
the corresponding triangles are disjoint (neighboring triangles can only intersect at an endpoint).
Thus, each one of them must contain a distinct point of the
optimal $\eps$-CP set, and hence the lemma follows. 
\end{proof}
\begin{figure}[h!]
\begin{center}
\epsfig{file=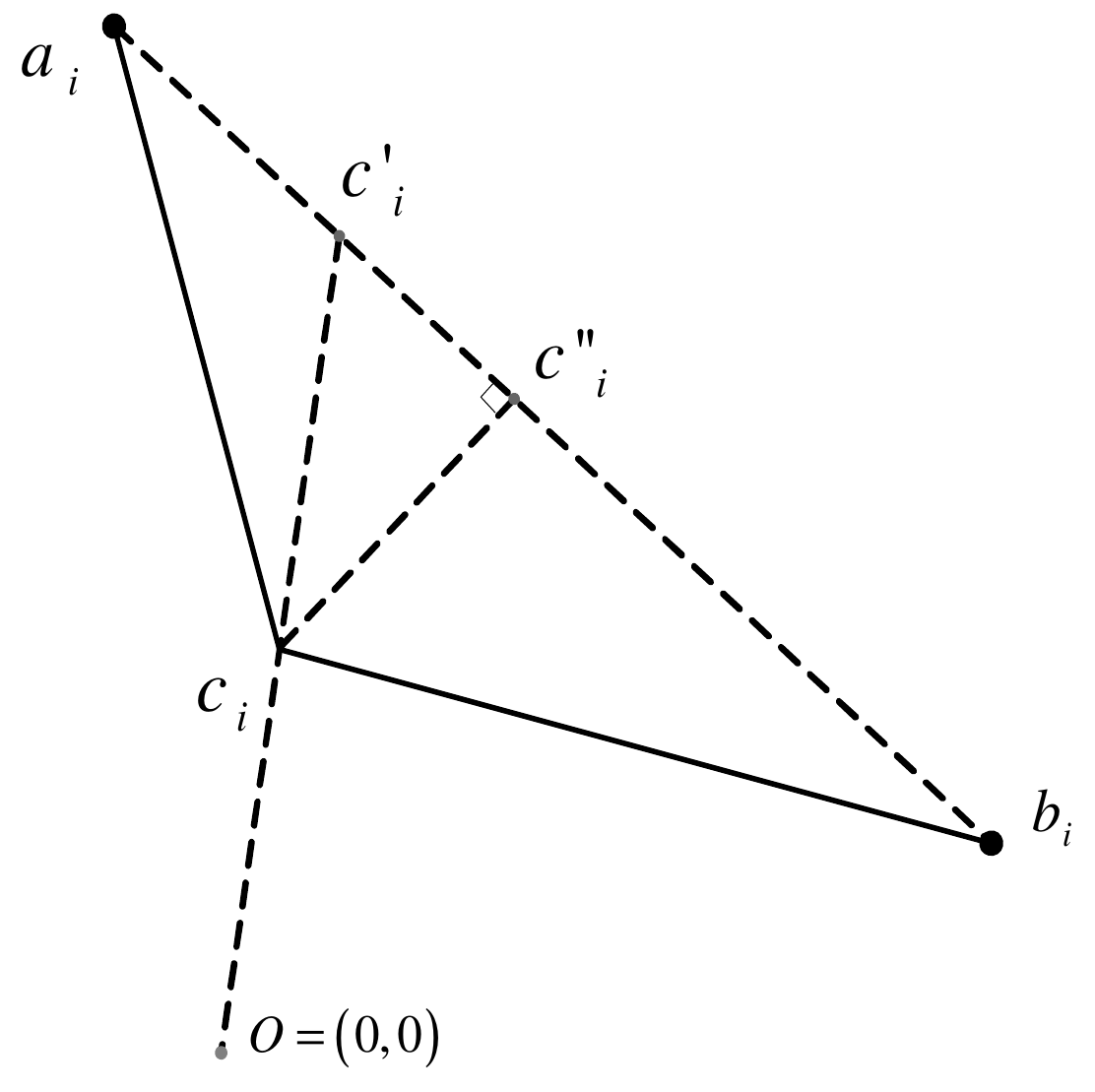,width=8cm}
\end{center}
\caption{Illustration of the proof of Lemma~\ref{lemma:small-area-means-done}.} \label{fig:area-distance}
\end{figure}
Finally, we need a lemma that relates the ratio distance within a triangle to its area. We stress that the lemma applies only
to triangles $T_i = \triangle(a_ib_ic_i)$ considered by the algorithm.
\begin{lemma} \label{lemma:small-area-means-done}
Consider a triangle $T_i = \triangle(a_ib_ic_i)$ considered in some iteration of the Chord algorithm such that $T_i \subseteq
T_1$. Let $\alpha_i \eqdef \min\{x(c_i), y(c_i)\}$. If $S(T_i) \leq \epsilon^2 \cdot \alpha_i^2$, then $\rd(c_i, a_ib_i) \leq
\epsilon$.
\end{lemma}
\begin{proof}
The basic idea of the argument is that the worst-case for the area-error tradeoff is (essentially) when $c_i = (\alpha_i,
\alpha_i)$, and the triangle $T_i$ is right and isosceles (i.e., $(a_ic_i) = (c_ib_i)$).
We now provide a formal proof. Let $T_i = \triangle(a_ib_ic_i)$ be a triangle considered by the algorithm. 
The points $a_i, b_i \in \env(\I)$ and we have $x(a_i) \leq x(c_i) \leq
x(b_i)$, $y(a_i) \geq x(c_i) \geq x(b_i)$ and the point $c_i$ lies below the line $a_ib_i$. 
Note the latter imply that $\angle(a_ic_ib_i) \geq \pi/2.$ 
(See Figure~\ref{fig:area-distance}.)
We will relate the area $S(T_i)$ to the ratio distance  $r \eqdef \rd(c_i, a_ib_i)$.

Consider the intersection $c'_i$ of the lines $a_ib_i$ and $Oc_i$, where $O$ denotes the origin. Then we have that $c'_i = (1+r)c_i$.
From the definition of the ratio distance it follows that for any point $p \in a_ib_i$ it holds $\rd(c_i,p) \geq r$ (in fact,
$c'_i$ is the unique minimizer). Let $c''_i$ be the projection of $c_i$ on $a_ib_i$. It follows that
\begin{equation}
\max \{ y(c''_i) / y(c_i), x(c''_i) / x(c_i) \} \geq 1+r \label{eq:star}.
\end{equation}
For the area we have that $S(T_i)  = (1/2)(a_ib_i)(c_ic''_i)$. Since $T_i$ is either right or obtuse, the length of its largest
base is at least twice the length of the corresponding height, i.e., $(a_ib_i) \geq 2(c_ic''_i)$. Hence, $S(T_i) \geq
(c_ic''_i)^2$. By expanding and using (\ref{eq:star}), we get $S(T_i) \geq r^2 \cdot \alpha_i^2$ and the lemma follows.
\end{proof}
\noindent At this point we have all the tools we need to complete the proof of Theorem~\ref{thm:ws-ub-area}. 
\begin{proof}[{\bf Proof of Theorem~\ref{thm:ws-ub-area}}]
Lemma~\ref{lem:chord-finds-eps-convex} gives correctness. To bound the performance ratio we proceed as follows: 
First, by Lemma~\ref{lem:area-half}, when a node of the tree is at depth $\lceil \log_4 (S(T_1)/S_0) \rceil$, 
the corresponding triangle will have area at most $S_0$. Hence, by Lemma~\ref{lemma:small-area-means-done} 
(noting that $\min_i \alpha_i  = \alpha$), it follows that the depth of the recursion tree $\mathcal{T}$ is $d = O(\log (S(T_1)/S_0))$. 
Every internal tree node is an ancestor of a node in $\mathcal{L}'$.
The Chord algorithm makes one query for every node of the tree, hence $\chd_{\eps} \le O(d) \cdot |\mathcal{L}'|$.
Lemma~\ref{lem:opt-lb} now implies that $\chd_{\eps} \le O \left( \log \left( S(T_1)/S_0 \right) \right) \cdot \opt_{\eps}$, which concludes the proof. 
\end{proof}


\subsubsection{Tight Upper Bound} \label{sssec:worst-upper-final}

\smallskip

In this subsection, we prove the asymptotically tight upper bound of  
$O \big( \frac{m+\log(1/\epsilon)}{\log m +
\log\log(1/\epsilon)} \big)$ on the worst-case performance ratio of the Chord algorithm.

The analysis is more subtle in this case and builds on the intuition obtained from the simple analysis of the previous subsection. 
The proof bounds in effect the length of paths in the recursion tree that consist of nodes with a single child. It shows that if we consider any $\epsilon$-convex Pareto set $S$, 
the segment of the lower envelope between any
two consecutive elements of $S$  cannot contain too many points of the solution produced by
the Chord algorithm.


\begin{theorem}\label{thm:ws-ub-tight}
The worst-case performance of the Chord algorithm (with respect to the ratio distance)
is $O \big( \frac{m+\log(1/\epsilon)}{\log m +
\log\log(1/\epsilon)} \big)$.
\end{theorem}
\begin{proof}
We begin by analyzing the case $\opt_{\eps}=3$ (i.e., the special case that one
intermediate point suffices -- and is required -- for an $\eps$-approximation) and then handle the general case. It turns out
that this special case captures most of the difficulty in the analysis.

Let $a$ (leftmost) and $b$ (rightmost) be the extreme points of the convex Pareto curve as computed by the algorithm. We
consider the case $\opt_{\eps}=3$, i.e., (i) the set $\{a,b\}$ is {\em not} an $\eps$-CP and (ii) there exists a solution point
$q^{\ast}$ such that $\{a,q^{\ast},b\}$ is an $\eps$-CP.

Fix $k \in \mathbb{N}$ with $k = \Theta \big( \frac{m+\log(1/\epsilon)}{\log m + \log\log(1/\epsilon)} \big)$. We will prove
that, for an appropriate choice of the constant in the big-Theta, the Chord algorithm introduces at most $k$ points in either
of the intervals $[a, q^{\ast}]$, $[q^{\ast}, b]$. Suppose, for the sake of contradiction, that the Chord algorithm adds more
points than that in the segment $aq^{\ast}$ (the proof for $q^{\ast}b$ being symmetric.)

\medskip

We say that, in some iteration of the Chord algorithm, a triangle is \emph{active}, if it contains the optimal point
$q^{\ast}$. In each iteration, the Chord algorithm has an active triangle which contains the optimal point $q^{\ast}$. Outside
that triangle, the algorithm has constructed an $\eps$-approximation. We note that the Chord algorithm may in principle go back
and forth between the two sides of $q^{\ast}$; i.e., in some iterations the line parallel to the chord touches the lower
envelope to the left of $q^{\ast}$ and in other iterations to the right.

\smallskip

Let $\triangle(abc)$ be the initial triangle. We focus our attention on the (not necessarily consecutive) iterations of the
Chord algorithm that add points to the left of $q^{\ast}$. We index these iterations in increasing order with $i \in [1,k+1]$.
Consider the ``active'' triangle $\triangle(a_ic_ib_i)$ generated in each such iteration, where $a_i$ is the new point of the
curve added in the iteration. We let $\triangle(a_0b_0c_0)$ be the initial triangle $\triangle(abc)$.
It is clear that (i) $a_{i+1}$ lies to the right of $a_i$, (ii) $b_{i+1}$ lies to the left of (or
is equal to) $b_i$ and (iii) $\triangle(a_{i+1}c_{i+1}b_{i+1}) \subseteq \triangle(a_ic_ib_i)$.  Also, let us denote
$b':=b_{k+1}$, that is $b'$ is the $b$-vertex (i.e., the vertex that lies to the right of $q^{\ast}$) of the active triangle in
the last iteration $k+1$. Note that $b'$ could be equal to the point $b$ (which would happen if the Chord algorithm
introduces points only to the left of $q^{\ast}$, i.e., proceeds towards $q^{\ast}$ monotonically), but in general this will not be the case.

Let $e_i$ be the intersection point of the line $a_ic_i$ with the line $b'q^{\ast}$; see Figure \ref{tight-upper}. 
Note that, to the left of $q^{\ast}$ the
convex Pareto curve has no points below the line $b'q^{\ast}$. All the points of the convex Pareto curve that are left of
$q^{\ast}$ and lie in the active triangle $\triangle(a_ic_ib_i)$ (these are the potentially not $\eps$-covered points) are
actually in the triangle $\triangle(a_ie_iq^{\ast})$. Consider the line that goes through $a_i$ and is parallel to
$a_{i-1}b'$ and let $d_i$ be its intersection with $b'q^{\ast}$. Note that the line $a_ic_i$ is parallel to $a_{i-1}b_i$ (by
construction of the algorithm, this is the line that added the point $a_i$ and formed the active triangle
$\triangle(a_ic_ib_i)$) and $b_i$ lies to the right of (or is equal to) $b'$, so the line $a_id_i$ is to the left of (or is equal to)
$a_ic_i$, hence $d_i$ lies to the left of (or is equal to) $e_i$ . Now, let $f_i$ be the intersection point of $a_{i-1}a_i$ with the line
$b'q^{\ast}$. Clearly, $f_i$ lies to the left of $d_i$. Furthermore, $f_i$ lies to the right of (or is equal to)
$d_{i-1}$. The reason is that, below $a_{i-1}$ the curve has no points (strictly) to the left of the line $a_{i-1}d_{i-1}$, so
$a_i$ is to the right of the line (or on the line).

\begin{figure*}[t]
\centering
\includegraphics[width=18cm]{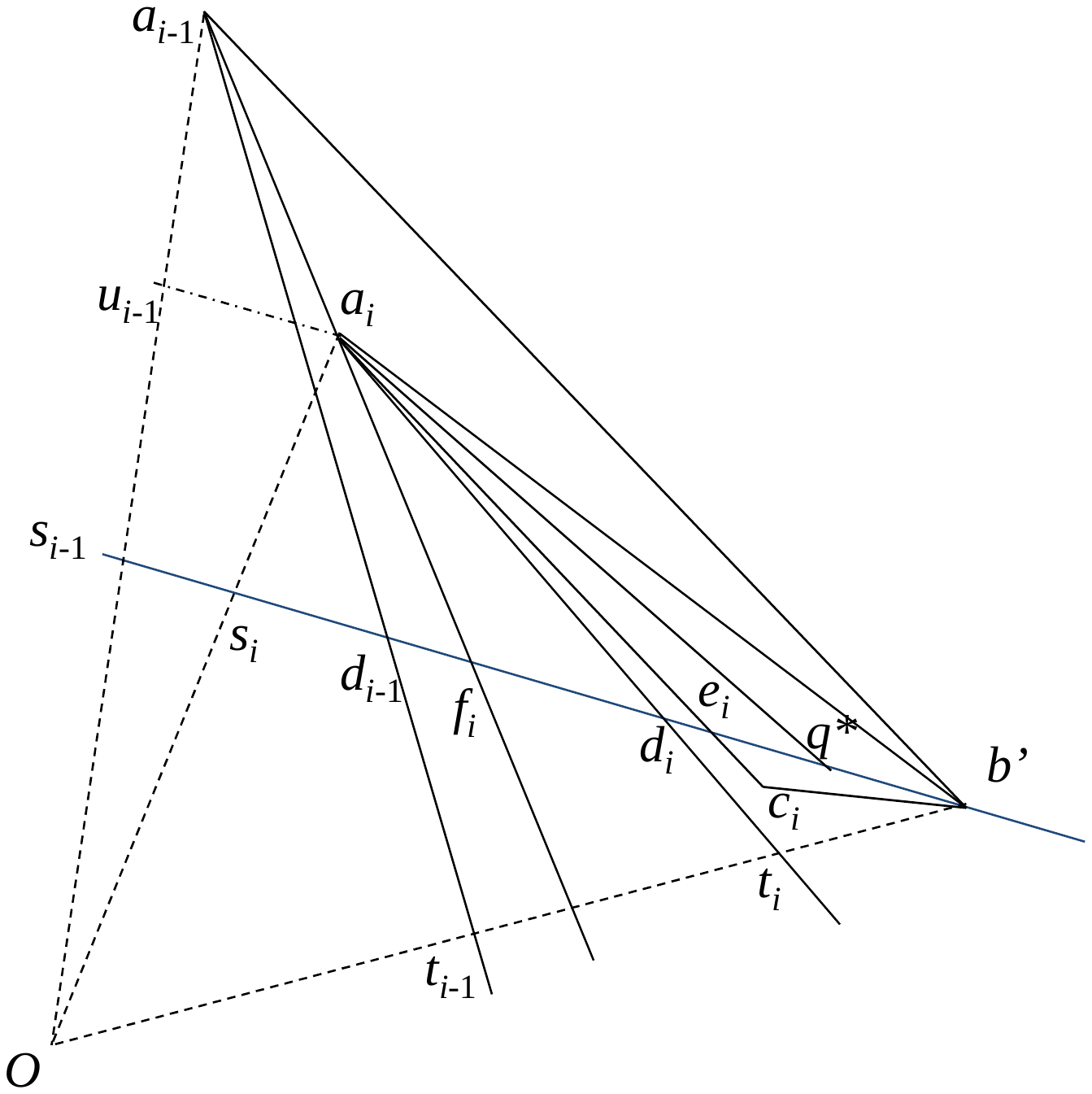}
\vspace*{-2.5cm}
\caption{Illustration of the proof of Theorem \ref{thm:ws-ub-tight}.}
\label{tight-upper}
\end{figure*}

\medskip

If a point $p \in \mathbb{R}^2_{+}$ is above a line $\ell$, in the sense that the segment connecting $p$
to the origin $O$ intersects $\ell$, say at some point $q$, then we define the
{\em excess ratio distance} of $p$ from $\ell$ to be $(pq)/(Oq)$, i.e., the ratio distance
of $q$ from $p$. 
Let $H_i$ be the excess ratio distance of $a_i$ from the line $b'q^{\ast}$ for $i\geq0$. 
Let $G_i$ be the excess ratio distance of $b'$ from the line $a_id_i$ for $i\geq1$. 
Referring to Figure \ref{tight-upper}, $H_i=(a_i s_i)/(Os_i)$
and $G_i = (b't_i)/(Ot_i)$.
Let $P_i = H_i / H_{i-1}$ and let
$R_i = 1-P_i$.

By definition, we have $H_i = P_i \cdot H_{i-1}$ for all $i$. Furthermore, $H_0$ is
the excess ratio distance of $a=a_0$ from the line $b'q*$.
Since the coordinates of the points are in $[2^{-m},2^{m}]$, it follows that $H_0 \leq 2^{2m}$.
Thus, we get
$$ H_k \leq 2^{2m} \cdot \prod_{i=1}^k P_i .$$

\begin{claim}\label{claim:tu1}
$P_i = \frac{(a_if_i)}{(a_{i-1}f_i)} = \frac{(f_id_i)}{(f_ib')}$ and $R_i= \frac{(d_ib')}{(f_ib')}$.
\end{claim}
\begin{proof}
Let $s_i$ be the intersection of the segment $a_iO$ with the line $b'q^{\ast}$
and $s_{i-1}$ the intersection of $a_{i-1}O$ with $b'q^{\ast}$.
By definition, $H_i = (a_is_i)/(Os_i)$, and $H_{i-1}=(a_{i-1}s_{i-1})/(Os_{i-1})$.
Let $u_{i-1}$ be the intersection of $a_{i-1}O$ with the line from $a_i$ parallel to $b'q^{\ast}$.
From the similar triangles  $\triangle(Os_{i-1}s_i)$ and  $\triangle(Ou_{i-1}a_i)$, we have
$H_i = (a_is_i)/(Os_i) = (u_{i-1}s_{i-1})/(Os_{i-1})$.
Therefore, $P_i = H_i/H_{i-1} = (u_{i-1}s_{i-1})/(a_{i-1}s_{i-1})$.
From the similar triangles $\triangle(a_{i-1}s_{i-1}f_i)$ and  $\triangle(a_{i-1}u_{i-1}a_i)$,
the latter ratio is equal to $(a_if_i)/(a_{i-1}f_i)$, yielding the first equality for $P_i$ in the claim.
The second equality follows from the similar triangles
$\triangle(f_ia_id_i)$ and  $\triangle(f_ia_{i-1}b')$, since $a_id_i$ is parallel to $a_{i-1}b'$.
The second equality implies then the expression for $R_i=1-P_i$.
\end{proof}

From the Claim we have, $(d_ib') = R_i \cdot (f_ib')$ 
and since $d_{i-1}$ lies left of $f_i$, we have $(d_ib') \leq R_i \cdot (d_{i-1}b')$.
Therefore, $$(d_kb') \leq  (d_1b') \cdot \prod_{i=2}^k R_i .$$

Let $t_i, t_{i-1}$ be the intersections of $Ob'$ with the lines $a_id_i$ and $a_{i-1}d_{i-1}$ respectively;
see Figure \ref{tight-upper}.
Clearly, $(b't_i) /(b't_{i-1}) \leq (b'd_i)/(b'd_{i-1})$, and hence $(b't_i) \leq R_i \cdot(b't_{i-1})$.
Thus, $G_i = \frac{(b't_i)}{Ot_i} \leq \frac{R_i (b't_{i-1})}{(Ot_{i-1})} = R_i \cdot G_{i-1}$.
Therefore,
$$ G_k \leq G_1 \cdot \prod_{i=2}^k R_i.$$

\begin{claim} \label{claim:tu2}
$H_k > \eps$ and $G_k > \eps$.
\end{claim}
\begin{proof}
Since the last iteration of the Chord algorithm adds a new point $a_{k+1}$, 
the segment $a_kb'$ does not $\eps$-cover all the Pareto points left of $q^{\ast}$ 
in the active triangle. These points are all in the triangle $\triangle(a_kd_kb')$.
The ratio distance of any point in this triangle from $a_kb'$ is upper bounded
by both $(a_ks_k)/(Os_k) = H_k$ and by  $(b't_k)/(Ot_k) = G_k$.
It follows that $H_k > \eps$ and $G_k > \eps$.
\end{proof}

 Thus, we get
\begin{equation}
\prod_{i=1}^k P_i > \eps / 2^{2m}  \label{p-bound}
\end{equation}
and 
\begin{equation}
G_1 \cdot \prod_{i=2}^k R_i  > \eps  \label{g-bound}
\end{equation}

\begin{claim} \label{claim:tu3}
$\prod_{i=2}^k R_i>1/2$.
\end{claim}
\begin{proof}
If $G_1 \leq 2\eps$ then the claim follows from inequality (\ref{g-bound}).
So suppose that $G_1 > 2\eps$.
The point $a_1$ is at ratio distance at most $\eps$ from the line $aq^{\ast}$
since $\{a, q^{\ast},b\}$ is an $\epsilon$-convex Pareto set. 
Therefore, $q^{\ast}$ is
at most excess ratio distance $\eps$ from the line $a_1d_1$, because $a_1d_1$ is parallel to $ab'$ and $b'$ is to the right of
$q^{\ast}$. 
Since $G_1 > 2\eps$, it follows that $(d_1q^{\ast})<(q^{\ast}b')$ and hence $(d_1b')<2(q^{\ast}b')$. Therefore,
$$(d_kb') \leq (d_1b') \cdot \prod_{i=2}^k R_i  < 2(q^{\ast}b') \cdot \prod_{i=2}^k R_i .$$
Since $(d_kb') > (q^{\ast}b')$ (as $d_k$ is left of $q^{\ast}$), we conclude that
$$\prod_{i=2}^k R_i > 1/2.$$ 
\end{proof}

Thus, we have a lower bound on the product of the $P_i$'s from
inequality (\ref{p-bound}) and on the product of the $R_i$'s from Claim \ref{claim:tu3}.
It is easy to see (and is well-known) that for a fixed product of the $P_i$'s, the product of the
$R_i$'s is maximized if all factors are equal. We include a proof for convenience.

\begin{claim}
Let $0 < x_i <1$ for $i=1,\ldots , k$.
The maximum of $\prod_i (1-x_i)$ subject to $\prod_i x_i =c$ is achieved when
all the $x_i$ are equal.
\end{claim}
\begin{proof}
Suppose $k=2$. Then $(1-x_1)(1-x_2) = 1 -(x_1+x_2) + x_1x_2$ is maximized subject to $x_1x_2=c$,
when $x_1+x_2$ is minimized, which happens when $x_1=x_2$ by the arithmetic-geometric mean inequality
($x_1+x_2 \geq 2\sqrt{x_1x_2}$ with equality iff $x_1=x_2$.
For general $k \geq 2$, if the $x_i$'s maximize $\prod_i (1-x_i)$ subject to $\prod_i x_i =c$
then we must have $x_i=x_j$ for all pairs $i \neq j$, because otherwise
replacing $x_i, x_j$ by their geometric mean will
increase $\prod_i (1-x_i)$.
\end{proof}

Thus, for any value of $\prod_{i=2}^k P_i $, the product $\prod_{i=2}^k R_i$
is maximized when  $P_i = 1/t$ for all $i =2,\ldots,k$ and $R_i = 1-1/t$.
Since $\prod_{i=2}^k R_i > 1/2$, we must have $k-1 < t$ because $(1-1/t)^t < 1/e <1/2$.
Therefore, $\epsilon/2^{2m} < \prod_{i=2}^k P_i < 1/(k-1)^{k-1}$,
hence $(k-1)^{k-1} < 2^{2m}/\epsilon$,
which implies that
$k= O \big( \frac{m+\log(1/\epsilon)}{\log m + \log\log(1/\epsilon)} \big)$.

\medskip

We now proceed to analyze the general case, essentially by reducing it to the aforementioned special case.  
Suppose that the optimal
solution has an arbitrary number of points, i.e., has the form $Q^{\ast} = \langle a, q_1, q_2, \ldots, q_r, b \rangle$. Charge
the points computed by the Chord algorithm to the edges of the optimal solution as follows:
if a point belongs to the portion of the lower envelope between the points $q_{i-1}$ and $q_{i}$
(where we let $a=q_0$ and $b=q_{r+1})$, then we charge the point to
the edge $q_{i-1}q_{i}$; if the Chord algorithm generates a point $q_i$ of the
optimal solution then we can charge it to either one of the adjacent edges.

We claim that every edge of $Q^{\ast}$ is charged with at most $2k+1$ points of the
Chord algorithm, where $k= O \big( \frac{m+\log(1/\epsilon)}{\log m +
\log\log(1/\epsilon)} \big)$ is the same number as in the
above analysis for the $\opt_{\eps}=3$ case.
To see this, consider any edge $q_{i-1}q_{i}$ of $Q^{\ast}$.
Let $a_0$ be the first point generated by the Chord algorithm that is charged to
this edge, i.e., $a_0$ is the first point that lies between $q_{i-1}$ and $q_{i}$.
We claim that the Chord algorithm will generate at most $k$ more points 
in each of the two portions $LE(q_{i-1}a_0)$ and $LE(a_0q_i)$ of the
lower envelope. The argument for the two portions is symmetric.

Consider the portion $\env(a_0q_i)$. The proof that the Chord algorithm will introduce
at most $k$ points in this portion is identical to the proof we gave above for the
$\opt_{\eps}=3$ case, with $a_0$ in place of $a$ and $q_i$ in place of $q^{\ast}$.
The only fact about the assumption $\opt_{\eps}=3$ that was used there was that
the edge $aq^{\ast}$ $\epsilon$-covers the portion of the lower envelope 
between $a$ and $q^{\ast}$.
It is certainly true here that the segment $a_0q_i$ $\epsilon$-covers 
the portion $\env(a_0q_i)$, since the edge $q_{i-1}q_i$ $\epsilon$-covers
$\env(q_{i-1}q_i) \supseteq \env(a_0q_i)$.
Hence, by the same arguments, the Chord algorithm will generate at most $k$
points between $a_0$ and $q_i$. Thus, the algorithm will generate no more
than  $(2k+1)(r+1)$ points overall, and hence its performance ratio is
$O \big( \frac{m+\log(1/\epsilon)}{\log m +
\log\log(1/\epsilon)} \big)$.
This completes the proof.
\end{proof}


\begin{Remark}
\emph{We briefly sketch the differences in the algorithm and its analysis for the case of an approximate Comb routine. 
First, in this case, the description of the Chord algorithm (Table~\ref{table:chord}) has to be slightly modified; this is needed 
to guarantee that the set of computed points is indeed an $\eps$-CP. In particular, in the Chord routine, 
we need to check whether $\rd(q, lr) \le \eps'$ for an appropriate $\eps'<\eps$. In particular, we choose $\eps'$
such that $(1+\eps')(1+\delta) \le (1+\eps)$, where $\delta$ is the accuracy of the approximate Comb routine, i.e.,  the routine $\comb_{\delta}$.
Consider the case that the $\comb_{\delta}$ routine always returns feasible points that belong to a $(1+\delta)$ scaled version of 
the lower envelope. The same analysis as in the current section establishes that 
the Chord algorithm performs at most $O \big( \frac{m+\log(1/\eps')}{\log m + \log\log(1/\eps')} \big)  \opt_{\eps'}$ 
calls to $\comb_{\delta}$ in this setting. If $\eps'$ is ``close'' to $\eps$ (say, $\eps' \ge \eps/2$) the first term is clearly 
$O \big( \frac{m+\log(1/\eps)}{\log m + \log\log(1/\eps)} \big)$. 
Hence, to prove the desired upper bound, it suffices to show that $\opt_{\eps'} = O(\opt_{\eps})$. (It is clear that $\opt_{\eps'} \ge \opt_{\eps}$, but in principle it may be
the case that $\opt_{\eps'}$ is arbitrarily larger.) This is provided to us by a planar geometric lemma from~\cite{DY2} (Lemma~5.1) which states that if
$(1+\eps') \ge \sqrt{1+\eps}$ then $\opt_{\eps'} \le 3 \opt_{\eps}$. Selecting $\eps' = \delta = \sqrt{1+\eps} -1 \ge \eps/2$ suffices for the above and completes our sketched description.}
\end{Remark}

\section{Average Case Analysis} \label{sec:average}
In Section~\ref{ssec:average-upper} we present our average case upper bounds and in Section~\ref{ssec:average-lower}
we give the corresponding lower bound.

\subsection{Upper Bounds} \label{ssec:average-upper}
In Section~\ref{ssec:ppp-upper} we start by proving our upper bound for random instances drawn from a Poisson Point Process (PPP). 
The analysis for the case of unconcentrated product distributions is somewhat more involved and is given in Section~\ref{ssec:prod-upper}. 

\smallskip

\noindent \textbf{Overview of the Proofs.}  The analysis for both cases has the same overall structure, however each case
has its difficulties. We start by giving a high-level overview of the arguments. For the sake of simplicity, in the
following intuitive explanation, let $n$ denote: (i) the expected number of points in the instance for a PPP and (ii) the
actual number of points for a product distribution. 

Similarly to the simple proof of Section~\ref{sssec:worst-upper-simple} for worst-case instances, to analyze our distributional instances 
we resort to an indirect measure of progress, namely the area of the triangles maintained in the algorithm's recursion tree. 
We think that this feature of our analysis is quite interesting and indicates that this measure is quite robust.

In a little more detail, we first show (see Lemma~\ref{lem:area recursion} for the case of PPP) that every subdivision performed by the algorithm
decreases the area between the upper and lower approximations by a significant amount (roughly at an exponential rate) with
high probability. It then follows that at depth $\log \log n$ of the recursion tree, each ``surviving triangle''  contains an
expected number of at most $\log \log n$ points with high probability. We use this fact, together with a charging argument in
the same spirit as in the worst-case, to argue that the expected performance ratio is $\log \log n$ in this case.

To analyze the expected performance ratio in the complementary event,
we break it into a ``good'' event,
under which the ratio is $\log n$ with high probability,
and a ``bad'' event, where it is potentially unbounded (in the Poisson case)
or at most $n$ (for the case of product distributions).
The potential unboundedness of the performance ratio in the Poisson case creates
complications in bounding the expected ratio of the algorithm over the entire space.
We overcome this difficulty by bounding the upper tail of the Poisson distribution (see Claim~\ref{claim:ppp-tail}).

In the case of product distributions, the worst case bound of $n$ on the competitive ratio
is sufficient to conclude the proof, but the technical challenges present themselves
in a different form. Here, the ``contents'' of a triangle being processed by the algorithm depend
on the information coming from the previous recursive calls making the analysis more involved.
We overcome this by understanding the nature of the information provided from the conditioning.

\smallskip

\noindent \textbf{On the choice of parameters.} A simple but crucial observation concerns the interesting range for the
parameters of the distributions. Suppose that we run the Chord algorithm with desired error $\eps>0$ on some 
random instance that lies entirely in the set  $[2^{-m}, 2^{m}]^2$. Then, it is no loss of generality to assume
that the number of random points in the instance (expected number for the PPP case) is upper bounded by some fixed
polynomial in $2^m$ and $1/\eps$. If this is not the case, it is easy to show that the Chord algorithm  makes at most
a constant number of $\comb$ calls in expectation.

\subsubsection{Poisson Point Process} \label{ssec:ppp-upper}
For the analysis of the PPP case we will make crucial use of the following technical claim:
\begin{claim} \label{claim:ppp-tail}
Let $X$ be a ${\rm Poisson}(\nu)$ random variable, with $\nu \ge 1$, and let $\cal E$ be some event. Then
$$\mathbb{E}[X~|~{\cal E}] \Pr[{\cal E}] \le  \max\left\{{1 \over \nu}, O(\nu^3) \Pr[{\cal E}] \right\}.$$
\end{claim}
\begin{proof}
Let $k^*$ be such that $\Pr[X \ge k^*+1] < \Pr[{\cal E}] \le \Pr[X \ge k^*]$. Clearly,
\begin{equation*}
\mathbb{E}[X~|~{\cal E}] \Pr[{\cal E}] \leq \littlesum_{i=k^*}^{+ \infty} i \cdot \Pr[X=i] 
=  \littlesum_{i=k^*}^{+ \infty} i \cdot {e^{- \nu} \nu^i \over i!} 
= \nu \littlesum_{i=k^*-1}^{+ \infty} {e^{- \nu} \nu^{i} \over i!} = \nu \Pr[X \ge k^*-1].
\end{equation*}
We distinguish two cases. If $k^* -1 \ge 2 \nu^2$, then Chebyshev's inequality yields
$\Pr[X \ge k^*-1] \le {1 \over \nu^2}$ which gives
$$\mathbb{E}[X~|~{\cal E}] \Pr[{\cal E}] \le {1\over \nu}.$$
If $k^*-1 \le 2 \nu^2$, then
\[\Pr[X = k^*] \le {k^*+1 \over \nu} \Pr[X=k^*+1] \le O(\nu) \Pr[{\cal E}]\]
and
\[ \Pr[X = k^*-1] \le {(k^*+1)^2 \over \nu^2} \Pr[X=k^*+1] \le O(\nu^2) \Pr[{\cal E}].\]
Hence,
\begin{eqnarray*}
\mathbb{E}[X~|~{\cal E}] \Pr[{\cal E}] &\le&  \nu  \cdot  \Pr[X \ge k^*-1]\\
&\le& \nu  \cdot  \big(\Pr[{\cal E}]+ \Pr[X = k^*-1] + \Pr[X = k^*]\big)\\
&\le& O(\nu^3) \cdot \Pr[{\cal E}].
\end{eqnarray*}
This concludes the proof of the claim. 
\end{proof}

We start by pointing out that if the intensity of the PPP is very large, the Chord algorithm will terminate after a constant number of 
calls in expectation:
\begin{proposition} \label{fact:dens-ppp}
Let $T_1 = \triangle (abc)$ be at the root of the Chord algorithm's recursion tree and let 
$\nu$ denote the intensity of the PPP. 
Let $\alpha \eqdef \min \{ x(c), y(c) \}$ and $S^{\ast} \eqdef  (\eps^2 \alpha^2/2) \cdot \min \{\lambda_{ab}, 1/\lambda_{ab}\}$. 
If $\nu \ge \nu_0 \eqdef 10 S(T_1) / (S^{\ast})^2$, then $\E [\chd_{\eps}(T_1)]  = O(1)$.
\end{proposition}

\begin{proof}
First note we can clearly assume that $y(a) > (1+\eps) \cdot y(b)$ and $x(b) > (1+\eps) \cdot x(a)$. 
Let $p_1 = (x(a), (1+\eps) \cdot y(b)) \in ac$ and $p_2 = ((1+\eps) \cdot x(a), y(b)) \in bc$. Let $T^{\ast} \subseteq
\triangle (cp_1p_2)$ be the right triangle of maximum area whose hypotenuse is parallel to $ab$. (This is the shaded
triangle in Figure~\ref{fig:ub-ac}.)
We claim that $S(T^{\ast}) \ge S^{\ast}.$
Indeed, it is clear that $c$ is a vertex of $T^{\ast}$ and that either $p_1$ or $p_2$ (or both) are vertices. 
Hence, one of the edges of $T^{\ast}$ has length at least $\eps \cdot \alpha$. Since
$\lambda_{ab}$ is the slope of the hypotenuse, the other edge has length at least $\min \{ \lambda_{ab}, 1/\lambda_{ab} \} \cdot (\eps 
\alpha)$. 

If there is a feasible point in $T^{\ast}$, the Chord algorithm will find it by calling $\comb(\lambda_{ab})$
and terminate (since such a point forms an $\eps$-CP set).
Let $X^{\ast}$ be the number of random points that land in the triangle $T^{\ast}$. Note that $X^{\ast}$ is a $\mathrm{Poisson} (\mu)$ random variable,
with $\mu = \nu \cdot S(T^{\ast})$. We can write
$$\E [\chd_{\eps}(T_1)]  = \E [\chd_{\eps}(T_1) \mid X^{\ast} =0] \Pr[X^{\ast}=0] +  \E [\chd_{\eps}(T_1) \mid X^{\ast} \ge 1] \Pr[X^{\ast} \ge 1].$$
Observe that the second term is bounded from above by a constant, hence it suffices to bound the first term.
Recall that the number of calls performed by the Chord algorithm is at most twice the number $Y_1$ of feasible points in the root triangle $T_1$.
Therefore, we can write
\[ \E [\chd_{\eps}(T_1) \mid X^{\ast} =0] \Pr[X^{\ast}=0] \le 2 \E [Y_1 \mid X^{\ast} =0] \Pr[X^{\ast}=0]. \]
Recall that $Y_1$ is a $\mathrm{Poisson} \left( \nu_1 \right)$ random variable, where $\nu_1 = \nu \cdot S(T_1)$ and note that we can assume wlog that
$\nu_1 \ge 1$ (since otherwise the expected number of feasible points in $T_1$ is at most $1$ and we are done). 
Hence, by Claim~\ref{claim:ppp-tail} the RHS above is bounded by
\[ 2 \max \left\{ \frac{1}{\nu_1}, O(\nu_1^3) \Pr[X^{\ast}=0] \right\}. \]
Thus, to complete the proof it suffices to show that $\nu_1^3 \Pr[X^{\ast}=0] = O(1)$. First note that this quantity equals 
$\nu^3 S^3(T_1) \exp\left(-\nu S(T^{\ast}) \right)$ which is at most $\nu_0^3 S^3(T_1) \exp\left(-\nu_0 S^{\ast} \right)$ by monotonicity
(which holds for our choice of $\nu_0$). The latter expression is at most $O(\beta^6 \exp(-10\beta))$, where $\beta = S(T_1)/S^{\ast} >1$, 
which is easily seen to be absolutely bounded.
\end{proof}

The same proposition also applies 
for the case that we have an approximate $\comb_{\delta}$ routine (in this case, we
replace $\eps$ by an appropriate $\eps' <\eps$ so that $(1+\delta)(1+\eps') \leq (1+\eps)$). 
We remark that a similar proposition can be shown for the Hausdorff distance; this does not hold for the case of the 
horizontal/vertical distance, which is why the average case bounds of this section do not apply for the latter metrics.

\smallskip

\noindent Note that by assumption $T_1 \subseteq [2^{-m}, 2^{m}]^2$ which implies that $S(T_1) \leq 2^{2m-1}$ and 
$\alpha \geq 2^{-m}$. We also have that $ \eps/2^{2m}
\leq \lambda_{ab} \leq 2^{2m}/\eps$ (since otherwise the set $\{a, b \}$ is an $\eps$-CP) which gives $S^{\ast} \ge \eps^3 2^{-4m-1}.$
Therefore, $\nu_0 = \poly(2^m/\eps).$




\begin{figure}[h!]
\vspace{-.1in}
\begin{center}
\epsfig{file=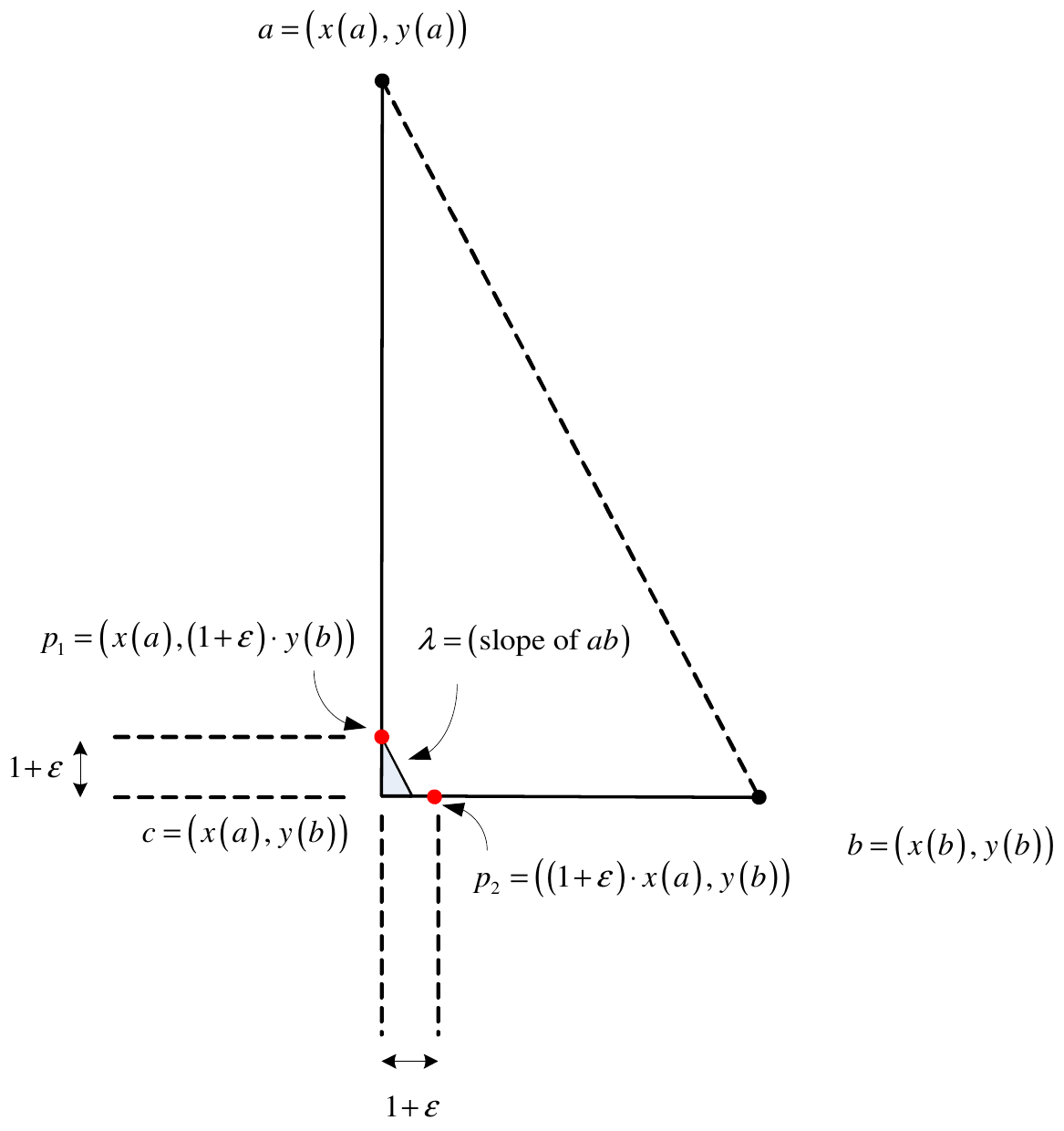,width=8cm}
\end{center}
\caption{Illustration of Proposition~\ref{fact:dens-ppp}.} \label{fig:ub-ac} \vspace{-.1in}
\end{figure}

The main result of this section is 
the following theorem, which combined with Proposition~\ref{fact:dens-ppp}, yields the desired upper bound of $O(\log m + \log \log (1/\eps))$ 
on the expected performance ratio.

\begin{theorem}\label{th: competitive ratio PPP}
Let $T_1$ be the triangle at the root of the Chord algorithm's recursion tree, and suppose that points are inserted into $T_1$
according to a Poisson Point Process with intensity $\nu$. The
expected performance ratio of the Chord algorithm on this instance is $O\left( \log \log \left( \nu S(T_1) \right) \right)$.
\end{theorem}

The proof of Theorem~\ref{th: competitive ratio PPP} will require a sequence of lemmas.
Throughout this section, we will denote $S_1 \eqdef S(T_1)$. Recall that the number of queries performed by the Chord algorithm 
is bounded from above by twice the total number of points in the triangle $T_1$. Since the expected number of points
in $T_1$ is $\nu \cdot S_1$, we have:

\begin{proposition} \label{lem: super pessimistic competitive ratio}
The expected performance ratio of the Chord algorithm is $O(\nu S_1)$.
\end{proposition}

Hence, we will henceforth assume that $\nu S_1$ is bounded from below by a sufficiently large positive constant. (If this is not the case,
the expected total number of points inside $T_1$ is $O(1)$ and the desired bound clearly holds.)

Our first main lemma in this section is an average case analogue of our Lemma~\ref{lem:area-half}:
the lemma says that the area of the triangles maintained 
by the algorithm decreases {\em geometrically} (as opposed to linearly, as is the case for arbitrary inputs) at every recursive step (with high probability). 
Intuitively, this geometric decrease is what causes the performance ratio to drop by an exponential in expectation.

\begin{lemma} \label{lem:area recursion}
Let $T_i = \triangle(a_ib_ic_i)$ be the triangle processed by the Chord algorithm at some recursive step. Denote $q_i =
\comb(\lambda_{a_ib_i})$. Let $T_{i,l} = \triangle (a_i a_i' q_i)$ and $T_{i,r}= \triangle (b_ib'_iq_i)$. For all $c >0$, with
probability at least $1-{1 \over \left( {\ln (\nu S_1)} \right)^c}$ conditioning on the information available to the algorithm,
\[ S(T_{i,l}), S(T_{i,r}) \le \sqrt{S(T_i)} \cdot \sqrt{{c \cdot \ln \ln {\nu} S_1 \over \nu}}.\]
\end{lemma}

\begin{proof}
It follows from the properties of the Chord algorithm (see e.g., Claim~\ref{clm:sandwich}) that, 
before the $\comb$ routine on input $T_i$ is invoked, the following
information is known to the algorithm, conditioning on the history (see Figure~\ref{fig:chord-avg}):
\begin{itemize}
\item There exist solution points at the locations $q_j$, for all $j \in [i-1]$.
\item There is no point below the line $a_j c_j$, for all $j \in [i]$.
\item There is no point below the line $c_j b_j$, for all $j \in [i]$.
\end{itemize}
\begin{figure}[h!]
\begin{center}
\epsfig{file=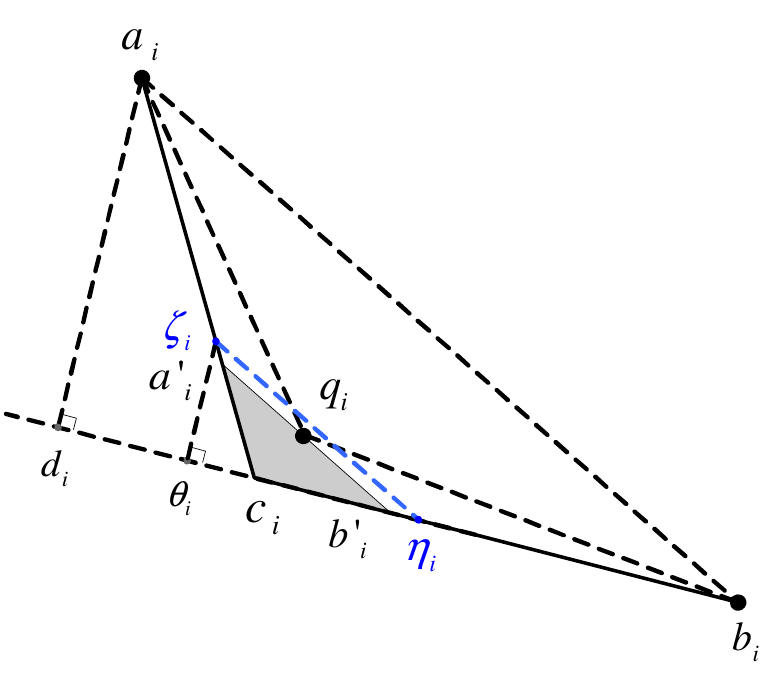,width=9cm}
\end{center}
\caption{Illustration of the average case area shrinkage property of the algorithm.} \label{fig:chord-avg}
\end{figure}
By the definition of the Poisson Point Process,
conditioning on the above information, the number of points in a region of area $S$ inside $T_i$ follows a Poisson
distribution with parameter $\nu \cdot S$. Hence, letting $\zeta_i \eta_i$ being parallel to $a_ib_i$ so that the triangle
$T^{\ast} \eqdef \triangle (c_i \zeta_i \eta_i)$ has area $S^{\ast} \eqdef {c \cdot \ln \ln (\nu S_1) \over \nu}$, it follows that,
with probability at least $1-{1 \over ({\ln (\nu S_1)})^c}$, $T^{\ast} \neq \emptyset$ (i.e., the triangle contains a feasible point). 
Hence, with probability at least $1-{1 \over ({\ln (\nu S_1)})^c}$ the point $q_i$ is contained in $T^{\ast}$. 

We bound from above the area of $T_{i,l}$ by the area of $T_{i,l}'= \triangle (a_i c_i \eta_i)$ (clearly, $T_{i, l} \subseteq T_{i,l}'$) and similarly the area of
$T_{i,r}$ by the area of $T_{i,r}'= \triangle(\zeta_i c_i b_i)$.
From the similarity of the triangles $T_i$ and $T^{\ast}$ we get
\begin{align*}
{(\zeta_i \eta_i) \over (a_ib_i)} = {(c_i \eta_i) \over (b_ic_i)} = {(\zeta_i \theta_i) \over (a_id_i)}.
\end{align*}
Hence,
$${S^{\ast} \over S(T_i)} = { {1 \over 2} (c_i\eta_i) (\zeta_i\theta_i) \over {1 \over 2} (b_ic_i)(a_id_i)} = \left({(c_i\eta_i) \over (b_ic_i)}\right)^2,$$
which gives $(c_i\eta_i) = (b_ic_i) \sqrt{{S^* \over S(T_i)}}$. Therefore,
\begin{eqnarray*}
S(T_{i,l}')&=& {1 \over 2} (a_id_i) (c_i\eta_i) = {1 \over 2} (a_id_i)(b_ic_i) \sqrt{{S^{\ast} \over S(T_i)}} \\
           &=& \sqrt{S(T_i) \cdot S^{\ast}} = \sqrt{S(T_i)} \cdot \sqrt{{c \cdot \ln \ln (\nu S_1) \over \nu}}.
\end{eqnarray*}
Finally,
$$S(T_{i,r}') = {1 \over 2} (\zeta_i\theta_i)(b_ic_i) = {1 \over 2} {(c_i\eta_i) \over (b_ic_i) } (a_id_i)(b_ic_i) =
S(T_{i,l}').$$ This concludes the proof of Lemma~\ref{lem:area recursion}.
\end{proof}

Let us choose $c \in \left({1 \over \ln \ln (\nu S_1)}, {\nu S_1  \over \ln \ln (\nu S_1)}\right)$, and let $T$ be
a triangle maintained by the algorithm at depth $d$ of the recursion. 
It follows from Lemma~\ref{lem:area recursion} that, with probability at least $1-{d \over ({\ln (\nu S_1)})^c}$,
$$S(T) \le S_1^{2^{-d}} \cdot \left( {c \cdot \ln \ln (\nu S_1) \over \nu} \right)^{1- 2^{-d}},$$
where to bound the probability of the above event we have taken a union bound over the events on the path of the recursion
tree connecting $T$ to the root of the recursion. 

Now consider the top $d^{\ast}:=\left\lceil \log_2 \ln (\nu S_1) \right\rceil$
levels of the recursion tree of the algorithm. 
Using Lemma~\ref{lem:area recursion} and a union bound it follows that, with overall probability at least $1-
{2 \cdot \ln (\nu S_1)  \over ({\ln (\nu S_1)})^c}$, the area of \emph{every} triangle at depth (exactly) $d^{\ast}$ of the recursion tree
is bounded from above by $$S^{\ast \ast} := \left( e \cdot c \cdot \ln \ln (\nu S_1) \right)/\nu,$$ where we used our assumption on the range of $c$.

Let $\cal A$ be the event that all the nodes (triangles) maintained by the algorithm at depth $d^{\ast}$ of the recursion tree (if
any) have area at most $S^{\ast \ast}$. In the aforementioned, we argued that the probability of the event $\cal A$ is at least 
$1- {2  \over \left( {\ln (\nu S_1)} \right)^{c-1}}$. We now show the following:

\begin{lemma} \label{lem: competitive ratio in the good event}
Conditioning on $\cal A$, the expected performance ratio of the Chord algorithm is $O\left(\log \log (\nu S_1) \right)$.
\end{lemma}

\begin{proof}
Let $\cal T$ be the recursion tree of the algorithm and ${\cal T}_{d^{\ast}}$ be obtained from $\cal T$ by pruning it at level $d^{\ast}$.
Let ${\cal V}_{d^{\ast}}$ be the set of nodes of ${\cal T}_{d^{\ast}}$
and let ${\cal L}_{d^{\ast}}$ be the subset of nodes in ${\cal V}_{d^{\ast}}$ that lie at depth $d^{\ast}$ from the root. 
Clearly ${\cal L}_{d^*}$ is a subset of the leaves of ${\cal T}_{d^{\ast}}$. 

For a triangle (node) $T$ maintained by the algorithm at depth $d^{\ast}$ of the recursion, that is $T \in {\cal L}_{d^{\ast}}$, we let the
random variable $X_T$ denote the number of points inside $T$. Also, denote by ${\cal L'}_{d^{\ast}}$ the set of lowest internal 
nodes of the tree ${\cal T}_{d^{\ast}}$. By (a straightforward analogue of) Lemma~\ref{lem:opt-lb}, we have $\opt_{\eps} \geq |{\cal L'}_{d^{\ast}}|$. 
Also, since ${\cal T}_{d^{\ast}}$ is a depth $d^{\ast}$ binary tree, it holds $|{\cal V}_{d^{\ast}}| \le 2d^{\ast} \cdot |{\cal L'}_{d^{\ast}}|$.

We condition on the information ${\cal F}$ available to the algorithm in the first $d^{\ast}$ levels of its recursion-tree (without
the information obtained from processing -- i.e., calling $\comb$ for -- any triangle at depth $d^{\ast}$). By assumption, ${\cal F}$
satisfies the event ${\cal A}$. Conditioning on the information ${\cal F}$, for all $T \in {\cal L}_{d^{\ast}}$, $X_T$
follows a Poisson distribution with parameter $\nu \cdot S(T)$. So, given that the event ${\cal A}$ holds, we have
$\mathbb{E}[X_T~\vline~{\cal F}] \le \nu \cdot S^{\ast \ast}$.

Note that the Chord algorithm makes a query to $\comb$ for every node in the tree ${\cal T}_{d^{\ast}}$.
Also recall that the number of queries performed by the algorithm on a triangle $T$ containing a total number of $X_T$
points is at most $2 X_T$. Hence, the expected total number of queries made can be bounded as follows:
\begin{eqnarray*}
\mathbb{E}[\chd_{\eps}~|~{\cal F}] &\le& |{\cal V}_{d^{\ast}}| + 2 \cdot \littlesum_{T \in {\cal L}_{d^{\ast}}} \mathbb{E}[X_T~|~{\cal F}]\\
&\le& |{\cal V}_{d^{\ast}}| + 2 \left|{\cal L}_{d^{\ast}}\right| \cdot \nu \cdot S^{\ast \ast}\\
&\le& |{\cal V}_{d^{\ast}}| +  4 |{\cal L}'_{d^{\ast}}| \cdot \nu \cdot S^{\ast \ast}\\
&\le& |{\cal L}'_{d^{\ast}}|  \cdot (2d^{\ast}+ 4  \nu \cdot S^{\ast \ast}),
\end{eqnarray*}
where the third inequality uses the fact that $\left|{\cal L}_{d^{\ast}}\right| \le 2 \left|{\cal L}'_{d^{\ast}}\right|$. So, conditioning on the
information ${\cal F}$, the expected performance ratio of the algorithm is
\begin{eqnarray*}
\mathbb{E}\left[ {\chd_{\eps} \over \opt_{\eps}}~\vline~{\cal F}\right] &\le& \mathbb{E}\left[ {\chd_{\eps} \over |{\cal
L}'_{d^{\ast}}|}~\vline~{\cal F}\right] \\
&\le& (2d^{\ast}+ 4  \nu \cdot S^{\ast \ast}) \\
&=& O \left( \log \log (\nu S_1) \right).
\end{eqnarray*}
Integrating over all possible ${\cal F}$ in ${\cal A}$ concludes the proof of Lemma~\ref{lem: competitive ratio in the good event}. 
\end{proof}

From Lemma~\ref{lem: competitive ratio in the good event} it follows that
$$\mathbb{E}\left[ {\chd_{\eps} \over \opt_{\eps}}~\vline~{\cal A}\right]=O\left( \log \log (\nu S_1) \right),$$
and from the preceding discussion we have that $\Pr\left[\bar{{\cal A}}\right] \le {2 \over ({\ln (\nu S_1)})^{c-1}}.$ Hence, we
have established the following.

\begin{lemma} \label{lem: good case competitive ratio}
For $c \in \left({1 \over \ln \ln (\nu S_1)}, {\nu S_1 \over \ln \ln (\nu S_1)}\right)$, there exists an event $\cal
A$, with $\Pr[{\cal A}] \ge 1-{2 \over ({\ln (\nu S_1)})^{c-1}}$, such that the expected performance ratio of the algorithm
conditioning on $\cal A$ is $O\left( \log \log (\nu S_1) \right)$.
\end{lemma}

Let ${\cal B}$ be the event that all the triangles at the level $\lceil \log_2 (\nu S_1) \rceil$ of the recursion tree of the
algorithm (if any) have area at most $(e \cdot c' \cdot \ln \nu S_1) / \nu$. With the same technique, but using
different parameters in the argument, we can also establish the following:

\begin{lemma}\label{lem: pessimistic competitive ratio}
For $c' \in \left({1 \over \ln (\nu S_1)}, {\nu S_1 \over \ln (\nu S_1)}\right)$, there exists an event $\cal B$, with
$\Pr[{\cal B}] \ge 1-{2 \over ({\nu S_1})^{c'-1}}$, such that the expected performance ratio of the algorithm conditioning
on $\cal B$  is $O\left( \log (\nu S_1) \right)$.
\end{lemma}

We want to use Lemmas~\ref{lem: good case competitive ratio} and~\ref{lem: pessimistic competitive ratio} together with 
Proposition~\ref{lem: super pessimistic competitive ratio} to deduce that the expected performance ratio of the algorithm 
is $O\left( \log \log (\nu S_1) \right)$. This may seem intuitive, but it is in fact not immediate. For technical purposes let us define the event ${\cal C} =
{\cal B} \setminus {\cal A}$, where ${\cal A}$ is the event defined in the proof of Lemma~\ref{lem: good case competitive
ratio}. It is easy to see that conditioning on ${\cal C}$ the expected performance ratio of the algorithm can still be bounded
by $O(\log (\nu S_1))$, since this expectation is affected only by whatever happens at level 
$\lceil \log_2 (\nu S_1)\rceil$ of the recursion tree and below. On the other hand, using the fact that 
$\Pr[{\cal A}] \ge 1-{2 \over ({\ln (\nu S_1)})^{c-1}}$, it follows that
$$\Pr[{\cal C}] \le {2 \over ({\ln (\nu S_1)})^{c-1}}.$$
We bound the expectation of the performance ratio using the law of total probability as follows:
\begin{eqnarray} \label{eq:long expression competitive ratio}
\mathbb{E}\left[ {\chd_{\eps} \over \opt_{\eps}}\right] 
&\le& \mathbb{E}\left[ {\chd_{\eps} \over \opt_{\eps}}~\vline~{\cal A}\right] \cdot \Pr[{\cal A}] 
+  \mathbb{E}\left[ {\chd_{\eps} \over \opt_{\eps}}~\vline~{\cal C}\right] \cdot \Pr[{\cal C}] \nonumber \\
&+& \mathbb{E}\left[ {\chd_{\eps} \over \opt_{\eps}}~\vline~\overline{{\cal A}\cup {\cal C}}\right] \cdot \Pr\left[\overline{{\cal A}\cup {\cal C}}\right] \nonumber \\
&\le& O\left( \log \log (\nu S_1) \right) \cdot \left( 1-{2 \over ({\ln (\nu S_1)})^{c-1}}\right) 
+ O\left( \log(\nu S_1) \right) \cdot {2 \over ({\ln (\nu S_1)})^{c-1}} \nonumber \\ 
&+& \mathbb{E}\left[ {\chd_{\eps} \over \opt_{\eps}}~\vline~\overline{{\cal A}\cup {\cal C}}\right] \cdot \Pr\left[\overline{{\cal A}\cup {\cal C}}\right].  \label{eqn:ena} 
\end{eqnarray}
where (\ref{eqn:ena}) follows from Lemmas~\ref{lem: good case competitive ratio} and~\ref{lem: pessimistic competitive ratio}.
 To conclude, we need to bound the last summand in the above expression. Note first that ${\cal B} \subseteq {\cal A}\cup {\cal
C}$. Hence,
$$\Pr\left[\overline{{\cal A}\cup {\cal C}}\right] \le {2 \over ({\nu S_1})^{c'-1}}.$$
We again use the fact that the number of queries made by the Chord algorithm (hence, also the performance ratio) 
is bounded by twice the total number of points in the triangle at the root of the
recursion tree. This number $X$ follows a Poisson distribution with parameter $\nu \cdot S_1$. Hence, we have
$$\mathbb{E}\left[ {\chd_{\eps} \over \opt_{\eps}}~\vline~\overline{{\cal A}\cup {\cal C}}\right] \cdot \Pr\left[\overline{{\cal A}\cup {\cal C}}\right] \le 2 \cdot \mathbb{E}\left[ X~\vline~\overline{{\cal A}\cup {\cal C}}\right] \cdot \Pr\left[\overline{{\cal A}\cup {\cal C}}\right].$$
To bound the right hand side of the above inequality we use Claim~\ref{claim:ppp-tail} and obtain:
\begin{eqnarray*}
\mathbb{E}\left[ {\chd_{\eps} \over \opt_{\eps}}~\vline~\overline{{\cal A}\cup {\cal C}}\right] \cdot \Pr \left[ \overline{{\cal A} \cup {\cal C}} \right] 
&\le& \max\left\{{1 \over \nu S_1}, O((\nu S_1)^3) \cdot \Pr\left[\overline{{\cal A}\cup {\cal C}}\right] \right\} \\
&\le& \max \left\{ {1\over \nu S_1}, O((\nu S_1)^3) {2 \over ({\nu S_1})^{c'-1}}\right\}.
\end{eqnarray*}
Choosing $c'=4$, the above RHS becomes $O(1)$. Plugging this into \eqref{eq:long expression competitive ratio} with $c=2$ gives
\begin{align*}
\mathbb{E}\left[ {\chd_{\eps} \over \opt_{\eps}}\right] = O\left( \log \log (\nu S_1) \right).
\end{align*}
This concludes the proof of Theorem~\ref{th: competitive ratio PPP}. 


\subsubsection{Product Distributions} \label{ssec:prod-upper}

We start by proving the analogue of Proposition~\ref{fact:dens-ppp}.
\begin{proposition} \label{fact:dens-unif}
Let $T_1 = \triangle (abc)$ be at the root of the Chord algorithm's recursion tree and suppose that $n$ points 
are inserted into $T_1$ independently from a $\gamma$-balanced distribution, where $\gamma \in [0, 1)$. 
Let $\alpha \eqdef \min \{ x(c), y(c) \}$, $S^{\ast} \eqdef  (\eps^2 \alpha^2/2) \cdot \min \{\lambda_{ab}, 1/\lambda_{ab}\}$
and $\beta \eqdef S(T_1)/S^{\ast}$. 
If $n \ge n_0 \eqdef 10\beta^2 / (1-\gamma)^2$, then $\E [\chd_{\eps}(T_1)]  = O(1)$.
\end{proposition}

\begin{proof}
As in Proposition~\ref{fact:dens-ppp} we have that 
$S(T^{\ast}) \ge S^{\ast}$ (see Figure~\ref{fig:ub-ac}).
The probability that a random point falls into $T^{\ast}$ is at least $t^{\ast} = (1-\gamma)(1/\beta)$, hence the probability than none of the $n$
points falls into $T^{\ast}$ is at most $(1-t^{\ast})^n$. As before, if there is a feasible point in $T^{\ast}$, the Chord algorithm 
will find it and terminate. In all cases, the algorithm terminates 
after at most $2n$ calls to $\comb$. Hence, $\E [\chd_{\eps}(T)] \leq 3 + 2n \cdot (1-t^{\ast})^n$. 
The last summand is bounded by $2n \cdot e^{-t^{\ast} \cdot n}$ which is at most $2n_0 \cdot e^{-t^{\ast} \cdot n_0}$ for all $n \ge n_0$ by monotonicity. 
The latter quantity equals $O(r^2 e^{-10r})$, where $r = \beta / (1-\gamma) >1$, which is easily seen to be $O(1)$.
\end{proof}

\noindent Recalling that $S(T_1) \leq 2^{2m-1}$ and $S^{\ast} \ge \eps^3 2^{-4m-1}$ we deduce that
$n_0 = \poly(2^m/\eps) / (1-\gamma)^2.$

\noindent The main result of this section is devoted to the proof of the following theorem:

\begin{theorem}\label{th: competitive ratio product delta-balanced distribution}
Let $T_1$ be the triangle at the root of the Chord algorithm's recursion tree, and suppose that
$n$ points are inserted into $T_1$ independently from a $\gamma$-balanced distribution, where $\gamma \in [0,1)$.
The expected performance ratio of the Chord algorithm on this instance is $O_{\gamma}(\log \log n)$.
\end{theorem}

\noindent Combined with Proposition~\ref{fact:dens-unif}, the theorem yields the desired upper bound of $O_{\gamma}(\log m + \log \log (1/\eps))$.
The proof has the same overall structure as the proof of Theorem~\ref{th: competitive ratio PPP}, but the details are more
elaborate. We emphasize below the required modifications to the argument. 

Since the performance ratio of the Chord algorithm is at most $2n$ on any instance with $n$ points, 
we will assume that $n$ is lower bounded by a sufficiently large absolute constant
($n \ge 12$ suffices for our analysis).
We start by giving an area shrinkage lemma, similar to Lemma~\ref{lem:area recursion}.
(See Figure~\ref{fig:chord-avg} for an illustration.)

\begin{lemma} \label{lem:area recursion product distn}
Let $T_i = \triangle(a_ib_ic_i)$ be a triangle processed by the Chord algorithm at
recursion depth at most $\lceil \log_2 \ln n \rceil-1$. Denote $q_i =\comb(\lambda_{a_ib_i})$. 
Let $T_{i,l} = \triangle (a_i a_i' q_i)$, $T_{i,r}= \triangle (b_ib'_iq_i)$ and $T_i' = \triangle (a_i'c_ib_i')$. 
For all $c >0$, with probability at least $1- \ln n^{-{{c (1-\gamma)^2} \over 2}}$ conditioning on the information available to the algorithm,
\begin{equation*}
S(T_{i,l}), S(T_{i,r}) \le \sqrt{S(T_i) S_1} \cdot \sqrt{{c \cdot \ln \ln n \over n}}; ~~\text{and}~~ S(T_i') \le  S_1 {c \cdot \ln \ln n \over n}.
\end{equation*}
\end{lemma}
\begin{proof}
We follow the proof of Lemma~\ref{lem:area recursion} with the appropriate modifications. Let $T_i$ be the triangle maintained
by the algorithm at some node of the recursion tree, and suppose that $T_i$ is at recursion depth at most $\lceil \log_2 \ln n
\rceil-1$. The information available to the algorithm when it processes $T_i$ (before it makes the query $\comb(\lambda_{a_ib_i})$) is:
\begin{itemize}
\item The location of all points $q_j$, $j \in [i-1]$ is known;
by our assumption on the depth it follows that $ i \leq 2 \ln n$.
\item  There is no point below the line (defined by)
$a_j c_j$, or below the line $c_j b_j$, for all $j \in [i]$.
\end{itemize}
Given this information, the probability that, among the remaining $n_i \ge n-2 \ln n$ points (whose location is unknown), none
falls inside a triangle $T$ of area $S(T)$ inside $T_i$, is at most
$$\left(1-  { (1-\gamma)^2} {S(T) \over S_1}\right)^{n_i}.$$
Indeed, let ${T}^*_i \subseteq T_1$ be the subset of the root triangle which is available for the location of
$q_i$; this is the convex set below the line $a_1b_1$, to the right of all lines $a_jc_j$, for $j \in [i]$, and above
all lines $c_jb_j$, for $j \in [i]$. The probability that a point whose location is unknown falls inside $T \subseteq
T^{\ast}_i$ is
$${{\cal D}\left[ T \right] \over {\cal D}[{T}^{\ast}_i]} \ge {(1-\gamma) {\cal U}\left[ T \right] \over {\cal U}[{T}^{\ast}_i] / (1-\gamma)} \ge  {(1-\gamma) {\cal U}\left[ T \right] \over {\cal U}[{ T_1}]/ (1-\gamma) } = {(1-\gamma)^2} {S(T) \over S_1}.$$
Choosing $S(T) \eqdef S_1 {c \cdot \ln \ln n \over n}$, the probability becomes
$${(1-\gamma)^2} \cdot {c \cdot \ln \ln n \over n}.$$
Hence, the probability that $T$ is empty is at most
\begin{align*}
\left(1-  { c (1-\gamma)^2} {\ln \ln n \over n}\right)^{n_i} \le e^{- {c (1-\gamma)^2} {\ln \ln n \over n} n_i} \le \left({1
\over \ln n}\right)^{{{c (1-\gamma)^2} \over 2}}.
\end{align*}
The proof of the lemma is concluded by identical arguments as in the proof of Lemma~\ref{lem:area recursion}.
\end{proof}
Using Lemma~\ref{lem:area recursion product distn} and the union bound, we can show that, with probability at least
$$1-{2  \over \left( \ln n \right)^{{{c (1-\gamma)^2} \over 2}-1}},$$
the following are satisfied:
\begin{itemize}
\item All triangles maintained by the algorithm at depth $\lceil \log_2 \ln n \rceil$ of its recursion tree have area at most
$$S_1 \cdot {e \cdot c \cdot \ln \ln n \over n}.$$
\item For every node (triangle) $i$ in the first $\lceil \log_2 \ln n \rceil -1$ levels of the recursion tree
\begin{align*}
S(T'_i) \le S_1{c \cdot \ln \ln n \over n},
\end{align*}
where $T'_i$ is defined as in the statement of Lemma~\ref{lem:area recursion product distn}.
\end{itemize}
The proof of the second assertion above follows immediately from Lemma~\ref{lem:area recursion product distn} and the union
bound. The first assertion is shown similarly to the analogous assertion of Theorem~\ref{th: competitive ratio PPP}. For the
above we assumed that $c \in \left({1 \over \ln \ln n}, {n  \over \ln \ln n}\right)$.

Now let us call ${\cal A}_c$ the event that the above assertions are satisfied. We can show the following.

\begin{lemma} \label{lem: competitive ratio in the good event - product}
Suppose $c \le {n \over 4 \cdot \ln n \cdot \ln \ln n}$. Conditioning on the event ${\cal A}_c$, the expected performance ratio
of the Chord algorithm is $O_{c,\gamma}\left(\log \log n\right)$.
\end{lemma}
\begin{proof}
The proof is in the same spirit to the proof of Lemma~\ref{lem: competitive ratio in the good event}, but more care is needed.
We need to argue that, under ${\cal A}_c$, the expected number of points falling inside a triangle at depth $\lceil \log_2 \ln
n \rceil$ of the recursion tree is $O_{c,\gamma}(\log \log n)$. Using rationale similar to that used in the proof of
Lemma~\ref{lem:area recursion product distn} above, we have the following: Let $T_i$ be the triangle maintained by the
algorithm at a node $i$ at depth $\lceil \log_2 \ln n \rceil$ of the recursion tree. Let also $p$ be a point whose location is
unknown to the algorithm (conditioning on the information known to the algorithm after processing the first $\lceil \log_2 \ln
n \rceil-1$ levels of the recursion tree). The probability that the point $p$ falls inside $T_i$ is
$${ {\cal D}\left[ T_i \right] \over {\cal D}[{T}^{\ast}_i] }  \le {{\cal U}\left[ T_i \right] / (1-\gamma) \over (1-\gamma) {\cal U}[{T}^{\ast}_i] },$$
where $T^{\ast}_i$ is the region below the line $a_1b_1$, above the lines $a_j c_j$, for all $j$ in the first $\lceil \log_2 \ln
n \rceil$ levels of the recursion tree, and above the lines $c_jb_j$, for all $j$ in the first $\lceil \log_2 \ln n \rceil$
levels of the recursion tree. To upper bound the probability that $p$ falls inside $T_i$, we need a lower bound on the size of
the area $S(T_i^{\ast})$. Such bound can be obtained by noticing that
$$S(T_i^*)=S_1 - \littlesum_j S(T'_j),$$
where the summation ranges over all $j$ in the first $\lceil \log_2 \ln n \rceil-1$ levels of the recursion tree. Hence,
\begin{align*}
S(T^{\ast}_i) &\ge S_1 - 2 \ln n \cdot S_1{c \cdot \ln \ln n \over n} \\&= S_1 \cdot {n - 2 \cdot c \cdot \ln n \cdot \ln \ln n
\over n} \ge {S_1 \over 2},
\end{align*}
where we used that $c \le {n \over 4 \cdot \ln n \cdot \ln \ln n}$. Hence, the probability that a point falls inside $T_i$ is
at most
$${{\cal U}\left[ T_i \right] / (1-\gamma) \over (1-\gamma) {\cal U}[{T}^{\ast}_i] } \le {1 \over (1-\gamma)^2} \cdot { S_1 \cdot {e \cdot c \cdot \ln \ln n \over n} \over S_1/2} \le  {2 \cdot e \cdot c \cdot \ln \ln n \over (1-\gamma)^2 n}.$$
Therefore, the expected number of points falling in $T_i$ is at most
$${2 \cdot e \cdot c \cdot \ln \ln n \over (1-\gamma)^2}.$$
The final part of the proof is a charging argument identical to the one in Lemma~\ref{lem: competitive ratio in the good event}.
\end{proof}

We have thus established the following:

\begin{lemma} \label{lem: good case competitive ratio - product}
For $c \in \left({1 \over \ln \ln n}, {n \over 4 \cdot \ln n \cdot \ln \ln n}\right)$, there exists an event ${\cal A}_c$, with
$\Pr[{\cal A}_c] \ge 1-{2  \over \left( \ln n \right)^{{{0.5 c (1-\gamma)^2}}-1}},$ such that the expected performance ratio of
the Chord algorithm conditioning on ${\cal A}_c$  is $O_{c,\gamma}\left( \log \log n \right)$.
\end{lemma}

We can also show the analogue of Lemma~\ref{lem: pessimistic competitive ratio}. Namely,

\begin{lemma}\label{lem: pessimistic competitive ratio - product}
For $c' \in \left({1 \over \ln n}, {\ln n  \over 6}\right)$, there exists an event ${\cal B}_{c'}$, with $\Pr[{\cal B}_{c'}]
\ge 1-{2 \over n^{{{0.5 c' (1-\gamma)^2}}-1}}$, such that the expected performance ratio of the Chord algorithm conditioning on ${\cal B}_{c'}$  is $O_{c',\gamma}\left( (\log n)^3\right)$.
\end{lemma}
\begin{proof}
The proof is similar to the proof of Lemma~\ref{lem: good case competitive ratio - product}, except that the bound is now a bit
trickier. For $c' \in \left({1 \over \ln n}, {\ln n  \over 6}\right)$, let ${\cal B}_{c'}$ be the event that
\begin{itemize}
\item all the triangles maintained by the algorithm at depth $\lceil \log_2 n \rceil$ of its recursion tree have area at most $S_1 \cdot {e \cdot c' \cdot \ln n \over n}.$
\item for every node $i$ inside the first $\lceil \log_2  n \rceil -1$ levels of the recursion tree
$$ S(T'_i) \le S_1{c' \cdot \ln n \over n},$$
where $S(T'_i)$ is defined as in the statement of Lemma~\ref{lem:area recursion product distn}.
\end{itemize}
Using arguments similar to those in the proof of Lemma~\ref{lem:area recursion product distn} and the union bound, we obtain
that
$$\Pr[{\cal B}_{c'}] \ge 1-{2 \over n^{{{0.5 c' (1-\gamma)^2}}-1}}.$$
Now let  $d^{\ast} \eqdef \lceil \log_2 n \rceil$ and 
${\cal T}_{d^{\ast}}$ be the recursion tree of the algorithm pruned at level $d^{\ast}$. We also define the set ${\cal
L}'_{d^{\ast}}$ as in the proof of Lemma~\ref{lem: competitive ratio in the good event}, (but with $d^*$ replaced by $\lceil \log_2 n
\rceil$). As in that proof, any $\eps$-CP set needs to use at least $|{\cal L}'_{d^{\ast}}|$ points. Moreover, 
the total number of nodes inside ${\cal T}_{d^{\ast}}$ is at most $2 d^{\ast} |{\cal L'}_{d^{\ast}} |.$
Whenever the algorithm processes a triangle, a planar region of area at most $S_1 \cdot { c' \cdot \ln n \over n}$ is removed
from $T_1$ (the root triangle). Therefore, after finishing the processing of the first $\lceil \log_2 n \rceil -1$ levels of the
tree, a total area of at most
$$2 \lceil \log_2 n \rceil S_1 \cdot {c' \cdot \ln n \over n} |{\cal L'}_{d^{\ast}}|$$
is removed from $T_1$. 

We distinguish two cases. If $|{\cal L'}_{d^{\ast}}| \ge {n \over \lceil \ln n\rceil^3}$, then the size of the
optimum is at least ${n \over \lceil \ln n\rceil^3}$ points. Since there is a total of $n$ points (and the algorithm never
performs more than $2 n$ $\comb$ calls), it follows that in this case the performance ratio is $O(\log^3 n)$. 

On the other hand, if $|{\cal L'}_{d^{\ast}}| \le {n \over \lceil \ln n\rceil^3}$, then the total area that has been removed from $T_1$ is at
most $2 S_1 \cdot {c' \over \ln n}.$ Hence, the remaining area is at least $S_1/2$, assuming $c' \le \ln n /6$. Given this
bound it follows that the expected number of points inside a triangle at level $\lceil \log_2 n\rceil$ of the recursion tree is at
most
$${2 \cdot e \cdot c' \cdot \ln n \over (1-\gamma)^2}.$$
Using the aforementioned and noting that the performance ratio ``paid'' within the first $\lceil \log_2 n\rceil-1$ levels of
the recursion tree is at most $O_{c',\gamma}(\log n)$, we can conclude the proof via arguments parallel to those in the proof of
Lemma~\ref{lem: competitive ratio in the good event}.
\end{proof}
Now let us choose $c ={8 \over (1-\gamma)^2 }$ and $c'={4 \over (1-\gamma)^2}$. From Lemmas~\ref{lem: good case competitive
ratio - product} and~\ref{lem: pessimistic competitive ratio - product} we have that
$$\Pr[{\cal A}_c] \ge 1- {2 \over (\ln n)^3}~~\text{and}~~\Pr[{\cal B}_{c'}] \ge 1- {2 \over n}.$$
Given this, we can conclude the proof Theorem~\ref{th: competitive ratio product delta-balanced distribution}. The argument is
the same as the end of the proof of Theorem~\ref{th: competitive ratio PPP}, except that we can now trivially bound the
performance ratio of the algorithm by $2n$ in the event ${\overline{{\cal A} \cup {\cal C}}}$. 

\subsection{Lower Bounds} \label{ssec:average-lower}

In this section we show that our upper bounds on the expected performance of the algorithm 
are tight up to constant factors. In particular, for the case of the Poisson point process we prove:

\begin{theorem}\label{thm: lower bound PPP}
Let $T_1$ be the triangle at the root of the Chord algorithm's recursion tree, and suppose that points are inserted into $T_1$
according to a Poisson Point Process with intensity $\nu$. There exists an infinite family of instances (parameterized by $\eps$ and $m$)
on which the expected performance ratio of the Chord algorithm is $\Omega(\log \log (\nu S(T_1)))$. 
In particular, we can select the parameters so that $S(T_1) = \poly(2^m/\eps)$, which
yields a lower bound of $\Omega (\log m + \log \log 1/\eps)$.
\end{theorem}


\begin{proof}
The lower bound construction is reminiscent to the worst-case instance of Section~\ref{ssec:worst-lower}.
In particular, the initial triangle $T_1 = \triangle (a_1b_1c_1)$ (at the root of the recursion tree) will be right and $(a_1c_1) >> (b_1c_1)$. 
To avoid clutter in the expressions, we present the proof for the case $m=1$. The generalization for all values of $m$ is straightforward. 

We fix $a_1 = (1, 2)$, $b_1 = (1+2\eps, 1)$ and $c_1 = (1,1)$ and select the intensity of the Poisson process to be $\nu \eqdef 1/\eps^2$. 
Note that for this setting of the parameters we have that $\nu S(T_1) =1/\eps$, we thus obtain an $\Omega (\log \log (1/\eps))$ lower bound.

Given the endpoints $a_1, b_1$, it is clear that $\opt_{\eps} \leq 3$ with probability $1$. Hence, it suffices to show that 
the Chord algorithm makes $\Omega(\log \log (1/ \eps))$ $\comb$ calls in expectation before it terminates.
To show this, we are in turn going to prove that for $\eps$ being a sufficiently small positive constant, with constant probability, 
there exists a path ${\cal P}$ of length $\Omega(\log \log (1/\eps))$ in the recursion tree of the algorithm. 

As shown in Lemma~\ref{lem:hd-vs-rd},
for such instances the ratio distance is very well approximated by the horizontal distance. In particular, consider the triangle $T_i = \triangle(a_ib_ic_i)$ 
(see Figure~\ref{fig:chord-avg-lb}), where $(b_ic_i) \leq 2\eps$. For any point $s \in T_i$, we have that 
$$\rd(s, a_ib_i) \leq \hd(s, a_ib_i) \leq \rd(s, a_ib_i) + 4\eps^2 + 2\eps / \lambda_{a_ib_i}.$$
Hence, as long as $\lambda_{a_ib_i}$ (the slope of $a_ib_i$) is sufficiently large, we can essentially use the horizontal distance metric
as a proxy for the ratio distance. Indeed, this will be the case for our lower bound instance below.

We now proceed to formally describe the construction.
The path ${\cal P}$ of length $\Omega(\log \log (1/\eps))$ will be defined by the triangles with $b_i=b_1$, 
i.e., the ones corresponding to the right subproblem in every subdivision performed by the algorithm. 
For notational convenience, we shift the coordinate system to the point $c_1$, so that $c_1=(0,0)$, $a_1=(0,1)$ and
$b_1=(2\epsilon,0)$. (Note that the horizontal distance is invariant under translation.) We label the
triangles in the path $\cal P$ by $T_1$, $T_2$, $\ldots$, and we let the vertices of triangle $T_i$ be $a_i=(x_i, y_i)$,
$c_i=(x_i',0)$, and $b_i=b_1$. Suppose that when the Chord algorithm processes the triangle $T_i$, the $\comb$ routine returns
the point $q_i$ on a line $a_i'b_i'$ parallel to $a_ib_i$ (as in Figure~\ref{fig:chord-avg-lb}). Note that $b_i' = c_{i+1}$ and $q_i =a_{i+1}$. 
Let  $a_i' = ({\cal X}_i, {\cal Y}_i) $.
\begin{figure}[h!]
\begin{flushleft}
\epsfig{file=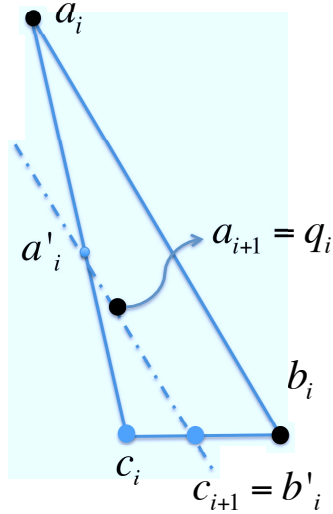,angle=90, width=15cm}
\end{flushleft}
\vspace{-3cm}
\caption{Illustration of the average case lower bound construction.} \label{fig:chord-avg-lb}
\end{figure}

\noindent Let $j^{\ast} \eqdef (1/2) \lceil \log \log (1/\eps) \rceil$. The theorem follows easily from the next lemma:

\begin{lemma} \label{lem: long path}
Let $\eps$ be a sufficiently small positive constant. With probability at least $1-4 j^{\ast}/(\ln(1/\eps))^{1/4}$,
for all $i \in [j^{\ast}]$,
\begin{equation} \label{eqn:rec-bound}
y_i \ge c_i' \cdot {\epsilon}^{1-2^{-i+1}} / {\ln (1/\epsilon)}; \quad \quad \textrm{and} \quad \quad x_i' \le c_{i}'' \cdot \epsilon^{1+2^{-i+1}} \cdot \ln (1/\epsilon),
\end{equation}
where $c_i' = 2^{1 - 2^{-i}}$ and $c_i''= \sqrt{2} \cdot \littlesum_{j=1}^{i} 2^{2^{-j}}$.
\end{lemma}

\begin{proof}
It is clear that (\ref{eqn:rec-bound}) is satisfied for $i=1$.
We are going to show by induction that, if (\ref{eqn:rec-bound}) holds for some $i$, it also holds for $i+1$ with
probability at least $1-4 / (\ln(1/\eps))^{1/4}$. The theorem then follows by a union bound over all $i \in [j^{\ast}]$.

By the similarity of the triangles $T_i = \triangle (a_ib_ic_i)$ and $T'_i = \triangle (a'_ib'_ic_i)$  we have
\begin{align}
{y_i \over {\cal Y}_{i}} = {2 \epsilon-x_i' \over x_{i+1}'-x_i'}. \label{eq: similar}
\end{align}
From the properties of the Poisson Point Process we have the following: conditioning on the information available to the
algorithm when it processes the triangle $T_i$, if a measurable region inside $T_i$ has area $S^{\ast} = \epsilon^2 \ln (1/\epsilon)$,
then the number of points inside this region follows a Poisson distribution with parameter $\nu \cdot S^{\ast} = \ln (1/\epsilon)$. 
Hence, with probability at least $1- \epsilon$, any such region contains at least one point. Hence, with probability at least $1-\eps$ we have
that $S(T'_i) \leq S^{\ast}$. Note that $S(T'_i) = (1/2)(x'_{i+1} - x'_i) {\cal Y}_i$.
Using~\eqref{eq: similar} and the induction hypothesis, this implies that, with probability at least $1-\eps$,
$$x_{i+1}' - x_i' \le {2 \over \sqrt{c_i'}} \cdot \ln (1/\eps) \cdot \epsilon^{1+2^{-i}};$$
hence
$$x_{i+1}' \le \left( {2 \over \sqrt{c_i'}} + c_i''\right) \cdot \ln(1/\eps) \cdot \epsilon^{1+2^{-i}}$$
as desired.

On the other hand, if the area of a region is no more than $S^{\ast\ast} = \epsilon^2 / \sqrt{\ln (1/\epsilon)}$ the probability that a point is
contained in that region is at most $\nu \cdot S^{\ast\ast} = 1/\sqrt{\ln {(1/\epsilon)}}$. Similarly, this implies that, with probability at least 
$1-1/ \sqrt{\ln{(1/\epsilon)}}$,
$${\cal Y}_{i}\ge {\sqrt{c_i'} \cdot \epsilon^{1- 2^{-i}} \over (\ln (1/\eps))^{3/4}}.$$
By the properties of the Poisson Point Process it follows that the point $q_i$ is uniformly distributed on the segment
$a_i'b_i'$. Hence, with probability at least $1- {\sqrt{2} \over (\ln (1/\eps))^{1/4}}$, it holds $$y_{i+1} \ge {\sqrt{2
c_i'} \cdot \epsilon^{1- 2^{-i}} \over \ln (1/\eps)}.$$ A union bound concludes the proof.
\end{proof}

We now show how the theorem follows from the above lemma.
First note that, for all $i$ it holds
$\sqrt{2} \cdot \littlesum_{j=1}^{i} 2^{2^{-j}} \le 2 \sqrt{2} \cdot i.$
Hence, by Lemma~\ref{lem: long path} and the choice of $j^{\ast}$, it is easy to check that with probability at least $1-4 j^{\ast}/
(\ln (1/\eps))^{1/4}$, we have $x'_{j^{\ast}}  = o(\eps)$ and $y_{j^{\ast}}  = \omega(\eps).$ The latter condition implies that the horizontal
distance is a very good approximation to the ratio distance. The latter, combined with the first condition, implies that the node (corresponding to the triangle)
$T_{j^{\ast}}$ is not a leaf of the recursion tree. That is, all the triangles $T_1, \ldots, T_{j^{\ast}}$ survive in the
recursion tree, since the Chord algorithm does not have a certificate that the points already computed are enough to
form an $\eps$-CP set. This concludes the proof of the theorem.
\end{proof}

An analogous result can be shown similarly for the case of $n$ points drawn from a balanced distribution.


\section{Conclusions and Open Problems} \label{sec:concl}
We studied the Chord algorithm, a simple popular greedy algorithm that is used (under different names) for the approximation of
convex curves in various areas. We analyzed the performance ratio of the algorithm, i.e., the ratio of the cost of the algorithm
over the minimum possible cost required to achieve a desired accuracy for an instance, with respect to the Hausdorff and the
ratio distance. We showed sharp upper and lower bounds, both in a worst case and in an average setting. In the worst case the
Chord algorithm is roughly at most a logarithmic factor away from optimal, while in the average case it is at most a doubly
logarithmic factor away.

We showed also that no algorithm can achieve a constant ratio in the worst-case, in particular,
at least a doubly logarithmic factor is unavoidable. We leave as an interesting open problem
to determine if there is an algorithm with a better performance than the Chord algorithm (both in the worst-case and in average case settings),
and to determine what is the best ratio that can be achieved.
Another interesting direction of further research is to analyze the performance of the
Chord algorithm in three and  higher dimensions, and to characterize what is
the best performance ratios that can be achieved by any algorithm.


\end{document}